\documentclass[a4paper,UKenglish,cleveref,autoref,thm-restate]{lipics-v2021}

\usepackage[T1]{fontenc}

\usepackage{standalone}

\usepackage{amsfonts}
\usepackage{amsmath}
\usepackage{amsthm}
\usepackage{hyperref}
\usepackage{thmtools}
\usepackage{mathtools}

\usepackage{cmll}
\usepackage{virginialake}
\odbackgroundtrue
\odframefalse
\odbackgroundfirstfalse
\odframefirstfalse
\vlnostructuressyntax

\usepackage{proofzilla}
\usepackage{xspace}
\usepackage{adjustbox}
\usepackage{twoopt}

\usepackage[bbgreekl]{mathbbol} 

\usepackage{amsfonts}

\usepackage{graphicx}

\def\lflag{
	\begin{tikzpicture}[x=height("H"),y=height("H")]
		\draw[thick] (0,0) -- (0,1);
		\draw[fill=black] (0,1) -- (-.4,.8) -- (0,.6) -- cycle;
	\end{tikzpicture}\,
}

\def\N{\mathbb N}
\def\set#1{\{#1\}}
\def\Set#1{\left\{#1\right\}}
\def\SET#1{\begin{Bmatrix}#1\end{Bmatrix}}

\def\tuple#1{\langle#1\rangle}
\usepackage{scalerel}

\def\intset#1#2{\set{#1,\ldots,#2}}
\def\sizeof#1{|#1|}

\renewcommand{\emptyset}{\varnothing}

\def\CCS{\textsf{CCS}\xspace}
\def\picalc{$\pi$-calculus\xspace}

\def\defn#1{{\itshape\bfseries\boldmath #1}}

\def\rclr#1{{\color{red}#1}}
\def\bclr#1{{\color{blue}#1}}
\def\sclr#1{{\color{magenta}#1}}
\def\gclr#1{{\color{pzgreen}#1}}
\def\vclr#1{{\color{violet}#1}}

\def\resp#1{(resp.~{#1})\xspace}

\newcommand{\qwith}{\quad\mbox{with}\quad}

\newcommand{\quand}{\quad\mbox{and}\quad}
\newcommand{\qquand}{\qquad\mbox{and}\qquad}

\newcommand{\myparagraph}[2][.]{\textbf{#2#1}}

\newtheorem{nota}[theorem]{Notation}

\def\proc{process\xspace}
\def\procs{processes\xspace}

\def\conet{conflict-net\xspace}

\newcommand{\varSet}{\mathcal V}

\def\cneg#1{{#1}^\lbot}

\def\free{\mathsf{free}}
\def\freeof#1{\free(#1)}

\newcommand{\swn}[1][]{\wn_{#1}}
\newcommand{\soc}[1][]{\oc_{#1}}

\def\lsend#1#2{\langle#1\oc#2\rangle}
\def\lrecv#1#2{(#1\wn#2)}

\def\lNu#1#2{\lnewsymb{#1}.#2}
\def\lWe#1#2{\lwensymb{#1}.#2}
\def\lYa#1#2{\lyasymb{#1}.#2}
\def\lFa#1#2{\forall{#1}.#2}
\def\lEx#1#2{\exists{#1}.#2}
\newcommand{\lQu}[3][]{\lqusymb_{#1}{#2}.#3}
\def\lnQu#1#2{\cneg{\lqusymb}{#1}.#2}
\def\lNa#1#2{\nabla#1.#2}
\def\lnNa#1#2{\cneg\nabla#1.#2}
\def\lNai#1#2#3{\nabla_{#1}#2.#3}

\def\lQup#1#2{\lQu{#1}{\!\left( #2 \right)}}
\def\lNup#1#2{\lNu{#1}{\!\left( #2 \right)}}
\def\lWep#1#2{\lWe{#1}{\!\left( #2 \right)}}
\def\lYap#1#2{\lYa{#1}{\!\left( #2 \right)}}
\def\lFap#1#2{\lFa{#1}{\!\left( #2 \right)}}
\def\lExp#1#2{\lEx{#1}{\!\left( #2 \right)}}

\newcommand{\lcofix}[2][X]{\bbnu #1.#2}

\def\lnu{\lcofix}

\newcommand{\biglbra}[1][]{\bigwith\limits_{#1}}
\newcommand{\biglsel}[1][]{\bigoplus\limits_{#1}}

\def\feq{\lleq}

\def\chole{\bullet}
\newcommand{\ctx}[1][\chole]{[#1]}

\newcommand{\fof}[1]{\left\llbracket#1\right\rrbracket}
\newcommand{\sof}[1]{\left\lfloor\!\left\lfloor #1 \right\rfloor\!\right\rfloor}

\newcommand{\cnof}[1]{\left\{\!\left\{#1\right\}\!\right\}_{\mathsf{co}}}
\newcommand{\snof}[1]{\left\{\!\left\{#1\right\}\!\right\}_{\mathsf{sl}}}

\def\trbv#1{\left\lceil#1\right\rceil}

\def\alphaeq{=_\alpha}

\def\namesset{\mathcal N}

\def\varset{\mathcal V}
\def\xX{\mathcal X}
\def\labelsset{\mathcal L}

\def\lab{\ell}

\def\pnu#1{(\nu #1)}
\def\pnup#1#2{(\nu #1) \left(#2\right)}
\def\pnus#1{(\nu \widetilde{#1})}

\def\psend#1#2{#1!\langle#2\rangle.}
\def\precv#1#2{#1?(#2).}
\def\psenda#1#2{#1!\langle#2\rangle}
\def\precva#1#2{#1?(#2)}

\def\ppar{\;|\;}
\def\pnil{\mathsf{Nil}}
\def\pbra#1#2#3{#1 \lseq \Set{#2 : #3 }}
\def\pbras#1#2#3{#1 \lseq \Set{#2 }_{#3}}

\def\pbrasv#1#2{#1 \lseq \begin{Bmatrix}#2\end{Bmatrix}}

\def\psels#1#2#3{#1 \lcoseq \Set{#2 }_{#3}}

\def\pselsv#1#2{#1 \lcoseq \begin{Bmatrix}#2\end{Bmatrix}}

\def\fsubst#1#2{[#1/#2 ]}
\def\fsubstup#1#2{[#1 \uparrow #2 ]}

\def\fsubsts#1{[ #1 ]}
\def\fsubminus#1{\setminus\set{#1}}

\def\steq{\equiv}

\def\cC{\mathcal C}
\def\cN{\mathcal N}

\def\redsem{\to}
\def\redsems{\twoheadrightarrow}

\def\rscomr{\mathsf{Com}}
\def\rsbrar{\mathsf{Bra}}
\def\rsselr{\mathsf{Sel}}

\def\rsresr{\mathsf{Res}}
\def\rsparr{\mathsf{Par}}
\def\rsstrr{\mathsf{Struc}}

\def\seqEval#1#2{e \downarrow_\ldia v}

\newcommand{\proves}[1][]{\mathord{\vdash_{#1}\,}}
\def\dD{\mathcal D}

\def\pomset{\mathsf{Pomset}}
\def\BV{\mathsf{BV}}
\def\BVq{\mathsf{BV}^1}
\def\MAV{\mathsf{MAV}}
\def\MAVq{\mathsf{MAV}^1}

\def\MLL{\mathsf{MLL}}

\def\MLLx{\mathsf{MLL}^\lunit}
\def\MALL{\mathsf{MALL}}

\def\muMALL{\mu\MALL}

\mathchardef\mhyphen="2D
\newcommand{\NML}{\mathsf{NML}}
\newcommand{\NMAL}{\mathsf{NMAL}}

\def\wcut{\cup\set{\cutr}}

\def\PIL{\mathsf{PiL}}
\def\PILm{\PIL^-}
\def\miniPIL{\mathsf{mini}\mhyphen\PIL}

\def\NL{\PIL}

\def\XS{\mathsf{X}}

\newcommand{\sS}[1][\mathcal S]{\sclr{ #1 }}
\newcommand{\emptystore}{\sclr{ \emptyset }}
\newcommand{\sdash}[1][\sS]{\sS[{#1}] \vdash}

\def\axrule{\mathsf {ax}}
\def\AXrule{\mathsf {AX}}
\def\cutr{\mathsf {cut}}

\def\mixr{\mathsf{mix}}
\def\precur{\lprec_\lunit}
\newcommand{\rrule}[1][]{{\mathsf{r}^{#1}}}
\newcommand{\brrule}[1][]{\mathsf{\beta}_{#1}}

\newcommand{\urrule}[1][]{\mathsf{\alpha}_{#1}}
\def\urule{\lunit}

\def\nuurule{\lnewsymb^\lunit}
\def\yurule{\lyasymb^\lunit}

\def\Qrule{\lqusymb}
\def\dQrule{\cneg\lqusymb}

\def\loadr#1{#1_\mathsf{load}}
\def\popr#1{#1_\mathsf{pop}}
\def\unitr#1{#1_{\lunit}}

\def\naur{\unitr\nabla}
\def\naloadr{\loadr\nabla}
\def\napopr{\popr\nabla}
\def\nnaur{\unitr{\cneg\nabla}}

\def\nnapopr{\popr{\cneg\nabla}}

\def\nqur{\naur}
\def\nqloadr{\naloadr}
\def\nqpopr{\napopr}

\def\nuur{\unitr\lnewsymb}
\def\nuloadr{\loadr\lnewsymb}
\def\nupopr{\popr\lnewsymb}
\def\yaur{\unitr\lyasymb}
\def\yaloadr{\loadr\lyasymb}
\def\yapopr{\popr\lyasymb}

\def\isnusymb{\lnewsymb}
\def\isyasymb{\lyasymb}
\def\isnasymb{\nabla}
\def\isdnasymb{\cneg\nabla}
\def\isnu#1{\sclr{#1^{\isnusymb}}}
\def\isya#1{\sclr{#1^{\isyasymb}}}

\def\isqu#1{\sclr{#1^{\lqusymb}}}
\newcommand{\isna}[2][]{\sclr{#2^{\isnasymb_{#1}}}}
\newcommand{\isnna}[2][]{\sclr{#2^{\isdnasymb_{#1}}}}

\def\nqsrule{\lnewsymb\mhyphen\lyasymb}
\def\nqscutr{\sS\mhyphen\cutr}

\def\Mrule{\mathsf{M}}
\def\Krule{\mathsf{K}}

\def\airule{\mathsf{ai}}

\def\swir{\mathsf{s}}
\def\qrule{\mathsf{q}}

\newcommand{\scope}[1][]{\mathsf{scope}_{#1}}
\def\shiftr{\mathsf{shift}}

\def\aidr{\airule\mathord{\downarrow}}

\def\qdr{\qrule\mathord{\downarrow}}

\def\ldist#1#2{\mathsf{d}_{#2}\left(#1\right)}

\def\lcomp#1#2{\mathsf{c}_{#2}\left( #1\right)}

\def\lcutelim{\rightsquigarrow}

\def\cutstep#1#2{#1\mbox{-vs-}#2}

\def\conc{\frown}
\def\conf{\#}

\def\pzlink#1#2#3#4#5#6{
	\tikz[overlay,remember picture,draw,fill,opacity=1]{
		\foreach \ppp in #6{
			\draw[thick,color=#5] (\ppp) -- ++(0,#3pt) -| (#1);
		}
		\draw[thick,color=#5] (#1) -- ++(0,#3pt) -| (#2) node[pos=.25,fill=white,inner sep=1pt]{\scriptsize$#4$};
	}
}
\def\pzlinks#1{
	\foreach \aaa/\bbb/\ccc/\ddd/\eee/\fff in {#1}{
		\pzlink{\aaa}{\bbb}{\ccc}{\ddd}{\eee}{\fff}
	}
}

\def\pzflow#1#2#3#4{
	\tikz[overlay,remember picture,draw,fill,opacity=1]{
		\draw[thick,color=#4] (#1) -- ++(0,#3pt) -| (#2);
	}
}

\def\pzflows#1{
	\foreach \aaa/\bbb/\ccc/\eee in {#1}{
		\pzflow{\aaa}{\bbb}{\ccc}{\eee}
	}
}

\defedgetype{flow}{draw,thick,color = magenta}{}

\pzvertex{conf}{\conf}{}
\pzvertex{conc}{\conc}{}
\pzvertex{Nu}{\lnewsymb}{}
\pzvertex{Na}{\nabla}{}
\pzvertex{nNa}{\cneg\nabla}{}
\pzvertex{Ya}{\lyasymb}{}
\pzvertex{Qu}{\lqusymb}{}

\newcommand{\link}[2][black]{{\color{#1}{#2}}}
\newcommand{\la}[1][red]{{\link[#1]a}}
\newcommand{\lb}[1][blue]{{\link[#1]b}}
\newcommand{\lc}[1][violet]{{\link[#1]c}}
\newcommand{\lcd}[1][black]{\link[#1]{
c}}
\newcommand{\labb}[1][black]{\link[#1]{
		ab
		}}
\newcommand{\laa}[1][black]{\link[#1]{
a'}}
\newcommand{\laaa}[1][black]{\link[#1]{
a'}}
\newcommand{\lbb}[1][black]{\link[#1]{
b'}}
\newcommand{\lbbb}[1][black]{\link[#1]{
b'}}

\newcommand{\lccd}[1][black]{\link[#1]{
c'}}
\newcommand{\lcccd}[1][black]{\link[#1]{
c'}}
\newcommand{\laabb}[1][black]{\link[#1]{
a'b}}
\newcommand{\laaabb}[1][black]{\link[#1]{
a'b}}
\newcommand{\lac}[1][black]{\link[#1]{
ac}}
\newcommand{\labc}[1][black]{\link[#1]{
abc}}

\newcommand{\coalto}[1][]{\rightarrow_{#1}}

\newcommand{\lcoalto}[2][]{\overset{#2}{\rightarrow}}

\def\coaldo{\downarrow}

\def\cotree{co-tree\xspace}
\def\cotrees{co-trees\xspace}
\def\linking{\Lambda}
\def\linktree{\tau(\Lambda)}

\def\domof#1{\mathsf{dom}(#1)}

\newcommand{\dualizerof}[1][]{\delta_{#1}}

\def\peq{\approx}
\def\speq{\simeq}
\def\weq{\sim_{\mathsf w}}
\def\pweq{\peq_{\mathsf w}}
\def\spweq{\speq_{\mathsf w}}
\def\teq{\sim}

\def\cohe{\mathsf{coh}}
\def\join{\vee}

\def\duasum{+}
\def\disj{\mathsf{dis}}

\def\newdisj#1#2{\mathsf{dis}(#1,#2)}
\def\newcohe#1#2{\mathsf{coh}(#1,#2)}

\title{Proofs as Execution Trees for the $\pi$-Calculus}

\author{Matteo Acclavio}{University of Sussex \and \url{https://matteoacclavio.com} }{}{https://orcid.org/0000-0002-0425-2825}{partially supported by the Villum Fonden grant no. 50079.}
\author{Giulia Manara}{University of Southern Denmark, Universitè Paris Cité, Università Roma Tre}{}{https://orcid.org/0009-0003-9583-1017}{co-funded by the European Union's Horizon 2020 research and innovation program under the Marie Sklodowska-Curie grant agreement No 945332, and by the European Union (ERC, CHORDS, 101124225).
}
\authorrunning{M. Acclavio and G. Manara} 
\Copyright{Matteo Acclavio and Giulia Manara} 




\acknowledgements{Thanks to Giovanni Bernardi and Fabrizio Montesi for the insightful discussion about the process calculi (see \cite{acc:man:mon:FaP})}

\nolinenumbers 

\hideLIPIcs  

\begin{document}

\maketitle
\begin{abstract}
	In this paper, we establish the foundations of a novel logical framework for the \picalc, based on the \emph{deduction-as-computation} paradigm.
	Following the standard proof-theoretic interpretation of logic programming,
	we represent processes as formulas, and we interpret proofs as computations. To be precise, we interpret proofs as execution trees.

	For this purpose, we define a cut-free sequent calculus for an extension of first-order multiplicative and additive linear logic.
	This extension includes a non-commutative and non-associative connective to faithfully model the prefix operator, and nominal quantifiers to represent name restriction.
	Finally, we design proof nets providing canonical representatives of derivations up to local rule permutations.
\end{abstract}

\section{Introduction}

Formal reasoning about the properties of concurrent program executions is significantly
more complex than analyzing sequential programs.
The main challenge in the concurrent setting
arises from the lack of formalisms for efficient representations
of the set of traces of a program in the presence of \emph{interleaving} concurrency, where the mutual order of certain tasks of a program is irrelevant.
This is due to the inherent limitations of languages commonly used to represent trace
reasoning, including natural language, where it can be impossible
to describe a set of events arranged in complex patterns in a canonical way,
other than by inefficiently listing all possible total orders.
A language for optimizing the trace analysis of concurrent programs should:
\begin{enumerate}
	\item\label{des:1} ignore irrelevant differences, such as the mutual order of independent events;
	\item\label{des:2} group traces that differ only in branching caused by internal choices within the program;
	\item\label{des:3} distinguish sets of traces that differ due to factors beyond the control of the program, such as race conditions and side effects.
\end{enumerate}
In this work, we develop a formalism satisfying these three desiderata to represent execution trees of processes of the \picalc in a canonical way, based on the \emph{deduction-as-computation} interpretation for proofs a logic extending first-order multiplicative additive linear logic.

\myparagraph{An approach inspired by logic programming}
In the \emph{logic programming} paradigm,
programs are interpreted as sets of formulas,
and computation is performed by applying methods (or rules) to these formulas.
In \cite{mil:uniform} Miller et al illustrate how a \emph{deduction-as-computation} interpretation of proof search in the sequent
calculus allows to account for program executions:
sequents correspond to snapshots of the state of the system,
and sequent rules can be interpreted as methods executing the instructions encoded
by logical connectives.

In the setting of deduction-as-computation, two forms of non-determinism appear in program executions that are not as easily observable in
other frameworks%
\footnote{
	In particular, within the proofs-as-processes setting that arises from the Curry-Howard isomorphism, the non-determinism resulting from an internal choice is not observable in the computations of typed processes, as the type of a process predetermines the choice.
}:
the \emph{don't care} non-determinism, depending on the possibility of applying
rules to independent sets of formulas, and the \emph{don't know} non-determinism, that arises from the possibility of applying (potentially different) rules to overlapping subsets of formulas.
In the proof-search interpretation of program execution, differences in derivations caused by don't care non-determinism are considered irrelevant, at the point that two proofs which can be transformed into one another through \emph{rule permutation}
(i.e., by {exchanging} the order of rules operating on disjoint sub-sequents) are usually identified.

This work aims to apply results in the study of proof equivalence \cite{hughes:invar,hughes:pws,hei:hug:conflict,strassburger:problem} in the framework of deduction-as-computation to provide canonical representations of sets of traces.
In particular, we develop a syntax based on results about canonical representation of proofs to uniquely model a set of traces differing in the order of independent events, in compliance with desiderata~\ref{des:1}.
We focus on a deduction-as-computation interpretation of \emph{proof nets} rather than sequent calculus derivation.

\myparagraph[?]{Why proof nets}
Various works \cite{and:maz:concPN,andreoli:focPN,acc:mai:DCM} have already highlighted the benefits of this approach where the syntax captures interleaving concurrency by default.
Proof nets were introduced as a graphical formalism for \emph{linear logic} proofs \cite{gir:ll}.
They abstract away irrelevant information contained in sequent calculus derivations, such as mutual order of independent inference rules.
This syntax allows for an optimal level of abstraction for the multiplicative fragment of linear logic ($\MLL$), providing  canonical representatives for proofs with respect to independent rule permutations,  a polynomial proof translation, and  a geometrical correctness criterion allowing to check in polynomial time if a graph is the encoding of a proof (making proof nets for $\MLL$ a proof system in the sense of \cite{cook:reckhow:79}).
However, the definition of proof nets for extensions of $\MLL$ requires trade-offs
between canonicity, the efficiency of correctness criterion and the efficiency of normalization
procedure (see \cite{heijltjes:houston:14} for $\MLL$ with units,
\cite{hug:van:sliceLICS,hug:van:slice,hughes:conflict,hei:hug:conflict} for multiplicative-additive linear logic
($\MALL$) and \cite{acc:EHPN} for multiplicative-exponential linear logic).

To provide an intuition of our approach, we show for the \proc $P$ in
\Cref{eq:introEx1} how we can annotate information about communications (and selections) performed during all the possible executions of $P$ while ignoring inessential details such as the specific order of independent transitions.
\begin{equation}\label{eq:introEx1}
	\hfill
	\adjustbox{max width=.96\textwidth}{$
	\begin{array}{c}
		P=\pnu{\vx1}\pnup {\vy1}{
			\vpz a{\psenda ya}
			\ppar
			\vpz {na}{\precva ya}
			\ppar
			\pselsv x{
				\vpz{1}{\lab_1}: \vpz{nb}{\precva xb},
				\vpz{3}{\lab_2}: \vpz{c}{\psenda xc}
			}
			\ppar
			\pbrasv x{
				\vpz{2}{\lab_1}: \vpz{b}{\psenda xb},
				\vpz{4}{\lab_2}: \vpz{nc}{\precva xc}
			}
		}
		\pzlinks{pza/pzna/12/\la/red/}
		\pzlinks{pz1/pz2/12/{\gclr{\lab_1}}/pzgreen/}
		\pzlinks{pz3/pz4/-16/{\gclr{\lab_2}}/pzgreen/}
		\pzlinks{pzb/pznb/16/\lb/blue/}
		\pzlinks{pzc/pznc/-12/\lc/violet/}
	\end{array}
	$}
	\hfill
\end{equation}
Note that the only (don't know) non-determinism during executions of $P$ is caused by an internal choice: the branching due to the choice of a label in  $L=\set{\lab_1,\lab_2}$. Thus, we have a unique proof net according to the desiderata \ref{des:1} and \ref{des:2}.
In this representation the set of links $\set{\gclr{\lab_1},\lb}$ and $\set{\gclr{\lab_2},\lc}$ are mutually exclusive -- in the terminology of event structures \cite{winskel:event}, we would say they are in conflict relation.

At the same time, our syntax still allows us to distinguish the two distinct executions of the \proc in \Cref{eq:introEx2}, as specified by desideratum \ref{des:3}, which are determined by a race condition on $x$ enforcing a (don't know) non-deterministic choice during the execution.
\begin{equation}\label{eq:introEx2}
	\adjustbox{max width=.96\textwidth}{$
	\pnup x{\vpz1{\psend xa} \vpz2{\psenda xb} \ppar \vpz3{\precva xy} \ppar \vpz4{\precva xz}}
	\pzlinks{pz1/pz3/12/\la/red/}
	\pzlinks{pz2/pz4/-12/\lb/blue/}
	\qquand
	\pnup x{\vpz1{\psend xa} \vpz2{\psenda xb} \ppar \vpz3{\precva xy} \ppar \vpz4{\precva xz}}
	\pzlinks{pz1/pz4/12/\la/red/}
	\pzlinks{pz2/pz3/-12/\lb/blue/}
	$}
\end{equation}

\myparagraph[?]{Which proof nets}
%
There are different syntaxes for proof nets for $\MALL$, each providing a different level of abstraction and capturing different subsets of the rule permutations in the sequent calculus -- rules permutations are reported in \Cref{fig:permutations1}.
The first syntax of proof nets for $\MALL$ was introduced by Girard in \cite{gir:ll}.
These are referred to as \emph{box nets} because of the way they encode the rules for the additive conjunction $\lwith$. These proof nets provide polynomial time proof translation and correctness criterion. However, they are not canonical with respect to rule permutations involving the rule for $\lwith$.
In \cite{girard:96:PN}, Girard also introduced \emph{monomial nets} for $\MALL$, lacking of a polynomial time correctness criterion (but a polynomial-time proof translation), but unable to not improve the level of abstraction in the general setting.
For this reason, Hughes and Van Glabbeek introduced in \cite{hug:van:sliceLICS,hug:van:slice} a new syntax of proof nets for $\MALL$ called \emph{slice nets}.
These proof nets capture all rule permutations in \Cref{fig:permutations1} while keeping a polynomial correctness criterion, but they lack of a polynomial proof translation.
The absence of a polynomial proof translation is unsurprising because each slice net can be conceived as a canonical representative of a class of derivations, which may include derivations whose sizes differ by an exponential factor.
This is due to the fact that the rule permutation between the two-premises multiplicative rule $\ltens$ and the two-premises additive rule $\lwith$-rules requires to duplicate an entire subtree of the derivation (see \Cref{fig:permutations1}).
In order to recover a polynomial proof translation, Hughes and Heijltjes introduced \emph{conflict nets} in \cite{hughes:conflict,hei:hug:conflict}, which can capture only \emph{local} rule permutations -- that is, all rule permutations except the one between the $\lwith$ and the $\ltens$.
This is obtained by having trees of axiom links (instead of a set of sets of axiom links as in slice nets) where axiom links are in a ``multiplicative'' concordance relation, or in a ``additive'' conflict relation.

At the same time, there are two main approaches for proof nets for first-order logic.
One is to consider the choice of witnesses for the existential quantifier (and, in our case, for the nominal quantifiers) as part of the information of a proof. This leads to the notion of what we refer to as \emph{witness nets} (as in \cite{hei:hug:str:ALL1}) developed in \cite{gir:ll,gir:quant1,gir:quant2}, which in \cite[Chapter 11]{gir:blind} Girard claims to be ``the only really satisfactory extension of proof-nets''.
However, how explained by Hughes in \cite{hughes:firstorder}, having the witnesses being part of the proof leads to undesirable consequences such as the lack of canonical proofs due to the existence of infinitely many possible witnesses (possibly of exponentially larger size); for an example, consider the formula $\lEx P(x) \limp \lEx P(x)$ which has infinitely many proofs, one for each possible witness for $x$.
Moreover, the complexity of cut-eliminination becomes exponential, and requiring non-local rewriting rules even in first-order $\MLL$.
For this reason, in  \cite{hughes:firstorder} Hughes introduced another approach for the design of proof nets for first-order logic, which abstracts away the choice of witnesses for the quantifiers, thus satisfying the principle of \emph{generality} (in the sense of Lambek \cite{lambek:deductive}) by identifying those proofs differing in the witness assignment.
These proof nets, called \emph{unification nets}, were initially developed for first-order $\MLL$, and have been recently extended for the purely additive fragment of linear logic in by Hughes, Heijltjes and Stra\ss burger in \cite{hei:hug:str:ALL1}).

In this paper we develop both conflict and slice (unification) nets for $\NL$ to provide canonical representatives of execution trees for the \picalc modulo different notions of interleaving:
\begin{itemize}
	\item Conflict nets (\Cref{subsec:CN}) abstract away the order of independent transitions, as soon as this order does not interact with branching of the execution tree.

	\item Slice nets (\Cref{subsec:SN}) abstract away the order of any independent transitions in a process, even when the execution tree branches.
\end{itemize}
In choreographic programming \cite{montesi:book}, the distinction between these two flavors of interleaving can be easily understood in terms of restriction on the possibility of performing out-of-order instruction in the operational semantics:
while slice nets identify execution trees of the standard operational semantics, where both communications and choices can be delayed, conflict nets allow to only delay communications, and not choices.

\myparagraph{Contributions of the Paper}
We develop a \emph{new} logical framework ($\NL$) based on an extension of first-order
multiplicative and additive linear logic ($\MALL^1$) with a
\emph{non-commutative non-associative connective} and \emph{nominal quantifiers},
to provide logical operators that faithfully model the high-level search instruction corresponding to the prefix composition and restriction in the \picalc.
We define a cut-free sequent calculus in which execution trees of a process can be interpreted as derivations of the corresponding formula.

We also define two syntaxes of proof nets for $\NL$ by combining the techniques used in \emph{unification nets} for first-order multiplicative linear logic \cite{hug:unification} and first-order additive linear logic \cite{hei:hug:str:ALL1},
with the techniques used in \emph{conflict nets} \cite{hughes:conflict,heijltjes:houston:14} and in \emph{slice nets} \cite{hug:van:slice} for $\MALL$.
As a restriction of our syntax, we then define unification conflict nets and unification slice nets for first-order $\MALL$.

Finally, by combining the correspondence between execution trees and derivations, and between derivations modulo local rule permutations and proof nets, we provide a syntax in which we have canonical representatives of execution trees modulo interleaving.

\begin{figure}[t]
	\adjustbox{max width=\textwidth}{$
		\vlderivation{
			\vlde{}{\set{\nuurule,\lpar}}{
				\lNu{\vx1}{
					\lNup {\vy1}{
						\rclr{\lsend ya}
						\lpar
						\lEx a{\rclr{\lrecv ya}}
						\lpar
						\begin{pmatrix}
							\gclr{\lsend x{\lab_1}} \lprec \lEx b{\bclr{\lrecv xb}}
							\\\lwith\\
							\gclr{\lsend x{\lab_2}} \lprec \vclr{\lsend xc}
						\end{pmatrix}
						\lpar
						\begin{pmatrix}
							\gclr{\lrecv x{\lab_1}} \lprec \bclr{\lsend xb}
							\\\lplus\\
							\gclr{\lrecv x{\lab_2}} \lprec \lEx c{\vclr{\lrecv xc}}
						\end{pmatrix}
					}
				}
			}{
				\vliin{\gclr\lwith}{}{
					\rclr{\lsend ya}
					,
					\lEx a{\rclr{\lrecv ya}}
					,
					\begin{pmatrix}
						\gclr{\lsend x{\lab_1}} \lprec \lEx b{\bclr{\lrecv xb}}
						\\\lwith\\
						\gclr{\lsend x{\lab_2}} \lprec \vclr{\lsend xc}
					\end{pmatrix}
					,
					\begin{pmatrix}
						\gclr{\lrecv x{\lab_1}} \lprec \bclr{\lsend xb}
						\\\lplus\\
						\gclr{\lrecv x{\lab_2}} \lprec \lEx c{\vclr{\lrecv xc}}
					\end{pmatrix}
				}{
					\vliin{\rclr{\mixr}}{}{
						\rclr{\lsend ya}
						,
						\lEx a{\rclr{\lrecv ya}}
						,
						\gclr{\lsend x{\lab_1}} \lprec \lEx b{\bclr{\lrecv xb}}
						,
						\begin{pmatrix}
							\gclr{\lrecv x{\lab_1}} \lprec \bclr{\lsend xb}
							\\\lplus\\
							\gclr{\lrecv x{\lab_2}} \lprec \lEx c{\vclr{\lrecv xc}}
						\end{pmatrix}
					}{
						\vlin{\rclr\exists}{}{
							\rclr{\lsend ya}
							,
							\lEx a{\rclr{\lrecv ya}}
						}{\vlin{\rclr\axrule}{}{\rclr{\lsend ya},\rclr{\lrecv ya}}{\vlhy{}}}
					}{
						\vlin{\gclr\lplus}{}{
							\gclr{\lsend x{\lab_1}} \lprec \lEx b{\bclr{\lrecv xb}}
							,
							\begin{pmatrix}
								\gclr{\lrecv x{\lab_1}} \lprec \bclr{\lsend xb}
								\\\lplus\\
								\gclr{\lrecv x{\lab_2}} \lprec \lEx c{\vclr{\lrecv xc}}
							\end{pmatrix}
						}{
							\vliin{\gclr\lprec}{}{
								\gclr{\lsend x{\lab_1}} \lprec \lEx b{\bclr{\lrecv xb}}
								,
								\gclr{\lrecv x{\lab_1}} \lprec \bclr{\lsend xb}
							}{
								\vlin{\gclr\axrule}{}{
									\gclr{\lsend x{\lab_1}}
									,
									\gclr{\lrecv x{\lab_1}}
								}{\vlhy{}}
							}{
								\vlin{\bclr\exists}{}{
									\lEx b{\bclr{\lrecv xb}}
									,
									\bclr{\lsend xb}
								}{
									\vlin{\bclr\axrule}{}{
										\bclr{\lrecv xb}
										,
										\bclr{\lsend xb}
									}{\vlhy{}}
								}
							}
						}
					}
				}{
					\vliin{\rclr{\mixr}}{}{
						\rclr{\lsend ya}
						,
						\lEx a{\rclr{\lrecv ya}}
						,
						\gclr{\lsend x{\lab_2}} \lprec \vclr{\lsend xc}
						,
						\begin{pmatrix}
							\gclr{\lrecv x{\lab_1}} \lprec \bclr{\lsend xb}
							\\\lplus\\
							\gclr{\lrecv x{\lab_2}} \lprec \lEx c{\vclr{\lrecv xc}}
						\end{pmatrix}
					}{
						\vlin{\rclr\exists}{}{
							\rclr{\lsend ya}
							,
							\lEx a{\rclr{\lrecv ya}}
						}{\vlin{\rclr\axrule}{}{\rclr{\lsend ya},\rclr{\lrecv ya}}{\vlhy{}}}
					}{
						\vlin{\gclr\lplus}{}{
							\gclr{\lsend x{\lab_2}} \lprec \vclr{\lsend xc}
							,
							\begin{pmatrix}
								\gclr{\lrecv x{\lab_1}} \lprec \bclr{\lsend xb}
								\\\lplus\\
								\gclr{\lrecv x{\lab_2}} \lprec \lEx c{\vclr{\lrecv xc}}
							\end{pmatrix}
						}{
							\vliin{\gclr\lprec}{}{
								\gclr{\lsend x{\lab_2}} \lprec \vclr{\lsend xc}
								,
								\gclr{\lrecv x{\lab_2}} \lprec \lEx c{\vclr{\lrecv xc}}
							}{
								\vlin{\gclr\axrule}{}{
									\gclr{\lsend x{\lab_2}}
									,
									\gclr{\lrecv x{\lab_2}}
								}{\vlhy{}}
							}{
								\vlin{\bclr\exists}{}{
									\vclr{\lsend xc}
									,
									\lEx c{\vclr{\lrecv xc}}
								}{
									\vlin{\bclr\axrule}{}{
										\vclr{\lsend xc}
										,
										\vclr{\lrecv xc}
									}{\vlhy{}}
								}
							}
						}
					}
				}
			}
		}
	$}

	\adjustbox{max width=\textwidth}{$\begin{array}{c}
	\\\\
		\begin{array}{c|c}
			\mbox{Conflict Net}
		&
			\mbox{Slice Net}
		\\\\\hline\\
			\begin{array}{c}
				\vpz1{\set{\rclr{\lsend ya},\rclr{\lrecv ya}}}
				\quad
				\begin{array}{ccc}
					\vpz5{\set{\bclr{\lsend xb},\bclr{\lrecv xb}}}
					\quad
					\vpz6{\set{\gclr{\lsend x{\lab_1}},\gclr{\lrecv x{\lab_1}}}}
					&
					\vpz7{\set{\vclr{\lsend xc,\vclr{\lrecv xc}}}}
					\quad
					\vpz8{\set{\gclr{\lsend x{\lab_2}},\gclr{\lrecv x{\lab_2}}}}
					\\
					\vpz3{\conc}
					&
					\vpz4{\conc}
					\\
					\multicolumn2c{\vpz2\conf}
				\end{array}
				\\
				\vpz0\conc\qquad\qquad
				\Gedges{
					pz0/pz1,pz0/pz2,
					pz2/pz3,pz2/pz4,
					pz3/pz5,pz3/pz6,
					pz4/pz7,pz4/pz8%
				}
			\end{array}
		&
			\begin{Bmatrix}
				\begin{Bmatrix}
					\set{\rclr{\lsend ya},\rclr{\lrecv ya}},	\\
					\set{\bclr{\lsend xb},\bclr{\lrecv xb}},	\\
					\set{\gclr{\lsend x{\lab_1}},\gclr{\lrecv x{\lab_1}}}
				\end{Bmatrix}\;,
				\\
				\begin{Bmatrix}
					\set{\rclr{\lsend ya},\rclr{\lrecv ya}},	\\
					\set{\vclr{\lsend xc,\vclr{\lrecv xc}}},	\\
					\set{\gclr{\lsend x{\lab_2}},\gclr{\lrecv x{\lab_2}}}
				\end{Bmatrix}\;\;
			\end{Bmatrix}
		\end{array}
	\end{array}$}
	\caption{
		A derivation of the formula $\fof{P}$ from \Cref{eq:introEx1} corresponding to execution tree in the right of \Cref{fig:compTreesEQ}, and its corresponding conflict net.
	}
	\label{fig:introPN}
\end{figure}

\myparagraph{Related Works on Processes as Formulas}
Following the ideas in \cite{mil:uniform},
Miller proposed in \cite{miller:pi} a theory within linear logic allowing to interpret the reduction semantics of the \picalc as implication in the theory, where parallel is internalized by the $\lpar$ and the choice operator $+$ in the original formulation of the \picalc \cite{mil:par:wal:pi} by the $\lplus$.
Guglielmi developed an extension of multiplicative linear logic with a non-commutative connective aiming at internalizing sequentiality in \cite{guglielmi:concurrecy,guglielmi:95:sequentiality}, lately leading to the design \emph{deep inference}  and the formalism of the \emph{calculus of structures} \cite{gug:SIS} to obtain a satisfactory proof system for the logic $\BV$.
In \cite{bru:02} Bruscoli established a \emph{processes-as-formulas} and \emph{computation-as-deduction} correspondence for a simple fragment of $\CCS$ (without recursion, choice, and restriction) where each successful terminating execution of a \proc corresponds to specific derivations in $\BV$.
This correspondence has been extended to the \picalc by Horne, Tiu et al. \cite{hor:tiu:19,hor:tiu:tow,hor:nom},  including the choice operator ($+$), modeled via the additive connective $\lplus$, and name restriction, modeled using \emph{nominal quantifiers} in the spirit of \cite{pitts:nominal,gabbay:pitts:nominal}.
We highlight here the main differences of our approach with respect to the aforementioned works:
\begin{itemize}
	\item
	we use a non-associative non-commutative self-dual connective $\lprec$ (instead of the non-commutative but associative $\lseq$ in $\BV$).
	This choice allows for a cut-free sequent calculus for $\NL$,
	while no sequent calculus for $\BV$ or any of its extension can exist~\cite{tiu:SIS-II};

	\item
	in \cite{hor:tiu:ama:cio:private,hor:nom} Horne et al. use
	the original version of the \picalc \cite{mil:par:wal:pi} which feature a choice operator $+$ with an undesirable ``non-local'' behavior, which requires to forwardly check that it will entail a communication rule.
	Its rule is written as follows
	$$
	+ : A + B \quad  \redsem \quad A' \qquad\mbox{only if } A \redsem A'
	$$
	That is, the choice operator $+$ is not complitely free to choose between $A$ and $B$, but it is constrained by the possibility of performing an action after such a choice. That is, if $A$ cannot perform any action, then the choice $A+B$ cannot reduce to $A$.

	A logical operator modeling such a choice operator should have a rule capable of spotting (within a given context) the sub-formulas on which some rules can be applied.
	Such a behavior, to the best of our knowledge, has never been studied in the literature of proof theory.
	For this reason, we consider the version of the \picalc from \cite{vasco:pi,gay:hole} in which the two choice operators play different roles:
	the \emph{label-send} $\pbras x {\lab:P_\lab}{\lab\in L}$ allows a \proc to choose its continuation independently of the environment (which we model with the additive conjunction $\lwith$, whose rule branches a derivation duplicating the context), while the \emph{label-receive} $\psels x {\lab:P_\lab}{\lab\in L}$ allows a \proc to choose according to the environment (which we model with the additive disjunction $\lplus$, whose rule is applied according to the context's need).
	This latter version of the \picalc is the one used in the literature of session types \cite{vasco:pi,gay:hole,hon:yos:car:multiparty} and choreographic programming \cite{montesi:book}.\footnote{
		Applications of the logical framework we develop, as well as more precise connections with session types and choreographic programming are presented in the companion paper \cite{acc:man:mon:FaP}.
	}

	\item
	in the work of Bruscoli \cite{bru:02}, and in the works of Horne and Tiu \cite{hor:nom,hor:tiu:tow} derivations correspond to executions, while in our work derivations represent execution trees.
	In the latter works, the original Milner's choice operator ($+$) is modelled in the system $\MAV$ by using the additive connective $\lplus$ from additive linear logic, since the additional information about the environment seems to be guaranteed (a-posteriori) by the fact that the correspondence is only established between successful executions and derivations.
	However, undesired behaviors may still occur in establishing the correspondence executions-as-derivations in $\MAV$.
	For an example, consider the process $A+\pnil \ppar \psenda xa \ppar \precv xa$, which is stuck, but derivable in $\MAV$.

	Note that if we restrict the label-send constructor $\pbras x {\lab:P_\lab}{\lab\in L}$ on singleton sets of labels, or, equivalently, if we redesign the rule for the additive connective $\lwith$ in such a way it only has a single premise (that is, we prune the other premise), then we can recover a correspondence between derivations and executions also in our setting.

	\item
	as Horne and Tiu in \cite{hor:tiu:19,hor:tiu:tow,hor:nom},
	we use a pair of dual nominal quantifiers
	(instead of a self-dual quantifier as in \cite{menni:nominal,mil:tiu:nabla,rov:bind})
	to model restriction.%
	\footnote{
		In \cite{hor:tiu:19} the authors report the use of a non-self-dual quantifier to model restriction was suggested them by Alessio Guglielmi in a private communication.
	}
	However, as explained in detail in \Cref{rem:nominal} and in \Cref{sec:embedding},
	our pair of dual quantifiers satisfies different proof theoretical properties.
\end{itemize}

\myparagraph{Structure of the paper}
In \Cref{sec:back} we recall standard definitions for sequent systems and syntax and semantics of the \picalc.
In \Cref{sec:new} we present formulas and sequent systems, explaining the design choices we made in the opeartors of the logic $\NL$.
In \Cref{sec:PT} we study their proof theoretical properties of our system, including relevant formula equivalences and cut-elimination.
In \Cref{sec:PN} we present the syntax of proof nets for $\NL$, providing translations from derivations to proof nets, and from proof nets to derivations (sequentialization).
In \Cref{sec:canon} we prove canonicity for our proof nets with respect to local rule permutations.
In \Cref{sec:PaF} we show how formulas in $\NL$ can be used to encode \procs of the \picalc, and how execution trees of a \proc $P$ can be represented by derivations of the corresponding formula.
Thereby, we show that equivalent execution trees (modulo interleaving) can be represented by the same proof net.
We conclude in \Cref{sec:conc} by discussing extensions of this framework and its possible applications.

\section{Preliminary Notions}\label{sec:back}

We assume the reader to be familiar with the notion of {trees} and of {formula tree}, as well as with the syntax of sequent calculus (see, e.g., \cite{troelstra_schwichtenberg_2000}), but we recall here the main definitions.
We may identify formulas with their formula-trees and we consider \defn{sequents} as forests made of formula-trees.%
\footnote{
	Said differently, a sequent is a set of occurrences of formulas.
	Note that defining a sequent as a multiset of formulas would require the introduction of additional structure to pinpoint on which occurrences of formulas rules are applied, making way more cumbersome the definition of proof nets (\Cref{sec:PN}) and preventing the confluence of cut-elimination due to the impossibility of distinguishing which occurrence of formula is active for a $\cutr$.
}
of formulas in a given grammar.

A \defn{sequent rule} $\rrule$ is an expression of the form $\vlupsmash{\vlinf{\rrule}{}{\sdash\Gamma}{}}$ , $\vlupsmash{\vlinf{\rrule}{}{\sdash\Gamma}{\sdash\Gamma_1}}$ , or $\vlupsmash{\vliinf{\rrule}{}{\sdash\Gamma}{\sdash\Gamma_1}{\sdash\Gamma_2}}$.
The sequent $\Gamma$ is called \defn{conclusion} and the sequents above the line \defn{premises}.
An occurrence of formula in the conclusion \resp{in a premise} of a rule but in none of its premises \resp{not in its conclusion} is said \defn{principal} \resp{\defn{active}}.
A \defn{sequent system} $\XS$ is a set of sequent rules.

A \defn{derivation} in $\XS$ is a non-empty tree $\dD$ of sequents, whose root is called \defn{conclusion}, such that each sequent in $\dD$ is conclusion of a rule in $\XS$, whose children are (all and only) the premises of the rule.
An \defn{open derivation} is a derivation whose leaves may be the conclusion of no rules, in which case are called \defn{open premises}.
We may denote a derivation \resp{an open derivation with an open premise $\Delta$} $\dD$ with conclusion $\Gamma$ by $\vldownsmash{\vlderivation{\vlpr{\dD}{}{\sdash\Gamma}}}$ $\left(\mbox{resp. }\vldownsmash{\vlderivation{\vlde{\dD}{}{\sdash\Gamma}{\vlhy{\sdash\Delta}}}}\right)$ .

\subsection{$\pi$-Calculus}

\begin{figure}[t]
	\centering
	\adjustbox{max width=\textwidth}{$\begin{array}{c}
			\begin{array}{c|c}
				\mbox{Processes}& \mbox{Structural Equivalence}
				\\
				\begin{array}{l@{\;}c@{\;}l@{\;}|@{\;}l}
					P,Q,R
					&\coloneqq	&
					\pnil  			&\mbox{nil}
					\\&|&
					\psend xy P 	&\mbox{send ($y$ on $x$)}
					\\&|&
					\precv xy P	&\mbox{receive ($y$ on $x$)}
					\\&|&
					P \ppar Q		&\mbox{parallel}
					\\&|&
					\pnu x P		&\mbox{restriction (or nu)}
					\\&|&
					\pbras x{\lab : P_\lab}{\lab\in L}
					&\mbox{label-send (on $x$)}
					\\&|&
					\psels x{\lab : P_\lab}{\lab\in L}
					&\mbox{label-receive (on $x$)}
				\end{array}
				&
				\begin{array}{r@{\;\steq\;}l}
					P 					& P^\alpha
					\\
					P \ppar Q 			& Q \ppar P
					\\
					(P \ppar Q) \ppar R & P \ppar (Q \ppar R)
					\\
					\pnu x \pnu y P 	&  \pnu y \pnu x P
					\\
					P \ppar \pnil 		& P
					\\
					\pnu x S	 		& S
					\\
					\pnu x P \ppar S	& \pnu x (P \ppar S)
				\end{array}
				\\\hline
				\begin{tabular}{c}
					with $x,y\in\namesset$ and $L \subset\labelsset$.
					The constructors binding variables
					\\
					are $\pnu xP$  binding $x$ in $P$, and $\precv xy P$ binding $y$ in $P$
					only
				\end{tabular}
				&
				\begin{tabular}{c}
					$P^\alpha$ $\alpha$-equivalent to $P$
					\\
					$x$ is not a name occurring free in $S$
				\end{tabular}
			\end{array}

			\\\hline\hline

			\begin{array}{c|c}
				\multicolumn{2}{c}{\text{Reduction Semantics}}
				\\
				\begin{array}{r@{:\;}r@{\;\redsem\;}ll}
					\rscomr														&
					\psend xa P \ppar \precv xb Q
					&
					P \ppar Q \fsubst ab
					&
					\\
					\rsbrar														&
					\pbras x{\lab: P_\lab}{\lab\in L}							&
					\pbra x{\lab_k}{P_{\lab_k}}									&
					\mbox{if }\lab_k\in L
					\\
					\rsselr														&
					\pbra x{\lab}{P_{\lab_k}}
					\ppar
					\psels x{\lab:Q_{\lab}}{\lab\in L}
					&
					P_{\lab_k} \ppar Q_{\lab_k}
					&
					\mbox{if } \lab_k\in L
				\end{array}
				&
				\begin{array}{r@{\;:\;}r@{\;\redsem\;}l@{\quad\mbox{if}\quad}l}
					\rsresr	& \pnu x P	&\pnu x P'	& P\redsem P'
					\\
					\rsparr	& P\ppar Q	& P'\ppar Q	& P\redsem P'
					\\
					\rsstrr	& P 		& Q 		& P\steq P' \redsem Q'\steq Q
				\end{array}
			\end{array}
		\end{array}$}
	\caption{
		Syntax for \procs, the relations generating the structural equivalence ($\steq$), and the reduction semantics of the \picalc.
		The $\alpha$-equivalence is defined in the usual way (see Appendix).
	}
	\label{fig:terms}
\end{figure}

We consider the version of \picalc presented in \cite{vasco:pi,gay:hole}, whose processes are generated from a countable set of \defn{(channel) names} $\namesset =\set{x,y,\ldots}$ and (disjoint) finite set of \defn{labels} $\labelsset$ grammar in \Cref{fig:terms}.
\footnote{As standard, may write $\psenda xy$ \resp{$\precva xy$} instead of $\psend xy \pnil$ \resp{$\precv xy \pnil$}.}
In the same figure, we recall the definition of the \defn{structural equivalence} ($\steq$), as well as the \defn{reduction semantics}.
We write $P\not\steq Q $ if $P \steq Q$  does not hold.
We may denote by $\cN\ctx[P]$ a \proc of the form $\pnu {x_1}{\cdots \pnu{x_n}{(P\ppar Q)}}$ for some names $x_1,\ldots x_n$ and a \proc $Q$, and write $a$ instead of $a.\pnil$ if $a\in\set{\psenda xy,\precva xy}$.
We denote by $\redsems$ the transitive closure of $\redsem$.

A \proc $P$ is \defn{stuck} if $P \not\steq \pnil$ and there is no $P'$ such that
$P \redsem P'$.
A \proc P is called \defn{deadlock-free} if $P$ is not stuck and there is no stuck \proc $P'$ such that $P \redsems P'$.
A \proc $P$ is \defn{race-free} if there is no $P'$ such that $P\redsems P'$ for a $P'$ structurally equivalent to one of the following processes:
\begin{equation}
	\hfill
	\begin{array}{l@{\qquad\qquad}l}
		\cN\ctx[(\psend xy R \ppar \psend xz Q  \ppar S)]
	&
		\cN\ctx[(\pbras x{P_\lab}{\lab\in L} \ppar \pbras x{P_\lab}{\lab\in L'} \ppar S)]
	\\
		\cN\ctx[(\precv xy R \ppar \precv xz Q  \ppar S)]
	&
		\cN\ctx[(\psels x{P_\lab}{\lab\in L} \ppar \psels x{P_\lab}{\lab\in L'} \ppar S)]
	\end{array}
	\hfill
\end{equation}

A \defn{execution trees} of a \proc $P$ is a trees of processes with root $P$, where a \proc $Q'$ is a child of $Q$ if $Q\redsem Q'$, and such that branching is determined by the intrinsic non-determinism of the reduction rule $\rsbrar$, that is, if two processes $Q_1$ and $Q_2$ are children of a same process $Q$, then $Q\redsem Q_1$ and $Q\redsem Q_2$ via $\rsbrar$ applied to the same minimal (w.r.t. term inclusion) sub-process of $Q$.
We may label the edges of a execution tree with the unique reduction rule in $\set{\rscomr,\rsbrar,\rsselr}$ required to reduce the term $P$ to $Q$.%
\footnote{
	The reduction rules $\rsresr$, $\rsparr$ and $\rsstrr$ are not ``meaningful'' with respect to the computation, and even if a transition step may require multiple instances of these rules to deal with the bureaucracy of the syntax and the structural congruence, only a single instance of a rule in $\set{\rscomr,\rsbrar,\rsselr}$ is required to perform a reduction step.
	For a formal definition of the labelling of the execution tree, see the definition of \emph{core-reduction} in \cite{acc:man:mon:FaP}.
}
The \defn{interleaving} equivalence relation ($\teq$) on execution trees is defined by the relations in \Cref{fig:interleaving}.
See \Cref{fig:compTreesEQ} for an example of two execution trees equivalent modulo interleaving.

\begin{figure}[t]
	\adjustbox{max width=\textwidth}{$\begin{array}{c}
		\begin{array}{c}
			\vmod3{\cN\ctx[P'\ppar \bclr{Q'}]}
			\\[30pt]
			\vmod2{\cN\ctx[\rclr{P'}\ppar \bclr{Q}]}
			\\[30pt]
			\vmod1{\cN\ctx[\rclr{P}\ppar Q]}
		\end{array}
		\Dledges{mod1/mod2/{\rclr{\alpha}},mod2/mod3/{\bclr{\beta}}}
		\teq
		\begin{array}{c}
			\vmod3{\cN\ctx[\rclr{P'}\ppar Q']}
			\\[30pt]
			\vmod2{\cN\ctx[\rclr{P}\ppar \bclr{Q'}]}
			\\[30pt]
			\vmod1{\cN\ctx[P\ppar \bclr{Q}]}
		\end{array}
		\Dledges{mod1/mod2/{\bclr{\beta}},mod2/mod3/{\rclr{\alpha}}}
	\qquad\qquad
		\begin{array}{ccc}
			\vmod3{\cN\ctx[P'\ppar \bclr{Q_1}]}
			&
			\cdots
			&
			\vmod4{\cN\ctx[P'\ppar \bclr{Q_n}]}
			\\[30pt]
			&\vmod2{\cN\ctx[\rclr{P'}\ppar \bclr{Q}]}
			\\[30pt]
			&\vmod1{\cN\ctx[\rclr{P}\ppar Q]}
		\end{array}
		\Dledges{mod1/mod2/{\rclr{\alpha}},mod2/mod3/{\bclr{\gamma_1}},mod2/mod4/{\bclr{\gamma_n}}}
		\teq
		\begin{array}{ccc}
			\vmod4{\cN\ctx[\rclr{P'}\ppar Q_1]}
			&\cdots&
			\vmod5{\cN\ctx[\rclr{P'}\ppar Q_n]}
			\\[30pt]
			\vmod2{\cN\ctx[\rclr{P}\ppar \bclr{Q_1}]}
			&
			\cdots
			&
			\vmod3{\cN\ctx[P\ppar \bclr{Q_n}]}
			\\[30pt]
			&\vmod1{\cN\ctx[P\ppar \bclr{Q}]}
		\end{array}
		\Dledges{
			mod1/mod2/{\bclr{\gamma_1}},
			mod1/mod3/{\bclr{\gamma_n}},
			mod2/mod4/{\rclr{\alpha}},
			mod3/mod5/{\rclr{\alpha}}
		}
	\\\\
		\begin{array}{c@{\!\!\!\!\!}c@{\!\!\!\!\!}c}
			\vmod4{\cN\ctx[P_1\ppar \bclr{Q_1}]}
			\cdots
			\vmod6{\cN\ctx[P_1\ppar \bclr{Q_n}]}
			&&
			\vmod5{\cN\ctx[P_m\ppar \bclr{Q_1}]}
			\cdots
			\vmod7{\cN\ctx[P_m\ppar \bclr{Q_n}]}
			\\[30pt]
			\vmod2{\cN\ctx[\rclr{P_1}\ppar \bclr{Q}]}
			&
			\cdots
			&
			\vmod3{\cN\ctx[\rclr{P_n}\ppar \bclr{Q}]}
			\\[30pt]
			&\vmod1{\cN\ctx[\rclr{P}\ppar Q]}
		\end{array}
		\Dledges{
			mod1/mod2/{\rclr{\delta_1}},
			mod1/mod3/{\rclr{\delta_m}},
			mod2/mod4/{\bclr{\gamma_1}},
			mod2/mod6/{\bclr{\gamma_n}},
			mod3/mod5/{\bclr{\gamma_1}},
			mod3/mod7/{\bclr{\gamma_n}}
		}
		\teq
		\begin{array}{c@{\!\!\!\!\!}c@{\!\!\!\!\!}c}
			\vmod4{\cN\ctx[\rclr{P_1}\ppar Q_1]}
			\cdots
			\vmod6{\cN\ctx[\rclr{P_m}\ppar Q_1]}
			&&
			\vmod5{\cN\ctx[\rclr{P_1}\ppar Q_n]}
			\cdots
			\vmod7{\cN\ctx[\rclr{P_m}\ppar Q_n]}
			\\[30pt]
			\vmod2{\cN\ctx[\rclr{P}\ppar \bclr{Q_1}]}
			&
			\cdots
			&
			\vmod3{\cN\ctx[\rclr{P}\ppar \bclr{Q_m}]}
			\\[30pt]
			&\vmod1{\cN\ctx[P\ppar \bclr{Q}]}
		\end{array}
		\Dledges{
			mod1/mod2/{\bclr{\gamma_1}},
			mod1/mod3/{\bclr{\gamma_n}},
			mod2/mod4/{\rclr{\delta_1}},
			mod2/mod6/{\rclr{\delta_m}},
			mod3/mod5/{\rclr{\delta_1}},
			mod3/mod7/{\rclr{\delta_m}}
		}
	\end{array}$}
	\caption{
		Generators of the execution tree equivalence with $\alpha,\beta\in\set{\rscomr,\rsselr}$ and $\gamma_i,\delta_j\in\set{\rsbrar}$,
		where  $\set{P_1,\ldots,P_m}$ \resp{$\set{Q_1,\ldots,Q_n}$}
		is the set of all \procs such that $P\redsem P_i$ \resp{$Q\redsem Q_j$} via $\rsbrar$.
	}
	\label{fig:interleaving}
\end{figure}

\begin{figure}[t]
	\adjustbox{max width=\textwidth}{$\begin{array}{c}
		\begin{array}{cc}
			\vmod7{\pnil}
			&
			\vmod8{\pnil}
			\\[30pt]
			\vmod5{\pnu x(\bclr{\precva xb \ppar   \psenda xb})}
			&
			\vmod6{\pnu x(\vclr{\precva xc \ppar   \psenda xc})}
			\\[30pt]
			\vmod3{\gclr{\pnu x(Q_1)}}
			&
			\vmod4{\gclr{\pnu x(Q_2)}}
			\\[30pt]
			\multicolumn{2}{c}{\vmod2{\pnu x( \gclr Q )}}
			\\[30pt]
			\multicolumn{2}{c}{\vmod1{\pnu x \pnu y(\rclr{\psenda ya \ppar \precva ya} \ppar Q)}}
		\end{array}
		\Dledges{
			mod1/mod2/{\rclr{\rscomr}},
			mod2/mod3/{\gclr{\rsbrar}},
			mod2/mod4/{\gclr{\rsbrar}},
			mod3/mod5/{\gclr{\rsselr}},
			mod5/mod7/{\bclr{\rscomr}},
			mod4/mod6/{\gclr{\rsselr}},
			mod6/mod8/{\vclr{\rscomr}}
			}
		\begin{array}{cc}
			\vmod8{\pnil}
			&
			\vmod9{\pnil}
			\\[30pt]
			\vmod6{\pnu x(\bclr{\precva xb \ppar   \psenda xb})}
			&
			\vmod7{\pnu y (\rclr{\psenda ya\ppar \precva ya})}
			\\[30pt]
			\vmod4{\gclr{\pnu x(Q_1)}}
			&
			\vmod5{\pnu x\pnu y(\psenda ya \ppar \precva ya \ppar \vclr{ \psenda xc \ppar \precva xc})}
			\\[30pt]
			\vmod2{\pnu x \pnu y(\rclr{\psenda ya \ppar \precva ya} \ppar Q_1)}
			&
			\vmod3{\pnu x\pnu y(\psenda ya \ppar \precva ya \ppar \gclr{Q_2})}
			\\[30pt]
			\multicolumn{2}{c}{\vmod1{\pnu x\pnu y( \psenda ya \ppar \precva ya \ppar \gclr Q)}}
		\end{array}
		\Dledges{
			mod1/mod2/{\gclr{\rsbrar}},
			mod1/mod3/{\gclr{\rsbrar}},
			mod2/mod4/{\rclr{\rscomr}},
			mod4/mod6/{\gclr{\rsselr}},
			mod6/mod8/{\bclr{\rscomr}},
			mod3/mod5/{\gclr{\rsselr}},
			mod5/mod7/{\vclr{\rscomr}},
			mod7/mod9/{\rclr{\rscomr}}
		}
		\\\\
		\mbox{where }
		Q_1=\pbra x{\lab_1}{\precva xb} \ppar  \pselsv x{\lab_1 : \psenda xb , \lab_2 : \precva xc}
		\mbox{ and }
		Q_2=\pbra x{\lab_2}{\psenda xc} \ppar  \pselsv x{\lab_1 : \psenda xb , \lab_2 : \precva xc}
	\end{array}$}
	\caption{
		Two equivalent execution trees of the \proc $P$ from \Cref{eq:introEx1}.
	}
	\label{fig:compTreesEQ}
\end{figure}

\section{A New Logical Framework for the \picalc}\label{sec:new}

\begin{figure}[t]
	\adjustbox{max width=\textwidth}{$
		\begin{array}{c|c|c}
			\mbox{Formulas}
			&
			\mbox{De Morgan Laws}
			&
			\mbox{$\alpha$-equivalence}
		\\
			\begin{array}{l@{\;}c@{\;}llll}
				A,B &\coloneqq	&
				\lunit
				&\mbox{unit (atom)}			\\&\mid&
				\lsend xy
				&\mbox{atom}				\\&\mid&
				\lrecv xy
				&\mbox{atom}				\\&\mid&
				A\lpar B
				&\mbox{par}					\\&\mid&
				A\ltens B
				&\mbox{tensor}				\\&\mid&
				A\lprec B
				&\mbox{prec}				\\&\mid&
				A\lplus B
				&\mbox{oplus}				\\&\mid&
				A\lwith B
				&\mbox{with}				\\&\mid&
				\lFa xA
				&\mbox{for all}				\\&\mid&
				\lEx xA
				&\mbox{exists}				\\&\mid&
				\lNu xA
				&\mbox{new}					\\&\mid&
				\lYa xA
				&\mbox{ya}
			\end{array}
		&
			\begin{array}{ccc}
				\cneg \lunit		&=&	\lunit
				\\
				\cneg{(\cneg A)}	&=&	A
				\\
				\cneg{\lsend xy}	&=& \lrecv xy
				\\
				\cneg{(A \lpar B)}	&=& \cneg A \ltens \cneg B
				\\
				\cneg{(A \lprec B)}	&=& \cneg A \lprec \cneg B
				\\
				\cneg{(A \lplus B)}	&=& \cneg A \lwith \cneg B
				\\
				\cneg{(\forall x.A)}&=& \lEx x{\cneg A}
				\\
				\cneg{(\lNu \xX A)}	&=& \lYa \xX{\cneg A}
			\end{array}
		&
			\begin{array}{c}
				a 	 = 	a
				\\
				\mbox{if $a \in \set{\lunit, \lsend xy , \lrecv xy}$}
			\\\\
				A_1 \circleddot A_2 	= B_1 \circleddot B_2
				\\
				\mbox{if $A_i = B_i$}
				\\
				\mbox{and $\circleddot \in \set{\lpar,\lprec,\ltens,\lplus,\lwith}$}
			\\\\
				\lQu x A 		= 	\lQu y A \fsubst yx
				\\
				\mbox{$y$ fresh for $A$}
				\mbox{and $\lqusymb \in \set{\lnewsymb{}, \lyasymb{}, \forall, \exists}$}
			\end{array}
		\end{array}
	$}
	\caption{Formulas (with $x,y\in \varset$), and their syntactic equivalences.}
	\label{fig:Seq}
\end{figure}


In this section we construct proof systems extending \emph{first-order multiplicative additive linear logic} (or $\MALL^1$) with new operators allowing us to fitfully capture the behavior of term constructors for processes of the \picalc w.r.t. the reduction semantics.

For this purpose, we enrich the language of $\MALL^1$ with a non-commutative connective $\lprec$ designed to capture the logical properties of the (non-commutative) prefix operator used in the \picalc \cite{mil:par:wal:pi} (but also in \CCS \cite{M80}).
Even if it would be desirable to require $\lprec$ to be associative, to capture the associativity of sequential composition of \procs, we instead let $\lprec$ being non-associative to reflect the fact that the prefix operator only allows prefixing a single atomic action at a time, and thus it does not model sequential composition because unable to compose sequentially non-atomic \procs.

To capture restriction, following the spirit of the nominal quantifiers as introduced in \cite{gabbay:pitts:nominal}, we use the nominal quantifier $\lnewsymb$ allowing variable binding.
As already explained in \cite{mil:tiu:nabla}, the universal quantifier cannot not be used to satisfactorily model restriction. For an example, consider the \procs $Q=\pnu x\pnu y( \psend\pnil xa \ppar \precv\pnil ya)$ and $R=\pnu z (\psend\pnil za \ppar \precv\pnil za)$: if we encode restriction by universal quantification, then any property for $Q$ should also be valid for $R$, because $\lFa x \lFa y P(x,y)$ entails $\lFa z P(z,z)$.
The use of the existential quantifiers to model the variable binding of the input action, and nominal quantifiers for restriction allow us to avoid an unsound semantical overlap due to the duality of existential and universal quantification.
For example, if restriction were modelled by the universal quantifier, then the deadlocked process $\pnu a {\psend xa \pnil} \ppar \precv xa\pnil$ would be encoded by the tautology $(\lFa a (\lsend xa \lprec \lunit)) \lpar (\lEx a (\lrecv xa \lprec\lunit)) = (\lEx a (\lrecv xa \lprec\lunit))\limp (\lEx a (\lrecv xa \lprec\lunit))$, while \Cref{thm:deadlock} shows that the formula encoding the process should not be provable. This mismatch is due to the translation, which allows unsound interactions between the binding of input actions with binding of restriction via a duality which is not valid in the semantics.

\begin{remark}
	In our work, we do not consider a self-dual nominal quantifier as the one studied in \cite{mil:tiu:nabla,che:nominal,rov:bind}, but we rather introduce a dual quantifier similarly to what is done in \cite{hor:tiu:19,hor:tiu:tow,hor:nom}, where such a design choice is justified in view of the semantics of the \picalc.

	In \cite{hor:nom} the authors list the three following logical properties a nominal quantifier $\lqusymb$ modeling the binder should satisfy:
	\begin{enumerate}
		\item\label{nom:1} \emph{equivariance}, $\lQu x{\lQu yA}$ and $\lQu y{\lQu xA}$ should be logically equivalent;
		\item\label{nom:2} \emph{non-diagonality}: the formula $\lQu x {\lQu y {A(x,y)}} $ should not imply $\lQu z {A(z,z)}$ or vice versa;
		\item\label{nom:3} \emph{scope extrusion}:
		if $\ominus$ is a connective modeling parallelism, then
		$(\lQu x A) \ominus B$ implies $\lQu x (A \ominus B)$ whenever $x$ does not occur in $B$.
	\end{enumerate}
	In our work, all these requirements are met (see \Cref{prop:feq}), but we also claim that the following additional condition should be added in this list, in view of how restriction and choice operators interact.
	\begin{enumerate}[4]
		\item\label{nom:4} \emph{name-choice}:
		if $\circleddot$ is a connective modeling choice, then $\lQu x A \circleddot \lQu x B $ should imply $\lQup x {A \circleddot B}$.
	\end{enumerate}
	Intuitively, this latter requirement is dictated by the observational indistinguishably of a \proc spawning a fresh name before making a global choice, and a \proc spawning a fresh name after such a choice is made.

	Note that in \cite{pitts:nominal} this latter condition is required to hold for all connectives, and in \cite{hor:nom} such condition holds for the sequential operator.
	However, such a behavior may not be desirable in certain contexts where two processes executed in sequence may behave differently if they share a communication channel or not.
	For an example, consider the case in which a communication channel become loose or vulnerable to attacks after events depending on the use of the channel.
\end{remark}

\begin{definition}\label{def:formulas}
	\defn{Formulas} are generated by a countable set of \defn{variables} ($\varSet$) by the grammar in \Cref{fig:Seq} modulo the standard
	\defn{De Morgan Laws} and
	\defn{$\alpha$-equivalence}
	from the same figure.
	A \defn{context} is a formula containing a special occurrence of an atomic variable $\chole$ (called \defn{hole}) and we denote by $\cC\ctx[A]$ the formula obtained by replacing $\chole$ with a formula $A$.
	An \defn{atom} is either the \defn{unit} $\lunit$, or a predicate $\lsend xy$ or $\lrecv xy$.
	The \defn{(linear) implication} $A\limp B$ is defined as $\cneg A \lpar B$, where the \defn{negation} is defined over formulas by extending the negation on atoms via the \defn{de Morgan laws} in \Cref{fig:Seq}.

	For each formula, we define the set $\freeof A$ of \defn{free variables} as the set of variables occurring in $A$ which are not bounded by any quantifier.
	The free variables in a sequent $\Gamma=A_1,\ldots, A_n$ is the set $\freeof\Gamma=\bigcap_{i=1}^{n}\freeof{A_i}$.
	From now on, we assume sequents to be \defn{clean}, that is, such that each variable $x$ can occur bound by at most a universal quantifier, or by two dual nominal quantifiers.
\end{definition}
\begin{remark}
	To provide a lighter presentation of our systems, as well as to highlight the connections with the \picalc, in this paper we consider formulas whose propositional atoms are generated by a limited signature containing no function symbols and two ``dual'' binary predicates $\lsend--$ and $\lrecv--$.
	However, a more expressive extension could be easily defined, and the results presented in this paper could be straightforwardly extended by addressing simple technical nuances, which we highlight in this paper whenever relevant.
\end{remark}

\begin{figure}[t]
	\centering
	\adjustbox{max width=\textwidth}{$\begin{array}{c}
		\vlinf{\axrule}{}{\sdash \lsend xy, \lrecv xy}{}
		\qquad
		\vlinf{\lpar}{}{\sdash \Gamma, A\lpar B}{\sdash \Gamma, A, B}
		\qquad
		\vliinf{\ltens}{}{\sdash[\sS_1 , \sS_2] \Gamma, A\ltens B,\Delta}{\sdash[\sS_1] \Gamma, A}{\sdash[\sS_2] B, \Delta}
		\qquad
		\vlinf{\lunit}{}{\sdash \lunit}{}
		\qquad
		\vliinf{\mixr}{}{\sdash[\sS_1 , \sS_2] \Gamma, \Delta}{\sdash[\sS_1] \Gamma}{\sdash[\sS_2] \Delta}
	\\\\
		\vlinf{\lplus}{}{\sdash\Gamma, A_1\lplus A_2}{ \sdash \Gamma, A_i}
		\qquad
		\vliinf{\lwith}{}{\sdash \Gamma,A\lwith B}{\sdash \Gamma,A}{\sdash \Gamma,B}
		\qquad
		\vlinf{\forall}{\dagger}{\sdash \Gamma, \lFa x{A}}{\sdash \Gamma, A}
		\qquad
		\vlinf{\exists}{}{\sdash \Gamma, \lEx x{A}}{\sdash \Gamma, A\fsubst yx}
	\\\\\hline\\
		\vliinf{\lprec}{}{
			\sdash[\sS_1 , \sS_2] \Gamma,\Delta, A\lprec B ,C \lprec D
		}{ \sdash[\sS_1] \Gamma, A,C}{\sdash[\sS_2] \Delta, B,D}
		\qquad
		\vliinf{\precur}{}{
			\sdash[\sS_1 , \sS_2] \Gamma,\Delta, A\lprec B
		}{ \sdash[\sS_1] \Gamma, A}{\sdash[\sS_2] \Delta, B}
	\\\\\hline\\
		\vlinf{\nuur}{\dagger}{\sdash \Gamma, \lNu x A}{\sdash[\sS] \Gamma, A}
		\qquad
		\vlinf{\nuloadr}{\dagger}{\sdash \Gamma, \lNu xA}{\sdash[\sS , \isnu x] \Gamma , A}
		\qquad
		\vlinf{\nupopr}{}{\sdash[\sS , \isnu y] \Gamma, \lYa xA}{\sdash \Gamma,A\fsubst yx}
	\\\\
		\vlinf{\yaur}{\dagger}{\sdash \Gamma, \lYa x A}{\sdash[\sS] \Gamma, A}
		\qquad
		\vlinf{\yaloadr}{\dagger}{\sdash \Gamma,\lYa xA}{\sdash[\sS , \isya x] \Gamma , A}
		\qquad
		\vlinf{\yapopr}{}{\sdash[\sS , \isya y] \Gamma, \lNu xA}{\sdash \Gamma, A\fsubst yx}
	\\\\\hline\hline\\
		\vlinf{\AXrule}{}{\sdash[\emptyset] A,\cneg A}{}
		\qquad
		\vliinf{\cutr}{}{\sdash[\sS_1,\sS_2] \Gamma, \Delta}{\sdash[\sS_1] \Gamma, A}{\sdash[\sS_2] \cneg A, \Delta}
		\qquad
		\vlinf{\nqsrule}{\dagger}{\sdash \Gamma, \lYa xA, \lNu xB}{\sdash \Gamma, A,B}
	\end{array}$}
	\caption{
		Sequent calculus rules with side conditions $\dagger\coloneqq x\notin \freeof{\Gamma}$.
		The system $\PIL$ is made of the rules above the double line. The rules below it are derivable.
	}
	\label{fig:rules}
\end{figure}

\begin{definition}
	A \defn{nominal variable} is an element of the form $\isna x$ with $x\in\varSet$ and $\nabla\in\set{\isnusymb,\isyasymb}$.
	If $\sS$ is a set of nominal variables, we say that $x$ \defn{occurs} in $\sS$ if $\isnu x$ or $\isya x$ is an element of $\sS$.
	A  \defn{(nominal) store} $\sS$ is a set of nominal variables such that each variable occurs at most once in $\sS$.
	A \defn{judgement} $\sdash\Gamma$ consists of a store $\sS$, and a clean sequent $\Gamma$.
\end{definition}

\begin{nota}
	We write judgements $\sdash\Gamma$ with
	$\sS=\emptyset$
	\resp{$\sS=\set{ \isna[1]{x_1},\ldots,\isna[n]{x_n}}$}
	simply as
	$\sdash \Gamma$ \resp{$\sdash[ {\isna[1]{x_1},\ldots,\isna[n]{x_n} }]\Gamma$, i.e., omitting parenthesis}.
	We write $\sS_1,\sS_2$ to denote the (disjoint) union of two stores such that a same variable does not occur in both $\sS_1$ and $\sS_2$.
\end{nota}

We define rule systems using rules from \Cref{fig:rules}.
The rules in \Cref{fig:rules} in the first block are standard rules for the first-order multiplicative and additive fragment of linear logic decorated with stores.
As expected, rules $\ltens$ and $\mixr$ split the context (and thus the store) among premises to enforce the linear use of resources, which is typical for multiplicative rules.
In contrast, the rule $\lwith$ (with) duplicates the context (and thus the store).
The rules for the connective $\lprec$ in the second block are also multiplicative in this sense, as they maintain the same context-splitting behavior.
The rules $\nuur$ and $\yaur$ simply remove quantification respecting the freshness condition ($\dagger$), as the standard rule $\forall$ for the universal quantifier.
We are not making use of substitution for these rules because we assume $\alpha$-renaming could be applied to the formula prior to the application of the rule, in order to satisfy the side condition $\dagger$.
Similarly, the rule $\nuloadr$ \resp{$\yaloadr$} removes quantification respecting the freshness condition $\dagger$, as the standard rule $\forall$ for the universal quantifier, but it also adds to the store the nominal variable $\isnu x$ \resp{$\isya x$}, where $x$ is the variable bound by the nominal quantifier of the principal formula.
The rule $\nupopr$ \resp{$\yapopr$} behaves similarly to the rule $\exists$ for the existential quantifier, but removing an occurrence of the dual nominal quantifier $\lyasymb$ \resp{$\lnewsymb$}.
The name is due to the fact that the variable used for the substitution has to be a nominal variable $\isnu x$ \resp{$\isya x$} in the store.

We prove in this section the admissibility of the that the rules below the double line, which are the standard rules for the general (non-atomic) axiom and cut, and a special rule $\nqsrule$ removing a pair of dual nominal quantifiers binding the same variable $x$.

\begin{definition}
	We define the following systems using rules from \Cref{fig:rules}.
	\begin{equation}\label{eq:systems}
		\begin{array}{r@{\;=\;}l@{\quad}r@{\;=\;}l@{\quad}r@{\;=\;}l}
			\MLL		&\Set{\axrule,\lpar,\ltens}
			&
			\MLLx		&\MLL	\cup\set{\lunit,\mixr}
			\\
			\MALL		& \MLL 	\cup\Set{\lplus,\lwith}
			&
			\MALL_1		&\MALL	\cup\Set{\forall,\exists}
			\\
			\NML		&\MLL	\cup\Set{\lprec,\precur,\lunit,\mixr}
			&
			\NMAL		&\MALL	\cup\Set{\lprec,\precur,\lunit,\mixr}
			\\
			\miniPIL	&\NMAL	\cup\Set{\exists,\forall,\nuur,\yaur,\nqsrule}
			&
			\PILm		&\NMAL	\cup\Set{\exists,\forall,\nuur,\yaur,\nuloadr,\nupopr}
			\\
			\multicolumn{4}{c}{\PIL		=\MALL_1	\cup\Set{\lprec,\precur,\lunit,\mixr}	\cup\Set{\nuur,\yaur,\nuloadr,\nupopr,\yaloadr,\yapopr}}
		\end{array}
	\end{equation}
	If $\XS$ is a system, we write $\proves[\XS]\Gamma$ to denote that $\sdash[\emptyset]\Gamma$ is derivable in $\XS$.
\end{definition}

\begin{remark}\label{rem:nomlink}
	During proof search, the nominal quantifier $\nabla\in\set{\lnewsymb,\lyasymb}$ could be removed by a rule $\naloadr$, storing the variable bound by the nominal quantifier in the store (as a nominal variable).
	Since the axiom rule and the unit rule have empty store, each nominal variable $\isna x$ in the store must be used for the substitution of a variable bound by the dual nominal quantifier $\cneg\nabla$.
	That is, any derivation establishes some pairing between each $\naloadr$ with some $\napopr$ rules above it in a derivation -- in absence of additive connectives, such a $\napopr$ above the $\naloadr$ is unique.

	A similar pairing can be observed in the weaker rule $\nqsrule$, which removes a pair of dual nominal quantifiers binding the same variable $x$.
	However, such a localized interaction rules out the possibility to prove quantifier swaps for nominal quantifiers and the nominal-choice laws from \Cref{eq:feq}.
\end{remark}

\begin{remark}\label{rem:nominal}
	The pair $\tuple{\lnewsymb,\lyasymb}$ of nominal quantifiers in $\NL$ behaves differently from the pair   $\tuple{\lnewsymb,\lwensymb}$ considered by Horne and Tiu for the logic $\BVq$ and its extensions \cite{hor:tiu:19,hor:tiu:tow,hor:nom}.
	One difference is the way $\lnewsymb$ and $\lyasymb$ interact in $\NL$, in which each nominal quantifier $\lqusymb$ interacts with at most  one dual quantifier $\cneg\lqusymb$,
	while in $\BVq$ a $\lnewsymb$ can interact with multiple $\lwensymb$.
	By means of example, the implication $(\lNu x{ A} \ltens \lNu x { B}) \limp \lNup x{A\lpar B}$ (i.e., the formula $\lYa x{ \cneg A} \lpar \lYa x { \cneg B} \lpar \lNup x{A\lpar B}$) is provable in $\BVq$ but not in $\NL$.

	This reminds the different ways the modalities in the modal logics $\Mrule$ and $\Krule$ interacts via the rules:
	in the former, each diamond ($\ldia$)interacts with exectly one box ($\lbox$), while in the latter, multiple diamonds can interact with a single box, as shown in the sequent rules of their sequent calculi -- see \cite{ind:monotonicModal,lel:pim:modal,acc:str:modal} for additional details.
	$$\vlinf{\Mrule}{}{\sdash \ldia B,\lbox A}{\sdash B,A}
	\qquand
	\vlinf{\Krule}{\text{\scriptsize $n\in\N$}}{\sdash \ldia B_1,\ldots, \ldia B_n,\lbox A}{\sdash B_1,\ldots, B_n,A}$$

	Moreover, the implication $\lNu xA \limp \cneg{\left(\lNu x \cneg A\right)}$ (that is, the formula $\lYa xA \lpar \lYa xA$) is not derivable in $\NL$, while it is in $\BVq$.

	Another difference depends on the way nominal quantifiers interact with the connective modeling sequentiality.
	Our nominal quantifiers do not satisfy {scope extrusion} over sequentiality, that is, the formula $\lNu x{(A\lprec B)} \feq (\lNu x A) \lprec B$ with $x\notin\freeof B$  is not derivable in $\NL$.
	However, this property is striclty needed in $\BVq$ in order to guarantee that the logical implication is a transitive relation (i.e., that if $A\limp B$ and $B\limp C$ are derivable, then also  $A\limp C$ is derivable).
\end{remark}

\section{Proof theoretical properties of $\NL$}\label{sec:PT}

In this section we prove the proof theoretical properties of the system $\NL$, including the derivability of rules and the possibility of embedding $\NL$ in $\MAVq$.

Our systems satisfy the property referred to as \emph{initial coherence} \cite{avr:canonical:01,mil:pim:13}, that is, the property that atomic axioms suffice to guarantee the possibility of deriving the general axiom rule.
Said differently, in $\NL$ we can derive any formula of the form $A\limp A$ using axiom rules restricted on atoms only.
\begin{proposition}\label{prop:nuyarule}
	Then the rule $\nqsrule$ is derivable in $\set{\nuloadr,\nupopr}$.
\end{proposition}
\begin{proof}
	It suffices to remark that each instance of $\nqsrule$ can be replaced by a  $\nuloadr$ followed (bottom-up) by a $\nupopr$.
\end{proof}
\begin{proposition}\label{prop:AX}
	Then the rule $\AXrule$ is derivable in $\MALL^1\cup\set{\lunit,\lprec,\nuloadr,\nupopr}$.
\end{proposition}
\begin{proof}
	We that the judgement $\sdash \cneg A, A$ is derivable in $\MALL^1\cup\set{\lunit,\lprec,\nqsrule}$ for any formula $A$ by induction on the structure of $A$:
	\begin{itemize}
		\item if $A=\lunit$, then such a derivation is made of a $\mixr$ with premises of the form $\sdash \lunit$, each conclusion of a $\lunit$-rule;

		\item if $A= B\lprec C$, then $\cneg A =\cneg B\lprec\cneg C$.
		We apply a rule $\lprec$, and we conclude by inductive hypothesis;

		\item if $A=\lQu x B$ with $\nabla\in\set{\lnewsymb,\lyasymb}$, then we apply (bottom-up) a rule $\nqsrule$-rule and we conclude by inductive hypothesis;

		\item otherwise, we proceed as standard in $\MALL_1$.

	\end{itemize}
	The statement follow from \Cref{prop:nuyarule}.
\end{proof}

We have the following properties for our connectives, quantifiers and unit. Note that the list in \Cref{eq:feq} is not complete, since additional implications and equivalences immediately follow by duality.
\begin{proposition}\label{prop:feq}
	The following logical equivalences and implications are derivable in $\NL$.
	\begin{equation}\label{eq:feq}\adjustbox{max width=.94\textwidth}{$\begin{array}{c}
		\begin{array}{c|c}
			\mbox{Unit Laws}
		&
			\mbox{Monoidal Laws}
		\\
			\begin{array}{r@{\feq}c@{\feq}l}
				(A\lpar \lunit)  		&  (A\ltens \lunit)  	& A
				\\
				(A\lprec \lunit)  		&  (\lunit \lprec A) 	& A
				\\
				(\lunit\lwith \lunit)	& (\lunit \lplus \lunit)& \lunit
			\end{array}
		&
			\begin{array}{c}
				A\circleddot B \feq B\circleddot A
				\\
				(A\circleddot B) \circleddot C \feq A \circleddot (B \circleddot C)
				\\
				\mbox{with } \circleddot \in \set{\lpar,\ltens,\lplus,\lwith}
			\end{array}
		\\
			\lFa x\lunit\feq \lNu x\lunit \feq \lYa x\lunit\feq\lEx x\lunit	\feq  \lunit
		\\\hline
			\mbox{Scope extrusion}
		&
			\mbox{Quantifier Swap}
		\\
			\begin{array}{c}
				\lNu x{(A\lpar B)}		\feq (\lNu x A) \lpar B
			\\
				\lNu x B \feq B
			\\
				\mbox{if $x \notin \freeof B$}
			\end{array}
		&
			\begin{array}{r@{\feq}l}
				\lQu x {\lQu y A} 	& \lQu y {\lQu x A}
				\\
				\multicolumn{2}{c}{\mbox{with }\lqusymb\in\set{\exists,\lyasymb, \lnewsymb,\forall}}
			\end{array}
		\\\hline
			\mbox{Multiplicative refinement}
		&
			\mbox{Quantifier refinement}
		\\
			\begin{array}{c@{\limp}c}
				(A\ltens B) & (A\lprec B)
			\\
				(A\lprec B)	& (A\lpar B)
			\end{array}
		&
			\begin{array}{c@{\qquad}c}
				\lFa x A \limp \lNu x A
				&
				\lNu x A \limp \lEx x A
				\\
				\lFa x A \limp \lYa x A
				&
				\lYa x A \limp \lEx x A
			\end{array}
		\\\hline
			\mbox{Nominal-choice}
		&
			\mbox{Distributivity of choice}
		\\
			\begin{array}{c@{\;}c@{\;}c}
				(\lNu x A \lwith \lNu x B)	& \feq&
				(\lNu x (A \lwith B))
				\\
				(\lNu x A \oplus \lNu x B) 	&\feq&  (\lNu x (A \oplus B))
			\end{array}
		&
			\begin{array}{c}
				((A\lpar B) \lwith (A\lpar C)) \limp (A\lpar (B\lwith C))
				\\
				(A\lpar (B\lwith C)) \limp ((A\lpar B) \lwith (A\lpar C))
				\\
				(A \lpar (B\lplus C) )\limp ((A \lplus B) \lpar (A\lplus C))
			\end{array}
		\end{array}
	\end{array}$}
	\end{equation}
\end{proposition}
\begin{proof}
	Unit laws follows by the existence of the following derivations.
	$$
	\adjustbox{max width=\textwidth}{$\begin{array}{c}
		\vlderivation{
			\vlin{\lpar}{}{
				\vdash (\cneg A \ltens \lunit )\lpar A
			}{
				\vliin{\ltens}{}{
					\vdash (\cneg A \ltens \lunit ), A
				}{
					\vlin{\AXrule}{}{\vdash \cneg A, A}{\vlhy{}}
				}{
					\vlin{\lunit}{}{\vdash \lunit}{\vlhy{}}
				}
			}
		}
	\qquad
		\vlderivation{
			\vliq{\lpar \times 2}{}{
				\vdash \cneg A \lpar (A \lpar \lunit)
			}{
				\vliin{\mixr}{}{
					\vdash \cneg A, A, \lunit
				}{
					\vlin{\AXrule}{}{\vdash \cneg A, A}{\vlhy{}}
				}{
					\vlin{\lunit}{}{\vdash \lunit}{\vlhy{}}
				}
			}
		}
	\qquad
		\vlderivation{
			\vlin{\lpar}{}{
				\vdash (\cneg A \lprec \lunit) \lpar A
			}{
				\vliin{\precur}{}{
					\vdash \cneg A \lprec \lunit, A
				}{
					\vlin{\AXrule}{}{\vdash \cneg A, A}{\vlhy{}}
				}{
					\vlin{\lunit}{}{\vdash \lunit}{\vlhy{}}
				}
			}
		}
	\\
		\vlderivation{
			\vlin{\lpar}{}{
				\vdash (\lunit \lplus \lunit) \lpar \lunit
			}{
				\vlin{\lplus}{}{
					\vdash \lunit \lplus \lunit , \lunit
				}{
					\vlin{\AXrule}{}{\vdash \lunit, \lunit}{\vlhy{}}
				}
			}
		}
	\qquad
		\vlderivation{
			\vlin{\lpar}{}{
				\vdash (\lunit \lwith \lunit) \lpar \lunit
			}{
				\vliin{\lwith}{}{
					\vdash \lunit \with \lunit , \lunit
				}{
					\vlin{\AXrule}{}{\vdash \lunit, \lunit}{\vlhy{}}
				}{
					\vlin{\AXrule}{}{\vdash \lunit, \lunit}{\vlhy{}}
				}
			}
		}
	\qquad
		\vlderivation{
			\vliin{\ltens}{}{\vdash \lQu x \lunit \feq \lunit}{
				\vlin{\lpar}{}{\vdash \lunit\lpar  \lQu x\lunit}{
					\vlin{\Qrule}{}{\vdash \lunit,\lQu x\lunit}{
						\vlin{\AXrule}{}{\vdash \lunit, \lunit}{\vlhy{}}
					}
				}
			}{
				\vlin{\lpar}{}{\vdash \lnQu x\lunit \lpar \lunit}{
					\vlin{\dQrule}{}{\vdash \lnQu x\lunit, \lunit}{
						\vlin{\AXrule}{}{\vdash \lunit, \lunit}{\vlhy{}}
					}
				}
			}
		}
	\end{array}$}
	$$

	Monodail laws are proven as standard in $\MALL$.
	Scope extrusion and nominal quantifiers swaps are proven as shown in \Cref{eq:FEQscope} below.
	\begin{equation}\label{eq:FEQscope}
	\hfill
		\vlderivation{
			\vlin{2\times\lpar}{}{
				\vdash \lYa x{(\cneg A \ltens \cneg B)} \lpar ((\lNu x A) \lpar B )
			}{
				\vlin{\nuloadr}{}{
					\vdash \lYa x{(\cneg A \ltens \cneg B)} , \lNu x A , B
				}{
					\vlin{\nupopr}{}{
						\sdash[\isnu x] \lYa x{(\cneg A \ltens \cneg B)} ,  A , B
					}{
						\vliin{\ltens}{}{
							\vdash (\cneg A \ltens \cneg B),A,B
						}{
							\vlin{\AXrule}{}{\vdash \cneg A, A}{\vlhy{}}
						}{
							\vlin{\AXrule}{}{\vdash \cneg B, B}{\vlhy{}}
						}
					}
				}
			}
		}
		\qquad
		\vlderivation{
			\vliq{2\times \nuloadr}{}{
				\vdash \lNu x{\lNu y A} \lpar \lYa y{\lYa x \cneg A}
			}{
				\vliq{2\times \nupopr}{}{
					\sdash[\isnu x,\isnu y] \lYa y {\lYa x {\cneg A}}
				}{
					\vlin{\AXrule}{}{\vdash \cneg A, A}{\vlhy{}}
				}
			}
		}
	\hfill
	\end{equation}
	Quantifier swap for universal and existential quantifiers are standard as in in $\MALL_1$.
	For multiplicative refinement we only show the derivation for $(A\ltens B) \limp (A\lprec B)$ on the left of \Cref{eq:derFE}, since the derivation for $(A\lprec B)\limp (A\lpar B)$ is similar.
	Nominal refinements are proven as shown in the right of \Cref{eq:derFE} below.
	\begin{equation}\label{eq:derFE}
	\hfill
		\vlderivation{
			\vlin{\lpar}{}{\vdash (\cneg A \lpar \cneg B) \lpar (A \lprec B)
			}{
				\vlin{\lpar}{}{
					\vdash \cneg A \lpar \cneg B , A \lprec B
					}{
						\vliin{\precur}{}{
							\vdash \cneg A, \cneg B, A \lprec B
						}{
							\vlin{\AXrule}{}{\vdash A, \cneg  A}{\vlhy{}}
						}{
							\vlin{\AXrule}{}{\vdash B , \cneg B}{\vlhy{}}
						}
					}
			}
		}
	\hfill
		\vlderivation{
			\vlin{\lpar}{}{
					\vdash \lEx x{\cneg A} \lpar \lNa xA
				}{
					\vlin{\nqur}{}{
						\vdash \lEx x{\cneg A} , \lNa xA
					}{
						\vlin{\exists}{}{
							\vdash \lEx x{\cneg A} , A
							}{
								\vlin{\AXrule}{}{\vdash \cneg A , A}{\vlhy{}}}
					}
				}
		}
	\hfill
	\end{equation}
	Finally, distributivity of the choice are standard in $\MALL$ (they are proven by applying (bottom-up) $\lpar$ and $\lwith$ rules first, followed by $\lplus$ and $\ltens$ rules.), and nominal-choice laws are proven as follows.
	\begin{equation}\label{eq:derFE1}\adjustbox{max width=.9\textwidth}{
		$\begin{array}{c}
		\vlderivation{
			\vliin{\ltens}{}{
				\vdash (\lNu x A \lwith \lNu x B) \feq (\lNu x (A \lwith B))
			}{
				\vlin{\lpar}{}{
					\vdash \lYap x{\cneg A \lplus \cneg B} \lpar (\lNu x{A} \lwith \lNu x{B})
				}{
					\vliin{\lwith}{}{
						\vdash \lYap x{\cneg A \lplus \cneg B} , \lNu x{A} \lwith \lNu x{B}
						}{
							\vlin{\nuloadr}{}{
								\vdash \lYap x{\cneg A \lplus \cneg B} , \lNu x{A}
							}{
								\vlin{\nupopr}{}{
									\sdash[\isnu x] \lYap x{\cneg A \lplus \cneg B} , {A}
								}{
									\vlin{\lplus}{}{
										\vdash {\cneg A \lplus \cneg B} , {A}
										}{
											\vlin{\AXrule}{}{\vdash \cneg A, A}{\vlhy{}}
										}
								}
							}
						}{
							\vlin{\nuloadr}{}{
								\vdash \lYap x{\cneg A \lplus \cneg B} , \lNu x{B}
							}{
								\vlin{\nupopr}{}{
									\sdash[\isnu x] \lYap x{\cneg A \lplus \cneg B} , {B}
								}{
									\vlin{\lplus}{}{
										\vdash {\cneg A \lplus \cneg B} , {B}
										}{
											\vlin{\AXrule}{}{\vdash \cneg B, B}{\vlhy{}}
										}
								}
							}
						}
				}
			}{
				\vlin{\lpar}{}{
					\vdash (\lYa x{\cneg A} \lplus \lYa x{\cneg B}) \lpar \lNup x{A \lwith B}
				}{
					\vlin{\nuloadr}{}{
							\vdash \lYa x{\cneg A} \lplus \lYa x{\cneg B} ,  \lNup x{A \lwith B}
						}{
							\vliin{\lwith}{}{
									\sdash[\isnu x] \lYa x{\cneg A} \lplus \lYa x{\cneg B},  {A \lwith B}
								}{
									\vlin{\lplus}{}{
											\sdash[\isnu x] \lYa x{\cneg A} \lplus \lYa x{\cneg B}, A
										}{
											\vlin{\nupopr}{}{
													\sdash[\isnu x] \lYa x{\cneg A} , A
												}{
													\vlin{\AXrule}{}{\vdash \cneg A , A}{\vlhy{}}
												}
										}
								}{
									\vlin{\lplus}{}{
											\sdash[\isnu x] \lYa x{\cneg A} \lplus \lYa x{\cneg B}, B
										}{
											\vlin{\nupopr}{}{
													\sdash[\isnu x] \lYa x{\cneg B} , B
												}{
													\vlin{\AXrule}{}{\vdash \cneg B , B}{\vlhy{}}
												}
										}
								}
						}
				}
			}
		}
	\\\\
		\vlderivation{
			\vliin{\ltens}{}{
				\vdash (\lNu x A \oplus \lNu x B) 	\feq  (\lNu x (A \oplus B))
			}{
				\vlin{\lpar}{}{
					\vdash (\lYa x{\cneg A} \lwith \lYa x{\cneg B}) \lpar \lNup x{A \lplus B}
				}{
					\vliin{\lwith}{}{
						\vdash \lYa x{\cneg A} \lwith \lYa x{\cneg B} , \lNup x{A \lplus B}
					}{
						\vlin{\nuloadr}{}{
							\vdash \lYa x{\cneg A} , \lNup x{A \lplus B}
						}{
							\vlin{\nupopr}{}{
								\sdash[\isnu x] \lYa x{\cneg A} , {A \lplus B}
							}{
								\vlin{\lplus}{}{
									\vdash{\cneg A} , {A \lplus B}
								}{
									\vlin{\AXrule}{}{\vdash{\cneg A} , A
								}{\vlhy{}}}
							}
						}
					}{
						\vlin{\nuloadr}{}{
							\vdash \lYa x{\cneg B} , \lNup x{A \lplus B}
						}{
							\vlin{\nupopr}{}{
								\sdash[\isnu  x] \lYa x{\cneg B} , {A \lplus B}
							}{
								\vlin{\lplus}{}{
									\vdash{\cneg B} , {A \lplus B}
								}{
									\vlin{\AXrule}{}{\vdash{\cneg B} , B
								}{\vlhy{}}}
							}
						}
					}
				}
			}{
				\vlin{\lpar}{}{
					\vdash (\lNu x{A} \lwith \lNu x{B}) \lpar \lYap x{\cneg A \lplus \cneg B}
				}{
					\vlin{\yaloadr}{}{
						\vdash \lNu x{A} \lwith \lNu x{B}, \lYap x{\cneg A \lplus \cneg B}
					}{
						\vliin{\lwith}{}{
							\sdash[\isya x] \lNu x{A} \lwith \lNu x{B}, \cneg A \lplus \cneg B
						}{
							\vlin{\yapopr}{}{
								\sdash[\isya x] \lNu xA, \cneg A \lplus \cneg B
							}{
								\vlin{\nupopr}{}{
									\vdash  A, \cneg A \lplus \cneg B
								}{
									\vlin{\AXrule}{}{
										\vdash  A, \cneg A
									}{\vlhy{}}
								}
							}
						}{
							\vlin{\yapopr}{}{
								\sdash[\isya x] \lNu xA, \cneg A \lplus \cneg B
							}{
								\vlin{\nupopr}{}{
									\vdash  B, \cneg A \lplus \cneg B
								}{
									\vlin{\AXrule}{}{
										\vdash  B,\cneg B
									}{\vlhy{}}
								}
							}
						}
					}
				}
			}
		}
	\end{array}$}
	\end{equation}
\end{proof}

\begin{remark}\label{rem:PILm}
	The system $\miniPIL$ satisfies all the equivalences and implications in \Cref{prop:feq}, except for quantifier swap and nominal choices,
	while the system $\PILm$ satisfies all the equivalences and implications in \Cref{prop:feq}, except for the bottom-most nominal choice.
	For this latter, it only satisfies the left-to-right implication $\proves[\PILm] (\lNu x A \oplus \lNu x B) 	\limp (\lNu x (A \oplus B))$ (see the right branch of the bottom-most derivation in \Cref{eq:derFE1}), but not the converse.
\end{remark}

\subsection{Cut-Elimination}\label{sec:cut}

\begin{figure}[t]
	\adjustbox{max width=\textwidth}{$\begin{array}{c}
		\vlderivation{
			\vliin{\cutr}{}{
				\sdash \Gamma , a
			}{
				\vlhy{\sdash \Gamma, a}
			}{
				\vlin{\axrule}{}{\sdash \cneg a, a}{}
			}
		}
		\;\lcutelim\;
		\vlderivation{\sdash \Gamma, a}
	\qquad
		\vlderivation{
			\vliin{\mixr}{}{\sdash \Gamma}{
				\vlhy{\sdash \Gamma}
			}{
				\vliin{\cutr}{}{
					\sdash \quad
				}{
					\vlin{\urule}{}{\sdash \lunit}{\vlhy{}}
				}{
					\vlin{\urule}{}{\sdash \lunit}{\vlhy{}}
				}
			}
		}
		\;\lcutelim\;
		\vlderivation{\vlhy{\sdash\Gamma}}
	\\[20pt]
		\vlderivation{
			\vliin{\cutr}{}{
				\sdash[\sS,\sS_1,\sS_2] \Gamma, \Delta_1, \Delta_2
			}{
				\vlin{\lpar}{}{
					\sdash \Gamma, A \lpar B
				}{
					\vlhy{\sdash\Gamma, A, B}
				}
			}{
				\vliin{\ltens}{}{
					\sdash[\sS_1,\sS_2] \cneg A \ltens \cneg B, \Delta_1 , \Delta_2
				}{
					\vlhy{\sdash[\sS_1] \cneg A, \Delta_1}
				}{
					\vlhy{\sdash[\sS_2] \cneg B, \Delta_2}
				}
			}
		}
	\;\lcutelim\;
		\vlderivation{
			\vliin{\cutr}{}{
				\sdash \Gamma, \Delta_1, \Delta_2
			}{
				\vliin{\cutr}{}{
					\sdash \Gamma, \Delta_1, A
				}{
					\vlhy{\sdash \Gamma, A, B}
				}{
					\vlhy{\sdash[\sS_1] \cneg A, \Delta_1}
				}
			}{
				\vlhy{\sdash[\sS_2] \cneg B, \Delta_2}
			}
		}
	\\[20pt]
		\vlderivation{
			\vliin{\cutr}{}{
				\sdash[\sS_1,\sS_2] \Gamma, \Delta
			}{
				\vlin{\lplus}{}{
					\sdash[\sS_1] \Gamma, A_1 \lplus A_2
				}{
					\vlhy{\sdash[\sS_1] \Gamma, A_i}
				}
			}{
				\vliin{\lwith}{}{
					\sdash[\sS_2] \cneg{A_1} \lwith \cneg{A_2} , \Delta
				}{
					\vlhy{\sdash[\sS_2] \cneg{A_1} , \Delta}
				}{
					\vlhy{\sdash[\sS_2] \cneg{A_2} , \Delta}
				}
			}
		}
		\;\lcutelim\;
		\vlderivation{
			\vliin{\cutr}{}{\sdash[\sS_1,\sS_2] \Gamma, \Delta}{
				\vlhy{\sdash[\sS_1] \Gamma , A_i}
			}{
				\vlhy{\sdash[\sS_2] \cneg{A_i} , \Delta}
			}
		}
	\\[20pt]
		\vlderivation{
			\vliin{\cutr}{}{
				\sdash \Gamma, \Delta
			}{
				\vlin{\forall}{}{
					\sdash[\sS_1] \Gamma,\lFa x{\cneg A}
				}{\vlpr{}{\dD}{{\sdash[\sS_1] \Gamma,\cneg A}}}
			}{
				\vlin{\exists}{}{
					\sdash[\sS_2] \lEx x{A},\Delta
				}{\vlhy{\sdash[\sS_2] A\fsubst{c}{x},\Delta}}
			}
		}
		\;\lcutelim\;
		\vlderivation{
			\vliin{\cutr}{}{\sdash \Gamma, \Delta}{
				\vlpr{}{\dD \fsubst cx}{{\sdash[\sS_1] \Gamma, \cneg A\fsubst{c}{x}}}
			}{
				\vlhy{\sdash[\sS_2] A\fsubst{c}{x},\Delta}
			}
		}
	\end{array}$}
	\caption{Cut-elimination steps for $\MALL_1$ and $\MLLx$, where $\dD\fsubst cx$ is the derivation obtained by replacing all occurrences of $x$ in $\dD$ with $c$.}
	\label{fig:cut-elimMALL}
\end{figure}

\begin{figure}[t]
	\adjustbox{max width=\textwidth}{$\begin{array}{c}
		\vlderivation{
			\vliin{\cutr}{}{
				\sdash[\sS_1,\sS_2,\sS_3,\sS_4] \Gamma_1, \Gamma_2, C \lprec D , E \lprec F, \Delta_1, \Delta_2
			}{
				\vliin{\lprec}{}{
					\sdash[\sS_1,\sS_2] \Gamma_1, \Gamma_2, C \lprec D , A \lprec B
				}{
					\vlhy{\sdash[\sS_1] \Gamma_1 , C , A}
				}{
					\vlhy{\sdash[\sS_2] \Gamma_2 , D , B}
				}
			}{
				\vliin{\lprec}{}{
					\sdash[\sS_3,\sS_4] \cneg A \lprec \cneg B, E \lprec F , \Delta_1, \Delta_2
				}{
					\vlhy{\sdash[\sS_3] \cneg A, E, \Delta_1}
				}{
					\vlhy{\sdash[\sS_4] \cneg B, F, \Delta_2}
				}
			}
		}
	\;\lcutelim\;
		\vlderivation{
			\vliin{\lprec}{}{
				\sdash[\sS_1,\sS_2,\sS_3,\sS_4] \Gamma_1, \Gamma_2, C \lprec D , E \lprec F, \Delta_1, \Delta_2
			}{
				\vliin{\cutr}{}{
					\sdash[\sS_1,\sS_3] \Gamma_1 ,C, E, \Delta_1
				}{
					\vlhy{\sdash[\sS_1] \Gamma_1 , C, A}
				}{
					\vlhy{\sdash[\sS_3] \cneg A, E, \Delta_1}
				}
			}{
				\vliin{\cutr}{}{
					\sdash[\sS_2,\sS_4] \Gamma_2, D, F, \Delta_2
				}{
					\vlhy{\sdash[\sS_2] \Gamma_2, D, B}
				}{
					\vlhy{\sdash[\sS_4] \cneg B, F, \Delta_2}
				}
			}
		}
	\\[20pt]
		\vlderivation{
			\vliin{\cutr}{}{
				\sdash[\sS_1,\sS_2,\sS_3,\sS_4] \Gamma_1, \Gamma_2, C \lprec D , \Delta_1, \Delta_2
			}{
				\vliin{\lprec}{}{
					\sdash[\sS_1,\sS_2] \Gamma_1, \Gamma_2, C \lprec D , A \lprec B
				}{
					\vlhy{\sdash[\sS_1] \Gamma_1 , C , A}
				}{
					\vlhy{\sdash[\sS_2] \Gamma_2 , D , B}
				}
			}{
				\vliin{\precur}{}{
					\sdash[\sS_3,\sS_4]\cneg A \lprec \cneg B, \Delta_1, \Delta_2
				}{
					\vlhy{\sdash[\sS_3] \cneg A, \Delta_1}
				}{
					\vlhy{\sdash[\sS_4] \cneg B, \Delta_2}
				}
			}
		}
	\;\lcutelim\;
		\vlderivation{
			\vliin{\precur}{}{
				\sdash[\sS_1,\sS_2,\sS_3,\sS_4] \Gamma_1, \Gamma_2, C \lprec D , \Delta_1, \Delta_2
			}{
				\vliin{\cutr}{}{
					\sdash[\sS_1,\sS_3] \Gamma_1 ,C,\Delta_1
				}{
					\vlhy{\sdash[\sS_1] \Gamma_1 , C, A}
				}{
					\vlhy{\sdash[\sS_3]\cneg A, \Delta_1}
				}
			}{
				\vliin{\cutr}{}{
					\sdash[\sS_2,\sS_4] \Gamma_2, D,\Delta_2
				}{
					\vlhy{\sdash[\sS_2] \Gamma_2, D, B}
				}{
					\vlhy{\sdash[\sS_4] \cneg B, \Delta_2}
				}
			}
		}
	\\[20pt]
		\vlderivation{
			\vliin{\cutr}{}{
				\sdash[\sS_1,\sS_2,\sS_3,\sS_4] \Gamma_1,\Gamma_2, \Delta_1, \Delta_2
			}{
				\vliin{\precur}{}{
					\sdash[\sS_1,\sS_2]\Gamma_1,\Gamma_2, A \lprec B
				}{
					\vlhy{\sdash[\sS_1]\Gamma_1, A}
				}{
					\vlhy{\sdash[\sS_2] \Gamma_2, B}
				}
			}{
				\vliin{\precur}{}{
					\sdash[\sS_3,\sS_4] \cneg A \lprec \cneg B, \Delta_1 , \Delta_2
				}{
					\vlhy{\sdash[\sS_3] \cneg A, \Delta_1}
				}{
					\vlhy{\sdash[\sS_4] \cneg B, \Delta_2}
				}
			}
		}
	\;\lcutelim\;
		\vlderivation{
			\vliin{\mixr}{}{
				\sdash[\sS_1,\sS_2,\sS_3,\sS_4] \Gamma_1, \Gamma_2, \Delta_1, \Delta_2
			}{
				\vliin{\cutr}{}{
					\sdash[\sS_1,\sS_3] \Gamma_1, \Delta_1
				}{
					\vlhy{\sdash[\sS_1] \Gamma_1,A}
				}{
					\vlhy{\sdash[\sS_3] \cneg A, \Delta_1}
				}
			}{
				\vliin{\cutr}{}{
					\sdash[\sS_2,\sS_4] \Gamma_2, \Delta_1
				}{
					\vlhy{\sdash[\sS_2] \Gamma_2, B}
				}{
					\vlhy{\sdash[\sS_4] \cneg B,\Delta_2}
				}
			}
		}
	\end{array}$}
	\caption{Cut-elimination steps for the connective $\lprec$ and its rules.}
	\label{fig:cut-elimPREC}
\end{figure}

We prove the admissibility of the rule $\cutr$, we provide a cut-elimination procedure adapting the one for $\MALL_1$.
In absence of the nominal quantifier, the proof taking into account the connective $\lprec$ is straightforward.
In the presence of the nominal quantifier, the proof is more intricate because of the implicit links between $\nqloadr$-rules and $\nqpopr$-rules in a derivation we discuss in \Cref{rem:nomlink}.
For example, consider the derivation with $\cutr$ of the formula $\lNu x a\limp \lNa x a$ in the left of \Cref{eq:exNomCutElim}, where we marked the flows of the nominal variables.
\begin{equation}\label{eq:exNomCutElim}\hfill
	\adjustbox{max width=.94\textwidth}{$
		\vlderivation{
			\vlin{\lpar}{}{\sdash\lNu {\vx0} a\limp \lNa {\vx{11}} a}{
				\vliin{\cutr}{}{
					\vdash \lYa {\vx1}\cneg a, \lNu {\vx{10}}a
				}{
					\vlin{\yaloadr}{}{
						\vdash \lYa {\vx2}\cneg a , \lNu {\vx5}a
					}{
						\vlin{\yapopr}{}{
							\sdash[\isya {\vx3}] \cneg a , \lNu {\vx4}a
						}{
							\vlin{\AXrule}{}{
								\vdash a, \cneg a
							}{\vlhy{}}
						}
					}
				}{
					\vlin{\yaloadr}{}{
						\vdash \lYa {\vx6}\cneg a , \lNu {\vx9}a
					}{
						\vlin{\yapopr}{}{
							\sdash[\isya {\vx7}] \cneg a , \lNu {\vx8}a
						}{
							\vlin{\AXrule}{}{
								\vdash a, \cneg a
							}{\vlhy{}}
						}
					}
				}
			}
		}
		\pzflows{x3.center/x4.center/8/magenta}
		\pzflows{x7.center/x8.center/8/magenta}
		\pzflows{x5.center/x6.center/-6/magenta}
		\flowedges{x0.center/x1.center,x1.center/x2.center,x2.center/x3.center,x4.center/x5.center,x6.center/x7.center,x8.center/x9.center,x9.center/x10.center,x10.center/x11.center}
		\qquad\lcutelim^*\qquad
		\vlderivation{
			\vlin{\lpar}{}{\sdash\lNu {\vx0} a\limp \lNa {\vx{11}} a}{
				\vliq{\yaloadr}{}{
					\vdash \lYa {\vx1}{\cneg a} , \lNu {\vx2}{a}
				}{
					\vlin{\yapopr}{}{
						\sdash[\isya{\vx3}] \cneg a , \lNu {\vx4}{a}
					}{
						\vlin{\AXrule}{}{
							\sdash a, \cneg a
						}{\vlhy{}}
					}
				}
			}
		}
		\pzflows{x3.center/x4.center/12/magenta}
		\flowedges{x0.center/x1.center,x1.center/x3.center,x11.center/x2.center,x2.center/x4.center}
	$}
	\hfill
\end{equation}
In order to perform cut-elimination, we need to be able to keeping track of the variables bound by dual nominal quantifiers, which are supposed to be linked by the $\cutr$-rule, even when the nominal quantifiers are removed.
For this purpose, we introduce the following auxiliary \defn{store-cut} rule we use during the rewriting process of cut-elimination.
\begin{equation}\label{eq:nqscutr}
	\hfill
	\vlinf{\nqscutr}{}{
		\sdash \Gamma
	}{
		\sdash[\sS,\isnu x,\isya x] \Gamma
	}
	\hfill
\end{equation}

\begin{restatable}[Cut elimination]{theorem}{cutelim}\label{thm:cutelim}
	Let $\Gamma$ a non-empty sequent.
	If $\proves[\NL\wcut]\Gamma$, then $\proves[\NL]\Gamma$.
\end{restatable}
\begin{proof}
	We define the \emph{weight} of a $\cutr$-rule $\rrule$ in a derivation $\dD$ as a pair $\tuple{\ldist\rrule\dD,\lcomp\rrule\dD}$ where $\ldist\rrule\dD$ is the maximal distance of $\rrule$ from a leaf above it in the derivation tree,
	and $\lcomp\rrule\dD$ is the complexity of the active formula(s) of $\rrule$.
	The \defn{weight} of a $\nqscutr$-rule $\rrule$ is defined similarly as $\tuple{\ldist\rrule\dD,0}$.
	The \emph{weight} of a derivation is the multiset of the weights of its $\cutr$-rules and $\nqscutr$-rules.

	To prove cut-elimination it suffices to apply the \emph{cut-elimination steps} in \Cref{fig:cut-elimMALL,fig:cut-elimPREC,fig:cut-elimNQ} to a top-most $\cutr$-rule or $\nqscutr$-rule in the derivation tree.
	The fact that we consider the procedure to operate on a derivation whose conclusion and premises judgements have empty stores and non-empty sequents ensures that the case analysis we consider covers all the possible cases.
	In particular, the case $\cutstep\urule\urule$ in \Cref{fig:cut-elimMALL} for the $\lunit$ (because the sequent in the conclusion cannot be empty), and the bottom-most case in \Cref{fig:cut-elimNQ} (because the store in the conclusion and in the premises is empty).

	In order to be able to apply this strategy, as standard in the literature, we consider the \emph{commutative cut-elimination steps}, that is, rule permutations as the ones in \Cref{fig:permutations1} involving a $\cutr$-rule or a $\nqscutr$-rule, allowing us to permute an instance of such rule above another rules.
	The termination of cut-elimination follows by the fact that each cut-elimination step applied to a top-most $\cutr$-rule $\rrule$ decreases $\lcomp\rrule\dD$, while each commutative cut-elimination step applied to the top-most $\cutr$-rule, or to a $\nqscutr$-rule reduces $\ldist\rrule\dD$.
	Note that a commutative step moving a $\cutr$-rule above a $\lwith$-rule duplicate the $\cutr$-rule.
	This is why we have to define the weight as a multiset: even if the complexity does not change, the maximal distance from these two new $\cutr$-rules from the leaves is strictly smaller than the one of the original one.
\end{proof}
\begin{remark}
	In systems containing a self-dual unit $\lunit$, which is the same unit for conjunction and disjunction, such as multiplicative linear logic with mix \cite{gir:ll}, $\pomset$ logic \cite{retore:phd} and $\BV$ \cite{gug:SIS}, it is possible to derive the empty sequent.
	This depends on the fact the empty sequent is not interpreted as false (as in classical logic), but rather as the unit $\lunit$, which is provable, and that the non-admissibility of the weakening rule (as in relevant logics \cite{relevant,acc:str:relevant}) would not entail the possibility of deriving any sequent.
	Citing Girard (as reported in \cite{bellin:subnets}) ``if one were to accept this rule [$\mixr$], the good taste would require to add the void sequent as an axiom (without weakening this has no dramatic consequence)''.
	This explains the structure of the cut-elimination step $\urule$-vs-$\urule$ in \Cref{fig:cut-elimMALL}.
\end{remark}
\begin{remark}
	If we consider a system where the only rule for nominal quantifier is the rule $\nqsrule$, then all judgements in derivations have empty store.

	Note that the system $\NL\setminus\set{\nupopr,\yaloadr}$ presented in \cite{manara:phd} also satisfy cut-elimination, and it already expressive enought to support the interpreting derivations as execution trees (see \Cref{sec:PaF}).
	Indeed, the only difference is that in such a system the nominal-choice law $(\lNu x A \oplus \lNu x B) \feq  (\lNu x (A \oplus B))$ does not hold, but only the left-to-right implication is provable as explained in \Cref{rem:PILm}.
\end{remark}

\begin{figure}[t!]
	\centering
	\adjustbox{max width=.6\textheight}{$\begin{array}{c}
		\vlderivation{
			\vliin{\cutr}{}{
				\sdash[\sS_1,\sS_2,\isna x] \Gamma, \Delta
			}{
				\vlin{\naloadr}{}{
					\sdash[\sS_1] \Gamma, \lNa xA
				}{\vlhy{{\sdash[\sS_1,\isna x] \Gamma, A}}}
			}{
				\vlin{\napopr}{}{
					\sdash[\sS_2,\isna x] \lnNa x{\cneg A},\Delta
				}{
					\vlhy{\sdash[\sS_2] \cneg A,\Delta}
				}
			}
		}
	\;\lcutelim\;
		\vlderivation{
			\vliin{\cutr}{}{\sdash[\sS_1,\sS_2,\isna x] \Gamma, \Delta}{
				\vlhy{{\sdash[\sS_1,\isna x] \Gamma, \cneg A}}
			}{
				\vlhy{\sdash[\sS_2] \cneg A,\Delta}
			}
		}
	\\\\
	%
	\vlderivation{
			\vliin{\cutr}{}{
				\sdash[\sS_1,\sS_2] \Gamma, \Delta
			}{
				\vlin{\nqloadr}{}{
					\sdash[\sS_1] \Gamma, \lNa xA
				}{\vlpr{}{\dD\ctx[\isna x]}{{\sdash[\sS_1,\isna x] \Gamma, A}}}
			}{
				\vlin{\nnaur}{}{
					\sdash[\sS_2] \lnNa x{\cneg  A},\Delta
				}{\vlhy{\sdash[\sS_2] \cneg A,\Delta}}
			}
		}
	\;\lcutelim\;
		\vlderivation{
			\vliin{\cutr}{}{\sdash[\sS_1,\sS_2] \Gamma, \Delta}{
				\vlpr{}{\dD\fsubstup \emptyset{\isna x}}{
					\sdash[\sS_1] \Gamma, A
				}
			}{
				\vlhy{\sdash[\sS_2] \cneg A,\Delta}
			}
		}
	\\\\
	%
		\vlderivation{
			\vlin{\nqloadr}{}{
				\sdash[\sS_1,\sS_2] \Gamma',\lNa x B
			}{
				\vlde{}{\dD\ctx[\isna x]}{
					\sdash[\sS_1,\sS_2,\isna x] \Gamma',B
				}{
					\vliin{\cutr}{}{
						\sdash[\sS_1,\sS_2,\isna x] \Gamma,\Delta
					}{
						\vlin{\nnapopr}{}{
							\sdash[\sS_1,\isna x] \Gamma, \lnNa xA
						}{\vlhy{\sdash[\sS_1]  \Gamma, A}}
					}{
						\vlin{\naur}{}{
							\sdash[\sS_2] \lNa x\cneg A,\Delta
						}{
							\vlhy{\sdash[\sS_2] \cneg A,\Delta}
						}
					}
				}
			}
		}
	\;\lcutelim\;
		\vlderivation{
			\vlin{\nqur}{}{
				\sdash[\sS_1,\sS_2] \Gamma',\lNa y B
			}{
				\vlde{}{\dD\fsubst \emptyset{\isna x}}{
					\sdash[\sS_1,\sS_2] \Gamma',B
				}{
					\vliin{\cutr}{}{\sdash[\sS_1,\sS_2] \Gamma, \Delta}{
						\vlhy{\sdash[\sS_1] \Gamma,  A}
					}{
						\vlhy{\sdash[\sS_2] \cneg A, \Delta  }
					}
				}
			}
		}
	\\\\
	%
		\vlderivation{
			\vliin{\cutr}{}{
				\sdash[\sS_1,\sS_2] \Gamma,\Delta
			}{\vlin{\nuur}{}{
					\sdash[\sS_1] \Gamma , \lNu xA
				}{\vlhy{\sdash[\sS_1] \Gamma,A}}
			}{\vlin{\yaur}{}{
					\sdash[\sS_2] \lYa x{\cneg A}, \Delta
				}{\vlhy{\sdash[\sS_2] \cneg A , \Delta}}
			}
		}
	\;\lcutelim\;
		\vlderivation{
			\vliin{\cutr}{}{\sdash[\sS_1,\sS_2] \Gamma, \Delta}{
				\vlhy{\sdash[\sS_1] \Gamma, A}
			}{
				\vlhy{\sdash[\sS_2] \Delta, \cneg A}
			}
		}
	\\\\\hline\\
	%
		\vlderivation{
			\vliin{\cutr}{}{
				\sdash[\sS_1,\sS_2] \Gamma, \Delta
			}{
				\vlin{\nuloadr}{}{
					\sdash[\sS_1] \Gamma, \lNu x A
				}{
					\vlhy{{\sdash[\sS_1,\isnu x] \Gamma,A}}
				}
			}{
				\vlin{\yaloadr}{}{
					\sdash[\sS_2] \lYa x{\cneg A},\Delta
				}{
					\vlhy{\sdash[\sS_2,\isya x] \cneg A,\Delta}
				}
			}
		}
	\;\lcutelim\;
		\vlderivation{
			\vlin{\nqscutr}{}{
				\sdash[\sS_1,\sS_2] \Gamma, \Delta
			}{
				\vliin{\cutr}{}{
					\sdash[\sS_1,\sS_2,\isnu x,\isya x] \Gamma, \Delta
				}{
					\vlhy{\sdash[\sS_1,\isnu x] \Gamma, A}
				}{
					\vlhy{\sdash[\sS_2,\isya x] \cneg A,\Delta}
				}
			}
		}
	\\\\
	%
		\vlderivation{
			\vlin{\nqscutr}{}{\sdash[\sS_1,\sS_2] \Gamma, \Delta}{
				\vliin{\cutr}{}{
				\sdash[\sS_1,\sS_2,\isya x ,\isnu x] \Gamma, \Delta
			}{
				\vlin{\nupopr}{}{
					\sdash[\sS_1, \isya x] \Gamma, \lNu xA
				}{
					\vlhy{\sdash[\sS_1] \Gamma, A}
				}
			}{
				\vlin{\yapopr}{}{
					\sdash[\sS_2,\isnu x] \lYa x\cneg A,\Delta
				}{
					\vlhy{\sdash[\sS_2] \cneg A ,\Delta}
				}
			}
			}
		}
	\;\lcutelim\;
	\vlderivation{
		\vliin{\cutr}{}{
			\sdash[\sS_1,\sS_2] \Gamma, \Delta
		}{
			\vlhy{\sdash[\sS_1] \Gamma, A}
		}{
			\vlhy{\sdash[\sS_2] \cneg A,\Delta}
		}
	}
	\end{array}$}
	\caption{
		Cut-elimination steps for the nominal quantifiers where
		$\dD\fsubst\emptyset{\isna x}$ is the derivation obtained by removing all occurrences of $\isna x$ in $\dD$,
		and
		$\dD\fsubstup{\emptyset}{\isna x}$ is the derivation obtained by removing all occurrences of $\isna x$ in $\dD$, and replacing any rule $\nqpopr$ introducing $\isna x$ in the store with a rule $\nqur$.
	}
	\label{fig:cut-elimNQ}
\end{figure}

\begin{corollary}\label{cor:transImp}
	The linear implication ($\limp$) in $\NL$ defines a transitive relation, that is, if $\proves[\NL]A\limp B$ and $\proves[\NL]B\limp C$, then $\proves[\NL]A\limp C$.
\end{corollary}
\begin{proof}
	If $\proves[\NL]A\limp B$, then there is a derivation $\dD^-_{A\limp B}$ with conclusion $\sdash \cneg A, B$ because the rule $\lpar$ is invertible (that is, its conclusion is derivable iff its premise is so).
	For the same reason, by hypothesis, there is a derivation $\dD^-_{B\limp C}$ in $\NL$ with conclusion $\sdash \cneg B, C$.
	Thus a derivation with conclusion $\sdash A\limp C$ made of (bottom-up) a $\lpar$-rule followed by a $\cutr$-rule whose premises are the conclusion of $\dD^-_{A\limp B}$ and $\dD^-_{B\limp C}$.
	We conclude by applying cut-elimination.
\end{proof}

We conclude this section stating that $\NL$ can be embedded in $\MAVq$ \cite{hor:nom} using a translation $\trbv{\cdot}$ replacing each occurrence of $\lprec$ with a $\lseq$, and each occurrence of $\lyasymb$ with a $\lwensymb$.
Formal definitions and details of the proof are provided in \Cref{sec:embedding}.

\begin{restatable}{theorem}{thmMAVembed}\label{thm:embedding}
	Let $A_1,\ldots, A_n$ be formulas.
	If $\proves[\NL] A_1,\ldots, A_n$, then $\proves[\MAVq] \bigparr_{i=1}^n\trbv{A_i}$.
\end{restatable}

\section{Proof Nets for $\NL$}\label{sec:PN}

In this section, we define proof nets for $\NL$ and we prove soundness and completeness of this syntax by providing sequentialization and proof translation (desequentialization) procedures.

To handle the interaction between multiplicative ($\lpar$, $\ltens$ and $\lprec$) and additive ($\lplus$ and $\lwith$) connectives in $\NL$ in a canonical way, we develop two syntaxes for proof nets:
\begin{itemize}
	\item \emph{conflict nets} (extending \cite{hughes:conflict,hei:hug:conflict}):
	these proof nets provide canonical representative for proofs modulo local rule permutations, with polynomial-time correctness criterion, sequentialization, and proof translation; and
	\item \emph{slice nets} \cite{hug:van:slice} (extending \cite{hug:van:sliceLICS,hug:van:slice}):
	these proof nets provide canonical representative for proofs modulo rule permutations (inlcuding non-local ones), with polynomial-time correcness criterion and sequentialization, but no polynomial time proof translation.
\end{itemize}
In particular, the former syntax is optimal for the application we aim at in this paper (see \Cref{sec:canon}).

In the reminder of this section, we first provide some shared definitions on links and dualizers.
We then define conflict nets, proving soundness and completeness by providing procedures for sequentialization and proof translation.
We then define slice nets and we prove soundness and completeness using the results on conflict nets.

We first adapt the definitions of links to take into account stores.
\begin{definition}
	A \defn{link} $\la$ on a judgement $\sdash \Gamma$ (or a sequent $\Gamma$) is either
	\begin{itemize}
		\item a \defn{sub-judgement} of $\sdash \Gamma$, that is, a judgement $\sdash[\sS']\Gamma'$ such that $\Gamma'$ is an induced sub-forest of $\Gamma$ and $\sS'\subseteq S$; or
		\item \defn{nominal link}: a pair $\set{x,y}$ of variables occurring in $\Gamma$ with $x$ bound by $\lnewsymb$ and $y$ by a $\lyasymb$, or a pair of the form $\set{\isnu x,y}$ \resp{$\set{\isya x,y}$} made of a nominal variable $\isnu x$ \resp{$\isya x$} in the store, and a variable $y$ occurring in $\Gamma$ bound by the nominal quantifier $\lyasymb$ \resp{$\lnewsymb$};
	\end{itemize}
	A link is \defn{axiomatic} if it is nominal link containing no nominal variables, or a sub-sequent made of a single occurrence of $\lunit$ or a pair of the atoms of the form $\set{\lsend xy,\lrecv zt}$.
	A \defn{linking} \resp{\defn{axiomatic linking}} on $\sdash \Gamma$ is a set of links \resp{axiomatic links} on $\sdash \Gamma$.
\end{definition}
\begin{nota}
	We represent a link $\la$ by drawing a horizontal line labeled by $\la$ connected via vertical segments to the roots of each variable or subformula in the link.
\end{nota}
\begin{example}
	In the left of \Cref{eq:exPreLink} we show a decorated with the links provided on the right.
	\begin{equation}\label{eq:exPreLink}
		\begin{array}{c|c}
			\sdash[\isya w , \isnu{\vz1}] \lYa {\vv2}{\vA1\ltens \vB1}, (\vNu1 \vx1.{\vC1})\vpz1{\lpar}(\vD1\vlplus1 \vE1), \vYa1 \vy1.\vF1
			\pzlinks{A1/C1/12/\la/red/lplus1}
			\pzlinks{x1/y1/-10/\lc/violet/}
			\pzlinks{B1/Ya1/-16/\lb/blue/pz1}
			\pzlinks{z1/v2/12/z/magenta/}
		\quad&\quad
			\begin{array}{l}
				\sclr z =\set{\isnu z, v}
				\\
				\la=A, C, D \lplus E
				\\
				\lb=B,\lFa xF,(\lNu xC)\lpar(D\oplus E)
				\\
				\lc= \set{x,y}
			\end{array}
		\end{array}
	\end{equation}
	The sequent $\Gamma'=D,D\oplus E$ is not a link for the given judgement because it is not a sub-judgement since the formula $D$ is repeated twice.
\end{example}

We then fix a notation for substitutions which we use to encode the information about the witnesses of the quantifiers.

\def\wmaps{\sigma}
\def\wmapt{\tau}
\def\wmapr{\rho}

\begin{nota}
	We use the standard notation $\wmaps=\fsubsts{x_1/y_1,\ldots,x_n/y_n}$  for \defn{substitutions}, i.e., (partial) maps over the set of variables\footnote{In the language of $\NL$ we have no function symbols, but this definition could be generalized by defining a substitution as a map from variables to terms, as done in \cite{hei:hug:str:ALL1}.} with \defn{domain} $\set{y_1,\ldots, y_n}$.
	Moreover, we use the following denotations:
	\begin{itemize}
		\item $\dualizerof[\emptyset]$ is the \defn{empty substitution};

		\item $\wmaps\wmapt$ is the \defn{composition} of $\wmaps$ and $\wmapt$ (in which $\wmaps$ is applied after $\wmapt$);

		\item
		$\wmaps\fsubminus{x}$ is the substitution obtained from $\wmapt$
		by removing the substitution of the variable $x$;

		\item
		$\wmaps$ is \defn{more general} than $\wmapt$
		(denoted $\wmaps \leq \wmapr$)
		if there is a map $\wmapr$ such that $\wmaps\wmapr=\wmapt$;

		\item
		$\wmaps$ and $\wmapt$ are \defn{disjoint}
		(denoted $\newdisj\wmaps{\wmapt}$)
		if so are their domains.
		We may write
		$\wmaps\duasum\wmapt$ to denote $\wmaps\wmapt=\wmapt\wmaps$ whenever $\newdisj\wmaps{\wmapt}$;

		\item
		$\wmaps$ and $\wmapt$ are \defn{coherent}
		(denoted $\newcohe\wmaps{\wmapt}$)
		if there is $\wmapr$ such that $\wmaps\wmapr=\wmapt\wmapr$;

		\item
		the \defn{join} of $\wmaps$ and $\wmapt$ (denoted $\wmaps\join\wmapt$) is the least map $\wmapr$ such that $\wmaps\leq\wmapr$  and $\wmapt\leq\wmapr$;

	\end{itemize}
\end{nota}

\begin{definition}
	A \defn{dualizer} $\dualizerof[\la]$ for a link $\la $ on a judgement $\sdash \Gamma$ (or a sequent $\Gamma$) is a substitution with domain variables either occurring in $\sS$, or occurring bound by an existential quantifier ($\exists$) or by a nominal quantifier ($\lnewsymb$ or $\lyasymb$) in $\Gamma$.

	A \defn{linking with witnesses} $\tuple{\linking,\dualizerof^{\linking}}$ is a linking $\linking$ on $\sdash \Gamma$ provided with a \defn{witness map} $\dualizerof^{\linking}$ associating to each link in $\linking$ a (possibly empty) dualizer.
	
	The \defn{initial witness map} of an axiomatic linking $\linking$ is defined as follows for each $\la\in\linking$:
	\begin{itemize}
		\item if $\la=\set{\lsend xy,\lrecv zt}$, then $\dualizerof[\la]^{\linking}=\fsubsts{x/z,y/t}$;
		\item if $\la=\set{x,y}$ is a nominal link with $x$ bound by $\lnewsymb$ and $y$ bound by $\lyasymb$, then $\dualizerof[\la]^{\linking}=\fsubsts xy$;
		\item if $\la=\set{\lunit}$, then $\dualizerof[\emptyset]$.
	\end{itemize}
\end{definition}

%
\def\negspace{\!\!\!\!\!\!\!\!\!\!\!\!}
\def\mif{\mbox{ if }}
\begin{figure}[t]\adjustbox{max width=\textwidth}{$\begin{array}{c@{\!\!}cl}
		\begin{array}{c}
			\begin{array}{ccc}
				\vpz1{\la}  && \vic1 \cdots \vic{k}
				\\[5pt]
				&\vpz2{\bullet}
			\end{array}
			\multiGedges{pz2}{pz1,c1,ck}
			\\\coaldo\\
			\begin{array}{ccc}
				\vpz1{\lc} & &\vic1 \cdots \vic{k}
				\\[5pt]
				&\vpz2{\bullet}
			\end{array}
			\multiGedges{pz2}{pz1,c1,ck}
			\\
			\mbox{with }
			\bullet \in \set{\conc,\conf}
		\end{array}
	&\mbox{via}&
		\begin{cases}
			(\lpar): &\mif
			\vA1 \vpz1{\lpar} \vB1, \viC1,\ldots ,\viC n
			\pzlinks{A1/B1/12/\la/red/{C1,Cn}}
			\mbox{ and }
			\vA1 \vpz1{\lpar} \vB1, \viC1,\ldots ,\viC n
			\pzlinks{pz1/C1/12/\lc/violet/{Cn}}
		\\\\
			(\lplus): &\mif
			\viA1 \vpz1{\lplus} \viA2, \viB1,\ldots ,\viB n
			\pzlinks{A1/B1/12/\la/red/{Bn}}
			\left(\mbox{or }\viA1 \vpz1{\lplus} \viA2, \viB1,\ldots ,\viB n
			\pzlinks{A2/B1/12/\la/red/{Bn}}\right)
			\mbox{ and }
			\viA1 \vpz1{\lplus} \viA2, \viB1,\ldots ,\viB n
			\pzlinks{pz1/B1/12/\lc/violet/{Bn}}
		\\\\
			(\exists): &\mif
			\vpz1{\exists}x.\vA1,\quad \viB1,\ldots ,\viB n
			\pzlinks{A1/B1/12/\la/red/{Bn}}
			\mbox{ and }
			\vpz1{\exists}x.\vA1, \viB1,\ldots ,\viB n
			\pzlinks{pz1/B1/12/\lc/violet/{Bn}}
			\mbox{ and }
			\dualizerof[\lc]=\dualizerof[\la]\fsubminus{x}
		\\\\
			(\forall): &\mif
			\vpz1{\forall}x.\vA1, \viB1,\ldots ,\viB n
			\pzlinks{A1/B1/12/\la/red/{Bn}}
			\quand
			\vpz1{\forall}x.\vA1, \viB1,\ldots ,\viB n
			\pzlinks{pz1/B1/12/\lc/violet/{Bn}}
			\mbox{ and }
			x\notin\freeof{C_1,\ldots,C_n}
			\mbox{ and }
			x\in\domof{\dualizerof[\la]}
		\\\\
			(\naloadr): &\mif
			\vNa1 \vx1.\vA1,\viB1,\ldots ,\viB n
			\pzlinks{A1/B1/12/\la/red/{Bn}}
			\mbox{ and }
			\vNa1 \vx1.\vA1, \viB1,\ldots ,\viB n
			\pzlinks{Na1/B1/12/\lc/violet/{Bn}}
			\mbox{ and }
			x\notin\freeof{C_1,\ldots,C_n}
		\\\\
			(\naur): &\mif
			\vNa1 \vx1.\vA1,\viB1,\ldots ,\viB n
			\pzlinks{A1/B1/12/\la/red/{Bn}}
			\mbox{ and }
			\vNa1 \vx1.\vA1, \viB1,\ldots ,\viB n
			\pzlinks{Na1/B1/12/\lc/violet/{Bn}}
			\mbox{ and }
			x\notin\freeof{C_1,\ldots,C_n}
			\mbox{ and }
			\set{x,y}\notin\linktree
		\end{cases}
	\\\hline\\
		\begin{array}{c}
			\begin{array}{ccc}
				\vpz1{\la} \quad \vpz2{\lb} && \vic1 \cdots \vic k
				\\[5pt]
				&\vconc1
			\end{array}
			\multiGedges{conc1}{pz1,pz2,c1,ck}
		\\\coaldo\\
			\begin{array}{ccc}
				\vpz1{\lc} & &\vic1 \cdots \vic k
				\\[5pt]
				&\vconc1
			\end{array}
			\multiGedges{conc1}{pz1,c1,ck}
		\\
			\mbox{
				with $\disj(\dualizerof[{\la}], \dualizerof[{\lb}])$
			}
		\end{array}
	&\mbox{via}&
		\underbrace{\begin{cases}
			(\ltens):&\mif
			\vA1 \vpz1{\ltens} \vB1, \viC1,\ldots ,\viC n,\viD1,\ldots ,\viD m
			\pzlinks{A1/C1/12/\la/red/{Cn}}
			\pzlinks{B1/D1/-12/\lb/blue/{Dm}}
			\mbox{ and }
			\vA1 \vpz1{\ltens} \vB1, \viC1,\ldots ,\viC n,\viD1,\ldots ,\viD m
			\pzlinks{pz1/C1/12/\lc/violet/{Cn,D1,Dm}}
		\\\\
			(\lprec): &\mif
			\viA1 \vpz1{\lprec} \viB1,\viA2 \vpz2{\lprec} \viB2, \viC1,\ldots ,\viC n,\viD1,\ldots ,\viD m
			\pzlinks{A1/A2/12/\la/red/{C1,Cn}}
			\pzlinks{B1/B2/-12/\lb/blue/{D1,Dm}}
			\mbox{ and }
			\viA1 \vpz1{\lprec} \viB1,\viA2 \vpz2{\lprec} \viB2, \viC1,\ldots ,\viC n,\viD1,\ldots ,\viD m
			\pzlinks{pz1/pz2/12/\lc/violet/{C1,Cn,D1,Dm}}
		\\[10pt]
			(\precur):&\mif
			\vA1 \vpz1{\lprec} \vB1, \viC1,\ldots ,\viC n,\viD1,\ldots ,\viD m
			\pzlinks{A1/C1/12/\la/red/{Cn}}
			\pzlinks{B1/D1/-12/\lb/blue/{Dm}}
			\mbox{ and }
			\vA1 \vpz1{\lprec} \vB1, \viC1,\ldots ,\viC n,\viD1,\ldots ,\viD m
			\pzlinks{pz1/C1/12/\lc/violet/{Cn,D1,Dm}}
			\quad
			\begin{cases}
				\mbox{only if no other rule }
				\\
				\mbox{except $\mixr$ can be applied}
			\end{cases}
		\\
			(\mixr): &\mif
			\viA1 , \ldots ,\viA n,\viB1,\ldots, \viB m
			\pzlinks{A1/An/12/\la/red/{}}
			\pzlinks{B1/Bm/12/\lb/blue/{}}
			\mbox{ and }
			\viA1 , \ldots ,\viA n,\viB1,\ldots, \viB m
			\pzlinks{A1/Bm/12/\lc/violet/{An,B1}}
			\qquad\qquad\qquad\qquad
			\begin{cases}
				\mbox{only if no other rule}
				\\
				\mbox{can be applied}
			\end{cases}
		\\
		\end{cases}}_{\mbox{only if $\dualizerof[{\lc}] = \dualizerof[{\la}] \duasum \dualizerof[{\lb}]$}}
	\\\\
		&& (\napopr): \mif
		\vnNa1 \vx1.\vA1, \vNa1 \vy1.\vB1, \viC1,\ldots ,\viC n
		\pzlinks{A1/B1/12/\la/red/{C1,Cn}}
		\pzlinks{x1/y1/-12/\lb/blue/}
		\mbox{ and }
		\vnNa1 x.\vA1, \vNa1 \vy1.\vB1, \viC1,\ldots ,\viC n
		\pzlinks{nNa1/B1/12/\lc/violet/{C1,Cn}}
		\quad
		\mbox{
			only if
			$x,y\notin\freeof{C_1,\ldots,C_n}$,
			with $\dualizerof[\lc]=\dualizerof[\la]\fsubminus{y}$
		}
	\\\\\hline\\
		\begin{array}{c}
			\begin{array}{c}
				\vpz1{\la_1}\cdots \vpz3{\la_h}
				\qquad
				\vpz2{\lb_1}\cdots\vpz4{\lb_k}
				\\[5pt]
				\vconf1
			\end{array}
			\multiGedges{conf1}{pz1,pz2,pz3,pz4}
			\\\coaldo\\
			\begin{array}{ccc}
				\vpz1{\la_1}\cdots \vpz3{\la_h}
				&&
				\vpz2{\lb_1}\cdots\vpz4{\lb_k}
				\\[5pt]
				\vconf2&& \vconf3
				\\[5pt]
				&\vconf1
			\end{array}
			\multiGedges{conf1}{conf2,conf3}
			\multiGedges{conf2}{pz1,pz3}
			\multiGedges{conf3}{pz2,pz4}
			\\
			\mbox{with $\cohe(\dualizerof[{\la}], \dualizerof[{\lb}])$ and $\dualizerof[{\lc}] = \dualizerof[{\la}] \join \dualizerof[{\lb}]$}
			\\
			\mbox{$h,k >0$ and $h+k>2$}
		\end{array}
	&\mbox{via}&
		(\conf): \mif
		\mbox{exists $A\lwith B$ in $\Gamma$ such that }
		\begin{cases}
			A\in \la_i \mbox{ and } B\notin \la_i \mbox{ for all }i\in\intset1h
			\\
			A\notin \lb_j \mbox{ and } B\in \lb_j \mbox{ for all }j\in\intset1k
		\end{cases}
	\\\\\hline\\
		\begin{array}{c}
			\begin{array}{c}
				\vpz1{\la} \qquad \vpz2{\lb}
				\\[5pt]
				\vconf1
			\end{array}
			\multiGedges{conf1}{pz1,pz2}
			\\\coaldo\\
			\begin{array}{c}
				\vpz1{\lc}
				\\[10pt]
				\vconf1
			\end{array}
			\multiGedges{conf1}{pz1}
		\\
			\mbox{
				with
				$\cohe(\dualizerof[{\la}], \dualizerof[{\lb}])$
				and
				$\dualizerof[{\lc}] = \dualizerof[{\la}] \join \dualizerof[{\lb}]$
			}
		\end{array}
	&\mbox{via}&
		(\lwith): \mif
			\vA1 \vlwith1 \vB1, \viC1,\ldots ,\viC n
			\pzlinks{A1/C1/12/\la/red/{Cn}}
			\pzlinks{B1/Cn/-12/\lb/blue/{C1}}
			\mbox{ and }
			\vA1 \vlwith1 \vB1, \viC1,\ldots ,\viC n
			\pzlinks{lwith1/C1/12/\lc/violet/{Cn}}
	\qquad\qquad\qquad
		\begin{array}{c}
			\vpz1{\la}
			\\[10pt]
			\vpz2{\bullet}\multiGedges{pz2}{pz1}
			\\\coaldo\\
			\la
			\\
			\mbox{with }\bullet \in\set{\conc,\conf}
		\end{array}
		\quad\mbox{via}\quad
		(\bullet)
	\end{array}$}
	\caption{
		Coalescence steps for \cotrees, with $\nabla\in\set{\lnewsymb,\lyasymb}$ and $\la,\lb,c_1,\ldots, c_n$ leaves.
	}
	\label{fig:coalescence}
\end{figure}

\def\canon#1{\left\lfloor #1 \right\rfloor}
\def\cconf{\dot{\conf}}
\def\cconc{\dot{\conc}}

\subsection{Conflict Nets}\label{subsec:CN}

Conflict nets for $\MALL$ are trees alternating \emph{concord} ($\conc$) and \emph{conflict} ($\conf$) nodes, having the elements of an axiomatic linking as leaves, and satisfying a contractibility criterion with respect to a rewriting procedure called \emph{coalescence}.

\begin{definition}\label{def:cotree}
	A \defn{concord-conflict tree} (or \defn{\cotree} for short) for a linking $\linking$ on a judgement $\sdash \Gamma$ is a tree $\linktree$ with leaves labeled by links $\linking$, and internal nodes labeled by $\conc$ (\defn{concord} nodes) or by $\conf$ (\defn{conflict} nodes).
	It is \defn{axiomatic} if $\linking$ is axiomatic.
	We denote by $\canon{\linktree}$ the \cotree obtained by merging adjacent $\conc$-nodes \resp{$\conf$-nodes}, and by removing nodes with a single child (by attaching its child to its parent).
	A \cotree $\linktree$ is \defn{canonical} if $\linktree=\canon{\linktree}$.
	We may denote $\linktree_1 \conc \linktree_2$ or $\conc\left(\linktree_1,\ldots, \linktree_n\right)$
	\resp{$\linktree_1 \conf \linktree_2$ or $\conf\left(\linktree_1,\ldots, \linktree_n\right)$ }
	a \cotree with root a $\conc$-node \resp{$\conf$-node} having as children the roots of $\linktree_1,\ldots,\linktree_n$.

	As for linkings with witnesses, we define a \defn{\cotree with witnesses} $\tuple{\linktree,\dualizerof}$ as a \cotree $\linktree$ equipped with a \defn{witness map} $\dualizerof$ associating to each leaf $\la$ of $\linktree$ a (possibly empty) dualizer $\dualizerof[\la]$.
\end{definition}

In order to define the coalescence criterion, we need to consider \cotrees for linkings with witnesses, rather than \cotrees for linkings.

\def\dlinktree{\linktree^{\dualizerof}}
\begin{definition}
	In \Cref{fig:coalescence} we define \defn{coalescence steps} over \cotrees with witnesses.
	A \cotree $\linktree$ for a linking $\linking$ on $\vdash \Gamma$ is \defn{coalescent} if there is a \defn{coalescence path}, that is, a sequence of coalescence steps from the \cotree with witnesses $\tuple{\linktree,\dualizerof^{\linking}}$ (where $\dualizerof^{\linktree}$ is the initial witness map for $\linking$) to a \cotree with witnesses consisting of a single leaf $\vdash\Gamma$ with empty dualizer.

	A \defn{conflict net} for a sequent $\Gamma$ is a coalescent axiomatic \cotree $\linktree$ on $\vdash \Gamma$.
\end{definition}

\begin{theorem}\label{thm:CNsoundAndComplete}
	Let $\Gamma$ be a sequent.
	Then $\proves[\NL]\Gamma$ iff there is a \conet $\linktree$ on $\Gamma$.
\end{theorem}
\begin{proof}
	For each derivation $\dD$ of $\Gamma$, we define the (axiomatic) \cotree $\cnof{\dD}$ by translating top-down rules $\dD$ as shown in \Cref{fig:deseq}.
	Prove that $\linktree_\dD$ is coalescent is trivial: it suffices to consider a coalescence path where coalescence steps, which are in correspondence with rules in $\NL$, respect the order in which we translate the proof -- Note that rules $\mixr$ and $\precur$ may require to be postponed during such translation, and applied out-of-order because of the side conditions we have on coalescence steps.

	To prove the converse, we define \emph{deductive \cotrees} as \cotrees with witnesses such that the leaves which are judgements are labeled by derivations with conclusion the corresponding judgement.
	For the base case, we consider the initial axiomatic \cotree $\linktree$ with witnesses $\dualizerof^{\linktree}$ whose leaves which are judgements are labeled by rules $\axrule$ or $\urule$.
	We then define a derivation $\dD_{\linktree}$ with conclusion $\vdash \Gamma$ as shown in \Cref{fig:coacot}, reasoning by induction on the length of a given coalescence path for $\linktree$, which exists by definition of conflict net.
\end{proof}
\begin{corollary}[Conflict nets for $\MALL^1$]\label{cor:confForMALL1}
	A sequent $\Gamma$ is derivable in $\MALL^1$ if it admits a coalescent slice net $\linktree$ whose coalescence path only containing steps $\lpar$, $\lplus$, $\exists$, $\forall$, $\ltens$, $\lwith$, $\conf$, and $\bullet$.
\end{corollary}

\begin{figure}
	\centering\adjustbox{max width=\textwidth}{$
		\begin{array}{r@{=}l@{\qquad}r@{=}l}
			\cnof{
				\vlderivation{\vlin{\lunit}{}{\sdash \lunit}{\vlhy{}}}
			}
			&
			\biggl\{\vpz1{\lunit}\biggr\}
			\pzlinks{pz1/pz1/16/\la/red/}
		&
			%
			\cnof{
				\vlinf{\axrule}{}{\sdash \lsend xy, \lrecv xy}{\vlhy{}}
			}
			&
			\biggl\{\vpz1{\lsend xy},\vpz2{ \lrecv xy}\biggr\}
			\pzlinks{pz1/pz2/16/\la/red/}
		\\[24pt]
			%
			\cnof{
				\vlderivation{
					\vlin{\nqpopr}{}{\sdash[\Sigma, \isna y] \Gamma, \lnNa xA}{
						\vlpr{\dD_1}{}{\sdash \Gamma, A\fsubst yx}
					}
				}
			}
			&
			\canon{
				\cnof{\dD_1}\fsubst yx
				\conc
				\biggl\{\vx1\;,\;\vy1\biggr\}
				\pzlinks{x1/y1/16/\la/red/}
			}
		&
			%
			\cnof{
				\vlderivation{
					\vlin{\exists}{}{\sdash \Gamma, \lEx xA}{
						\vlpr{\dD_1}{}{\sdash \Gamma, A\fsubst yx}
					}
				}
			}
			&
			\cnof{\dD_1}\fsubst yx
		\\[24pt]
			%
			\cnof{
				\vlderivation{\vlin{\rrule[1]}{}{\sdash \Gamma}{\vlpr{\dD_1}{}{\sdash \Gamma_1}}}
			}
			&
			\cnof{\dD_1}
		&
			%
			\cnof{
				\vlderivation{\vlin{\naloadr}{}{\sdash \Gamma,\lNa xA}{\vlpr{\dD_1}{}{\sdash[\sS,\isna x] \Gamma, A }}}
			}
			&
			\cnof{\dD_1}
		\\[24pt]
			%
			\cnof{
				\vlderivation{
					\vliin{\rrule[2]}{}{\sdash[\sS_1,\sS_2] \Gamma}{
						\vlpr{\dD_1}{}{\sdash[\sS_1] \Gamma_1}
					}{
						\vlpr{\dD_2}{}{\sdash[\sS_2] \Gamma_2}
					}
				}
			}
			&
			\canon{\cnof{\dD_1} {\conc} \cnof{\dD_2}}
		&
			%
			\cnof{
				\vlderivation{
					\vliin{\lwith}{}{\sdash \Gamma}{
						\vlpr{\dD_1}{}{\sdash \Gamma_1}
					}{
						\vlpr{\dD_2}{}{\sdash \Gamma_2}
					}
				}
			}
			&
			\canon{\cnof{\dD_1} {\conf} \cnof{\dD_2}}
		\\[24pt]
		\multicolumn{4}{c}{
			\mbox{with }
			\rrule[1]\in\set{\lpar,\oplus,\forall,\nuurule,\yurule}
			\quand
			\rrule[2]\in\set{\ltens,\lprec,\precur,\mixr}
			\quand
			\nabla\in\set{\lnewsymb,\lyasymb}
		}
	\end{array}$}
	\caption{
		Translation of a derivation in $\NL$ into a conflict net,
		where $\cnof{\dD}\fsubst yx$ is the \cotree obtained by applying the substitution $\fsubst yx$ to all its links in $\cnof{\dD}$.
	}
	\label{fig:deseq}
\end{figure}

\begin{figure}[t]
	\centering\adjustbox{max width=\textwidth}{$\begin{array}{c|c}
		\begin{array}[t]{c|c|c}
			\dD_{\la}& \mbox{step} & \dD_{\lc}
		\\\hline
		 	\vlderivation{
		 		\vlpr{\pi}{}{\dualizerof[{\la}](\sdash A, B, \Gamma)}
	 		}
		 	&\lcoalto{\lpar}&
		 	\vlderivation{
		 		\vlin{\lpar}{}{
					\dualizerof[{\lc}](\sdash A \lpar B, \Gamma)
				}{
		 			\vlpr{\pi}{}{
						\dualizerof[{\la}](\sdash A, B, \Gamma)
					}
	 			}
		 	}
 		\\
	 		\vlderivation{
				\vlpr{\pi }{}{\dualizerof[{\la}](\sdash A_i, \Gamma)}
			}
	 		&\lcoalto{\oplus}&
	 		\vlderivation{
	 			\vlin{\lpar}{}{
					\dualizerof[{\lc}](\sdash A_1 \oplus A_2, \Gamma)
				}{
	 				\vlpr{\pi }{}{\dualizerof[{\la}](\sdash A_i , \Gamma)}
	 			}
	 		}
 		\\
 			\vlderivation{
 				\vlpr{\pi}{}{\dualizerof[{\la}](\sdash A , \Gamma)}
 			}
 			&\lcoalto{\forall}&
 			\vlderivation{
				\vlin{\forall}{}{
					\dualizerof[{\lc}](\sdash \lFa xA , \Gamma)
 				}{
					\vlpr{\pi}{}{\dualizerof[{\la}](\sdash A , \Gamma)}
				}
 			}
		\\
			\vlderivation{
				\vlpr{\pi}{}{\dualizerof[{\la}](\sdash A , \Gamma)}
			}
			&\lcoalto{\naur}&
			\vlderivation{
			   \vlin{\naur}{}{
				   \dualizerof[{\lc}](\sdash \lNa xA , \Gamma)
				}{
				   \vlpr{\pi}{}{\dualizerof[{\la}](\sdash A , \Gamma)}
			   }
			}
	   \\
			\vlderivation{
				\vlpr{\pi}{}{\dualizerof[{\la}](\sdash A\fsubst yx , \Gamma)}
			}
			&\lcoalto{\exists}&
			\vlderivation{
				\vlin{\exists}{}{
						\dualizerof[{\lc}](\sdash \lEx xA, \Gamma)
					}{
					\vlpr{\pi}{}{
						\dualizerof[{\la}](\sdash A\fsubst yx , \Gamma)
					}
				}
			}
		\\
 			\vlderivation{
 				\vlpr{\pi}{}{\dualizerof[{\la}](\sdash A\fsubst yx , \Gamma)}
 			}
 			&\lcoalto{\napopr}&
 			\vlderivation{
				\vlin{\napopr}{}{
 					 \dualizerof[{\lc}]( \sdash[\sS,\isna x]  \lnNa xA, \Gamma)
 				}{
					\vlpr{\pi}{}{\dualizerof[{\la}](\sdash A\fsubst yx, \Gamma)}
				}
 			}
 		\\
			\vlderivation{
				\vlpr{\pi}{}{\dualizerof[{\la}](\sdash[\sS,\isna x] A , B, \Gamma)}
			}
			&\lcoalto{\naloadr}&
			\vlderivation{\vlin{\naloadr}{}{
					\dualizerof[{\lc}](\sdash \lNa xA , \Gamma)
				}{
					\vlpr{\pi}{}{
						\dualizerof[{\la}](\sdash[\sS,\isna x] A , \Gamma)
					}
				}
			}
		\end{array}
	&
		\begin{array}[t]{c|c|c|c}
			\dD_{\la}& \dD_{\lb}&  \mbox{step} & \dD_{\lc}
		\\\hline
			\vlderivation{\vlpr{\pi_1}{}{\dualizerof[{\la}](\sdash[\sS_1] A, \Gamma)}}
		&
			\vlderivation{\vlpr{\pi_2}{}{\dualizerof[{\lb}](\sdash[\sS_2] B, \Delta)}}
		&\lcoalto{\ltens}&
			\vlderivation{
				\vliin{\ltens}{}{
					\dualizerof[{\lc}](\sdash[\sS_1,\sS_2]A \ltens B, \Gamma,\Delta)
				}{
					\vlpr{\pi_1}{}{\dualizerof[{\la}](\sdash[\sS_1]A, \Gamma)}
				}{
					\vlpr{\pi_2}{}{\dualizerof[{\lb}](\sdash[\sS_2]B, \Delta)}
				}
			}
		\\
			\vlderivation{\vlpr{\pi_1}{}{\dualizerof[{\la}](\sdash[\sS_1] A, \Gamma)}}
		&
			\vlderivation{\vlpr{\pi_2}{}{\dualizerof[{\lb}](\sdash[\sS_2] B, \Delta)}}
		&\lcoalto{\precur}&
			\vlderivation{
				\vliin{\precur}{}{
					\dualizerof[{\lc}](\sdash[\sS_1,\sS_2]A \lprec B, \Gamma,\Delta)
				}{
					\vlpr{\pi_1}{}{\dualizerof[{\la}](\sdash[\sS_1]A, \Gamma)}
				}{
					\vlpr{\pi_2}{}{\dualizerof[{\lb}](\sdash[\sS_2]B, \Delta)}
				}
			}
		\\
			\vlderivation{\vlpr{\pi_1}{}{\dualizerof[{\la}](\sdash[\sS_1]A_1,A_2 \Gamma)}}
			&
			\vlderivation{\vlpr{\pi_2}{}{\dualizerof[{\lb}](\sdash[\sS_2]B_1,B_2 \Delta)}}
			&\lcoalto{\lprec}&
			\vlderivation{
				\vliin{\lprec}{}{
					\dualizerof[{\lc}](\sdash[\sS_1,\sS_2]A_1 \lprec B_1,A_2 \lprec B_2, \Gamma,\Delta)
				}{
					\vlpr{\pi_1}{}{\dualizerof[{\la}](\sdash[\sS_1]A_1,A_2, \Gamma)}
				}{
					\vlpr{\pi_2}{}{\dualizerof[{\lb}](\sdash[\sS_2]B_1,B_2, \Delta)}
				}
			}
		\\
			\vlderivation{\vlpr{\pi_1}{}{\dualizerof[{\la}](\sdash[\sS_1]\Gamma)}}
			&
			\vlderivation{\vlpr{\pi_2}{}{\dualizerof[{\lb}](\sdash[\sS_2]\Delta)}}
			&\lcoalto{\mixr}&
			\vlderivation{
				\vliin{\mixr}{}{
					\dualizerof[{\lc}](\sdash[\sS_1,\sS_2]\Gamma,\Delta)
				}{
					\vlpr{\pi_1}{}{\dualizerof[{\la}](\sdash[\sS_1]\Gamma)}
				}{
					\vlpr{\pi_2}{}{\dualizerof[{\lb}](\sdash[\sS_2]\Delta)}
				}
			}
		\\
			\vlderivation{\vlpr{\pi_1}{}{\dualizerof[{\la}](\sdash[\sS]A, \Gamma)}}
			&
			\vlderivation{\vlpr{\pi_2}{}{\dualizerof[{\lb}](\sdash[\sS]B, \Delta)}}
			&\lcoalto{\lwith}&
			\vlderivation{
				\vliin{\lwith}{}{
					\dualizerof[{\lc}](\sdash[\sS] A \lwith B, \Gamma,\Delta)
				}{
					\vlpr{\pi_1}{}{\dualizerof[{\la}](\sdash[\sS]A, \Gamma)}
				}{
					\vlpr{\pi_2}{}{\dualizerof[{\lb}](\sdash[\sS]B, \Delta)}
				}
			}
		\end{array}
	\\
		\mbox{with }
		\dualizerof[{\lc}]=
		\begin{cases}
			\dualizerof[{\la}]\fsubminus x
			&
			\mbox{for steps $\exists$ and $\napopr$,}
			\\
			\dualizerof[{\la}]
			&
			\mbox{otherwise}
		\end{cases}
	&
		\mbox{with }
		\dualizerof[{\lc}]=
		\begin{cases}
			\dualizerof[{\la}]\join\dualizerof[{\lb}]
		&
			\mbox{for step $\lwith$}
		\\
			\dualizerof[{\la}]\duasum\dualizerof[{\lb}]
		&
			\mbox{otherwise}
		\end{cases}
	\end{array}$}
	\caption{
		Effect of coalescence steps in \Cref{fig:coalescence} on \cotrees with leaves labeled by derivations.
		The steps $\bullet$ and $\conf$ change no link labels.
	}
	\label{fig:coacot}
\end{figure}

As in \cite{hei:hug:conflict}, we define the \defn{size} of a proof net $\linktree$ on $\Gamma$ as the number $\sizeof{\linktree}$ of nodes in the \cotree $\linktree$ plus the number $\sizeof\Gamma$ of nodes in the forest $\Gamma$.
\begin{proposition}\label{prop:coal}
	The coalescence criterion is polynomial in the size of the proof net.
\end{proposition}
\begin{proof}
	The result follows from the same argument (and algorithms) used in \cite{hei:hug:conflict} in the proof of the similar result for $\MALL$.
	The new multiplicative coalescence steps (the ones involving the $\lprec$) are as complex as the $\ltens$.
	Coalescence steps involving quantifiers require to perform operations on dualizers which are linear in the size of the dualizer (see \cite{mar:mon:unification}), and the size of the dualizer is linear in the size of the formula.
	Thus the complexity is at most $\mathcal O(n^5)$ where $n$ is the size of the proof net.
\end{proof}

\subsection{Slice Nets}\label{subsec:SN}

In this paper we define slice nets using the correctness criterion from \cite{hughes:conflict} based on \emph{erasing steps} rather than on slicing and switchings as in \cite{hug:van:slice,hug:van:sliceLICS}.
The criterion presented here is more similar to the criterion we used for conflict nets, it requires fewer definitions to be stated, and it is polynomial on the size of of a slice net.

\def\gG{\mathcal G}
\def\gof#1{\gG\left(#1\right)}
\newcommand{\sgof}[2][\sS]{\gG\left(#1\vdash #2\right)}
\newcommand{\lgof}[3][\sS]{\gG\left(#1\vdash #2::#3\right)}
\def\linkings{\Lambda\!\!\!\Lambda}
\def\emptygraph{\emptyset}
\newcommand{\erasto}[1][]{\rightarrow^{#1}_{\mathsf{er}}}
\newcommand{\derasto}[1][]{\downarrow_{\mathsf{er}}\mbox{ via } #1}
\def\deferrule#1#2#3{\begin{array}{c}#2\\\derasto[#1]\\#3\end{array}}
\newcommand{\erastos}{\erasto[*]}
\def\eof#1{\mathsf{E}(#1)}
\def\vof#1{\mathsf{V}(#1)}
\def\restr#1{|_{#1}}
\begin{definition}
	The \defn{parse graph}\footnote{Indeed, parse graphs are always forests.} of a judgement $\sdash \Gamma$ is defined as the graph $\sgof{\Gamma}$ with:
	\begin{itemize}
		\item set of vertices $\vof{\sdash \Gamma}$ contains the store $\sS$, and the set of occurrences of atoms, connectives, quantifiers, and bound variables in the sequent $\Gamma$; and
		\item set of edges $\eof{\sdash \Gamma}$ containing an edge $\set{v,w}$ whenever $v$ is the main operator of a formula $A=A_1\circleddot A_2$ \resp{$A=\lQu xA_1$} and $w$ is the main operator of the $A_1$ or $A_2$ \resp{of the subformula $A_1$ or the variable $x$}.
	\end{itemize}
	The parse graph of a sequent $\Gamma$ or a formula $A$ is defined similarly.

	If $\linkings$ is a set of axiomatic linkings on $\sdash \Gamma$, then the \defn{linked parse graph} is the graph $\lgof\Gamma\linkings$ obtained by adding to $\sgof\Gamma$ the set of edges $\eof\linkings$ containing an edge between each pair of vertices corresponding to a pair of atoms or variables in link in one of the linking in $\linkings$ -- therefore, no edge has to be considered for a link of the form $\set\lunit$. 
\end{definition}

\begin{figure}[t]
	\centering
	\adjustbox{max width=.9\textwidth}{$\begin{array}{c}
		\begin{array}{c|l}
			\mbox{Erasing rule}& \mbox{Side conditions}
		\\\hline\hline
			\deferrule{\axrule}{
				\lgof{\lsend xy,\lrecv zt}{\set{\la}}
			}{
				~
			}
			&
			\la=\set{\lsend xy,\lrecv zt}
			\qwith
			\dualizerof[\la]^{\set{\la}}=\fsubsts{y/x,t/z}
		\\\hline
			\deferrule{\lunit}{
				\lgof{\lunit}{\set{\emptyset}}
			}{
				~
			}
			&
		\\\hline
			\deferrule{\lpar}{
				\lgof{\Gamma, A\lpar B}{\linkings}
			}{
				\lgof{\Gamma, A, B}{\linkings}
			}
			&
		\\\hline
			\deferrule{\mixr}{
				\lgof[\sS_1,\sS_2]{\Gamma,\Delta}{\linkings}
			}{
				\lgof[\sS_1]{\Gamma}{\linkings\restr{\sdash[\sS_1]\Gamma}}
				,
				\lgof[\sS_2]{\Delta}{\linkings\restr{\sdash[\sS_2]\Delta}}
			}
			&
			\begin{tabular}{l}
				if there is no $\set{v,w}\in \eof\linkings$ such that
				\\
				$v\in\vof{\sdash[\sS_1]\Gamma}$ and $w\in\vof{\sdash[\sS_2]\Delta}$
			\end{tabular}
		\\\hline
			\deferrule{\lprec}{
				\lgof[\sS_1,\sS_2]{\Gamma,\Delta, A\lprec B,C\lprec D}{\linkings}
			}{
				\lgof[\sS_1]{\Gamma, A,C}{\linkings\restr{\sdash[\sS_1]\Gamma,A}}
				,
				\lgof[\sS_2]{\Delta, B,D}{\linkings\restr{\sdash[\sS_2]\Delta,B}}
			}
			&
			\begin{tabular}{l}
				$\lgof[\sS_1,\sS_2]{\Gamma,\Delta, A\lprec B,C\lprec D}{\linkings}$ is connected and \\
				$\lgof[\sS_1,\sS_2]{\Gamma,\Delta, A, B,C, D}{\linkings}$ is not connected
			\end{tabular}
		\\\hline
			\deferrule{\precur}{
				\lgof[\sS_1,\sS_2]{\Gamma,\Delta, A\lprec B}{\linkings}
			}{
				\lgof[\sS_1]{\Gamma, A}{\linkings\restr{\sdash[\sS_1]\Gamma,A}}
				,
				\lgof[\sS_2]{\Delta, B}{\linkings\restr{\sdash[\sS_2]\Delta,B}}
			}
			&
			\begin{tabular}{l}
				$\lgof[\sS_1,\sS_2]{\Gamma,\Delta, A\lprec B}{\linkings}$ is connected and \\
				$\lgof[\sS_1,\sS_2]{\Gamma,\Delta, A, B}{\linkings}$ is not connected
			\end{tabular}
		\\\hline
			\deferrule{\ltens}{
				\lgof[\sS_1,\sS_2]{\Gamma,\Delta, A\ltens B}{\linkings}
			}{
				\lgof[\sS_1]{\Gamma, A}{\linkings\restr{\sdash[\sS_1]\Gamma,A}}
				,
				\lgof[\sS_2]{\Delta, B}{\linkings\restr{\sdash[\sS_2]\Delta,B}}
			}
			&
			\begin{tabular}{l}
				$\lgof[\sS_1,\sS_2]{\Gamma,\Delta, A\ltens B}{\linkings}$ is connected and \\
				$\lgof[\sS_1,\sS_2]{\Gamma,\Delta, A, B}{\linkings}$ is not connected
			\end{tabular}
		\\\hline
			\deferrule{\lplus}{
				\lgof{\Gamma, A_1\lplus A_2}{\linkings}
			}{
				\lgof{\Gamma, A_i}{\linkings}
			}
			&
			\mbox{if no vertex in $\vof{A_{1-i}}$ occurs in $\eof\linkings$ for $i\in\set{1,2}$}
		\\\hline
			\deferrule{\lwith}{
				\lgof{\Gamma, A_1\lwith A_2}{\linkings_1\uplus\linkings_2}
			}{
				\lgof{\Gamma, A_1}{\linkings_1}
				,
				\lgof{\Gamma, A_2}{\linkings_2}
			}
			&
			\mbox{if no vertex in $\vof{A_{1-i}}$ occurs in $\eof{\linkings_i}$ for $i\in\set{1,2}$}
		\\\hline
			\deferrule{\exists}{
				\lgof{\Gamma, \lEx xA}{\linkings}
			}{
				\lgof{\Gamma, A\fsubst yx}{\linkings}
			}
			&
			\begin{tabular}{l}
				if $x$ occurs in a $v\in\vof{A}$ and $v\in\la\in\linking\in\linkings$,
				\\
				then $\dualizerof[\la]^{\linkings}(x)=\dualizerof[\la]^{\linkings}(y)$
			\end{tabular}
		\\\hline
			\deferrule{\forall}{
				\lgof{\Gamma, \lFa xA}{\linkings}
			}{
				\lgof{\Gamma, A}{\linkings}
			}
			&
			\mbox{if } x\notin \freeof{\Gamma}
		\\\hline
			\deferrule{\naur}{
				\lgof{\Gamma, \lNa xA}{\linkings}
			}{
				\lgof{\Gamma, A}{\linkings}
			}
			&
			\begin{tabular}{l}
				if $x\notin \freeof{\Gamma}$
				\quand
				\\
				$x$ does not occur in a nominal link in any $\linking\in\linkings$
			\end{tabular}
		\\\hline
			\deferrule{\naloadr}{
				\lgof{\Gamma, \lNa xA}{\linkings}
			}{
				\lgof[\sS,\isna x]{\Gamma, A}{\linkings\fsubst{\set{\isna x,y}}{\set{x,y}}}
			}
			&
			\mbox{if } x\notin \freeof{\Gamma}
			\mbox{ and } \set{x,y} \in \eof{\linkings}
		\\\hline
			\deferrule{\napopr}{
				\lgof[\sS ,\isna y]{\Gamma, \lnNa xA}{\linkings}
			}{
				\lgof{\Gamma, A\fsubst ya}{\linkings\setminus\la}
			}
			&
			\begin{tabular}{l}
				$\la=\set{x,\isna y} \in \eof{\linkings}$
				\\
				where $\linkings\setminus\la\coloneqq \set{\linking\setminus\set\la \mid \linking\in\linkings}$
			\end{tabular}
		\end{array}
	\end{array}$}

	\caption{Erasing steps for linked parse graphs.}
	\label{fig:erasing}
\end{figure}

\begin{nota}
	Let $\linkings=\set{\linking_i}_{i\in I}$ be a set of atomic linkings for a judgement $\sdash \Gamma$.
	If $\sdash[\sS']\Gamma'$ is a sub-judgement of $\sdash \Gamma$, then we define $\linkings\restr{\sdash[\sS'] \Gamma'}$ as the set of non-empty linkings $\set{\linking'_i}_{j\in J}$ (for a $J\subseteq I$) such that each linking $\linking_i'$ contains all links in $\linking_i$ which are pairs of vertices in $\vof{\sdash[\sS'] \Gamma'}$,
	that is,
	$
	\linkings\restr{\sdash[\sS'] \Gamma'}=\SET{\linking_i\restr{\sdash[\sS'] \Gamma'}\coloneqq\set{\lambda\in\linking_i\mid \lambda\subseteq\vof{\sdash[\sS'] \Gamma'}}\mid i\in I, \; \linking_i\restr{\sdash[\sS'] \Gamma'}\neq \emptyset}
	$.
\end{nota}
\begin{definition}
	In  \Cref{fig:erasing} we define \defn{erasing steps} as the rewriting rules over sets of linkings with witnesses%
	\footnote{
		We represent erasing steps without explicitly writing the context, that is, if $L\erasto R_1,\ldots,R_k$ is a rule, then it can be applied in a set in such a way $\set{L,x_1,\ldots, x_n}$ can be rewritten as $\set{R_1\ldots,R_k,x_1,\ldots, x_n}$.
	}%
	.
	We say that a set of axiomatic linkings $\linkings$ on a sequent $\Gamma$ is \defn{erasable} if there is a sequence of erasing steps (called \defn{erasing path}) starting from the singleton $\set{\lgof[\emptystore]\Gamma{\linkings,\dualizerof^{\linkings}}}$ and ending with the singleton containing $\emptygraph=\lgof[\emptystore]\emptyset{\emptyset,\dualizerof[\emptyset]}$.

	A \defn{slice net} for $\Gamma$ is an erasable  set of axiomatic linkings $\linkings$.
\end{definition}
\begin{remark}
	The side conditions of our erasing steps ensures that the connective(s) removed by the step is ``ready'' to be sequentialized, similarly to how the readiness is defined in \cite{hughes:conflict} for the slice nets for $\MALL$.
\end{remark}

In \Cref{fig:erasingExample} we show an erasing path for the slice net from \Cref{fig:introPN} in the introduction.
\begin{figure}[t]
	\centering
	\adjustbox{max width=.8\textwidth}{$\begin{array}{c}
		\Set{
			\lgof[]{
				\lNu{\vx1}{
					\lNup {\vy1}{
						\rclr{\lsend ya}
						\lpar
						\lEx a{\rclr{\lrecv ya}}
						\lpar
						\begin{pmatrix}
							\gclr{\lsend x{\lab_1}} \lprec \lEx b{\bclr{\lrecv xb}}
							\\\lwith\\
							\gclr{\lsend x{\lab_2}} \lprec \vclr{\lsend xc}
						\end{pmatrix}
						\lpar
						\begin{pmatrix}
							\gclr{\lrecv x{\lab_1}} \lprec \bclr{\lsend xb}
							\\\lplus\\
							\gclr{\lrecv x{\lab_2}} \lprec \lEx c{\vclr{\lrecv xc}}
						\end{pmatrix}
					}
				}
			}{
				\linkings
			}
		}
	\\
		\derasto[(2\times\nuur)+(2\times\lpar)]
	\\
		\Set{
			\lgof[]{
				\rclr{\lsend ya}
				,
				\rclr{\lrecv ya}
				,
				\begin{pmatrix}
					\gclr{\lsend x{\lab_1}} \lprec \lEx b{\bclr{\lrecv xb}}
					\\\lwith\\
					\gclr{\lsend x{\lab_2}} \lprec \vclr{\lsend xc}
				\end{pmatrix}
				,
				\begin{pmatrix}
					\gclr{\lrecv x{\lab_1}} \lprec \bclr{\lsend xb}
					\\\lplus\\
					\gclr{\lrecv x{\lab_2}} \lprec \lEx c{\vclr{\lrecv xc}}
				\end{pmatrix}
			}{
				\linkings
			}
		}
	\\
		\derasto[\mixr]
	\\
		\SET{
			\lgof[]{
				\rclr{\lsend ya}
				,
				\rclr{\lrecv ya}
			}{
				\Set{\set{\rclr{\lsend ya},\rclr{\lrecv ya}}}
			}
			,
		\\
			\lgof[]{
				\begin{pmatrix}
					\gclr{\lsend x{\lab_1}} \lprec \lEx b{\bclr{\lrecv xb}}
					\\\lwith\\
					\gclr{\lsend x{\lab_2}} \lprec \vclr{\lsend xc}
				\end{pmatrix}
				,
				\begin{pmatrix}
					\gclr{\lrecv x{\lab_1}} \lprec \bclr{\lsend xb}
					\\\lplus\\
					\gclr{\lrecv x{\lab_2}} \lprec \lEx c{\vclr{\lrecv xc}}
				\end{pmatrix}
			}{
				\begin{Bmatrix}
					\begin{Bmatrix}
						\set{\bclr{\lsend xb},\bclr{\lrecv xb}},	\\
						\set{\gclr{\lsend x{\lab_1}},\gclr{\lrecv x{\lab_1}}}
					\end{Bmatrix}\;,
					\\
					\begin{Bmatrix}
						\set{\vclr{\lsend xc,\vclr{\lrecv xc}}},	\\
						\set{\gclr{\lsend x{\lab_2}},\gclr{\lrecv x{\lab_2}}}
					\end{Bmatrix}\;\;\;
				\end{Bmatrix}
			}
		}
	\\
		\derasto[\axrule]
	\\
		\Set{
			\lgof[]{
				\begin{pmatrix}
					\gclr{\lsend x{\lab_1}} \lprec \lEx b{\bclr{\lrecv xb}}
					\\\lwith\\
					\gclr{\lsend x{\lab_2}} \lprec \vclr{\lsend xc}
				\end{pmatrix}
				,
				\begin{pmatrix}
					\gclr{\lrecv x{\lab_1}} \lprec \bclr{\lsend xb}
					\\\lplus\\
					\gclr{\lrecv x{\lab_2}} \lprec \lEx c{\vclr{\lrecv xc}}
				\end{pmatrix}
			}{
				\begin{Bmatrix}
					\begin{Bmatrix}
						\set{\bclr{\lsend xb},\bclr{\lrecv xb}},	\\
						\set{\gclr{\lsend x{\lab_1}},\gclr{\lrecv x{\lab_1}}}
					\end{Bmatrix}\;,
					\\
					\begin{Bmatrix}
						\set{\vclr{\lsend xc,\vclr{\lrecv xc}}},	\\
						\set{\gclr{\lsend x{\lab_2}},\gclr{\lrecv x{\lab_2}}}
					\end{Bmatrix}\;\;\;
				\end{Bmatrix}
			}
		}
	\\
		\derasto[\lwith+(2\times \lplus)]
	\\
		\SET{
			\lgof[]{
				\gclr{\lsend x{\lab_1}} \lprec \lEx b{\bclr{\lrecv xb}}
				,
				\gclr{\lrecv x{\lab_1}} \lprec \bclr{\lsend xb}
			}{
				\begin{Bmatrix}
					\begin{Bmatrix}
						\set{\bclr{\lsend xb},\bclr{\lrecv xb}},	\\
						\set{\gclr{\lsend x{\lab_1}},\gclr{\lrecv x{\lab_1}}}
					\end{Bmatrix}
				\end{Bmatrix}
			}
			\;,
			\\
			\lgof[]{
				\gclr{\lsend x{\lab_2}} \lprec \vclr{\lsend xc}
				,
				\gclr{\lrecv x{\lab_2}} \lprec \lEx c{\vclr{\lrecv xc}}
			}{
				\begin{Bmatrix}
					\begin{Bmatrix}
						\set{\vclr{\lsend xc,\vclr{\lrecv xc}}},
						\\
						\set{\gclr{\lsend x{\lab_2}},\gclr{\lrecv x{\lab_2}}}
					\end{Bmatrix}
				\end{Bmatrix}
			}
		}
	\\
		\derasto[2\times\lprec]
	\\
		\SET{
			\lgof[]{
				\gclr{\lsend x{\lab_1}}
				,
				\gclr{\lrecv x{\lab_1}}
			}{
				\begin{Bmatrix}
					\begin{Bmatrix}
						\set{\bclr{\lsend xb},\bclr{\lrecv xb}}
					\end{Bmatrix}
				\end{Bmatrix}
			}
			\;,
			\\
			\lgof[]{
				\lEx b{\bclr{\lrecv xb}}
				,
				\bclr{\lsend xb}
			}{
				\begin{Bmatrix}
					\begin{Bmatrix}
						\set{\gclr{\lsend x{\lab_1}},\gclr{\lrecv x{\lab_1}}}
					\end{Bmatrix}
				\end{Bmatrix}
			}
			\;,
			\\
			\lgof[]{
				\gclr{\lsend x{\lab_2}}
				,
				\gclr{\lrecv x{\lab_2}}
			}{
				\begin{Bmatrix}
					\begin{Bmatrix}
						\set{\gclr{\lsend x{\lab_2}},\gclr{\lrecv x{\lab_2}}}
					\end{Bmatrix}
				\end{Bmatrix}
			}
			\;,
			\\
			\lgof[]{
				\vclr{\lsend xc}
				,
				\lEx c{\vclr{\lrecv xc}}
			}{
				\begin{Bmatrix}
					\begin{Bmatrix}
						\set{\vclr{\lsend xc,\vclr{\lrecv xc}}},
					\end{Bmatrix}
				\end{Bmatrix}
			}
		}
	\\
	\derasto[4\times\axrule]
	\\
		\SET{\set{\lgof[]\emptygraph\emptyset}}
	\end{array}$}
	\caption{A possible erasing path for the slice net from \Cref{fig:introPN}.}
	\label{fig:erasingExample}
\end{figure}

\begin{theorem}
	Let $\Gamma$ be a non-empty sequent.
	Then $\proves[\NL]\Gamma$ iff there is a slice net $\linkings$ on $\Gamma$.
\end{theorem}
\begin{proof}
	In \Cref{fig:sliceDeseq} we provide translation from derivations to slice nets defined inductively on the structure of the derivation.
	The obtained set of axiomatic linkings is erasable by definition, since each inductive step of the translation preserve erasability.

	As for conflict nets, sequentialization for slice nets follows from the existence of a derivation constructed from a given erasing path, since each erasing step corresponds to an application of a sequent rule (with the same name).
\end{proof}
\begin{corollary}[Slice nets for $\MALL^1$]\label{cor:sliceForMALL1}
	A sequent $\Gamma$ is derivable in $\MALL^1$ if it admits an erasable slice net $\linktree$ whose erasing path only containing steps $\axrule$, $\lpar$, $\lplus$, $\exists$, $\forall$, $\ltens$, and $\lwith$.
\end{corollary}

\begin{figure}[t]
	\centering\adjustbox{max width=\textwidth}{$
		\begin{array}{r@{=}l@{\qquad}r@{=}l}
			\snof{
				\vlderivation{\vlin{\lunit}{}{\sdash \lunit}{\vlhy{}}}
			}
			&
			\Set{\biggl\{\vpz1{\lunit}\biggr\}}
			\pzlinks{pz1.center/pz1.center/16/\la/red/}
		&
			%
			\snof{
				\vlinf{\axrule}{}{\sdash \lsend xy, \lrecv xy}{\vlhy{}}
			}
			&
			\Set{\biggl\{\vpz1{\lsend xy},\vpz2{ \lrecv xy}\biggr\}}
			\pzlinks{pz1/pz2/12/\la/red/}
		\\[24pt]
			%
			\snof{
				\vlderivation{
					\vlin{\nqpopr}{}{\sdash[\Sigma, \isna y] \Gamma, \lnNa xA}{
						\vlpr{\dD_1}{}{\sdash \Gamma, A\fsubst yx}
					}
				}
			}
			&
			\Set{
				\linking \cup
				\Set{\biggl\{\vx1\;,\;\vy1\biggr\}}
				\pzlinks{x1/y1/12/\la/red/}
				\mid
				\linking \in \snof{\dD_1}\fsubst yx
			}
		&
			%
			\snof{
				\vlderivation{
					\vlin{\exists}{}{\sdash \Gamma, \lEx xA}{
						\vlpr{\dD_1}{}{\sdash \Gamma, A\fsubst yx}
					}
				}
			}
			&
			\snof{\dD_1}\fsubst yx
		\\[24pt]
			%
			\snof{
				\vlderivation{\vlin{\rrule[1]}{}{\sdash \Gamma}{\vlpr{\dD_1}{}{\sdash \Gamma_1}}}
			}
			&
			\snof{\dD_1}
		&
			%
			\snof{
				\vlderivation{\vlin{\naloadr}{}{\sdash \Gamma,\lNa xA}{\vlpr{\dD_1}{}{\sdash[\sS,\isna x] \Gamma, A }}}
			}
			&
			\snof{\dD_1}
		\\[24pt]
			%
			\snof{
				\vlderivation{
					\vliin{\rrule[2]}{}{\sdash[\sS_1,\sS_2] \Gamma}{
						\vlpr{\dD_1}{}{\sdash[\sS_1] \Gamma_1}
					}{
						\vlpr{\dD_2}{}{\sdash[\sS_2] \Gamma_2}
					}
				}
			}
			&
			\set{\linking_1\cup\linking_2\mid
				\linking_1\in\snof{\dD_1} ,
				\linking_2\in\snof{\dD_2}
			}
		&
			%
			\snof{
				\vlderivation{
					\vliin{\lwith}{}{\sdash \Gamma}{
						\vlpr{\dD_1}{}{\sdash \Gamma_1}
					}{
						\vlpr{\dD_2}{}{\sdash \Gamma_2}
					}
				}
			}
			&
			\snof{\dD_1} \cup \snof{\dD_2}
		\\[24pt]
		\multicolumn{4}{c}{
			\mbox{with }
			\rrule[1]\in\set{\lpar,\oplus,\forall,\nuurule,\yurule}
			\quand
			\rrule[2]\in\set{\ltens,\lprec,\precur,\mixr}
			\quand
			\nabla\in\set{\lnewsymb,\lyasymb}
		}
	\end{array}$}
	\caption{
		Translation of a derivation in $\NL$ into a slice net,
		where $\cnof{\dD}\fsubst yx$ is the set of linkings obtained by applying the substitution $\fsubst yx$ to all its links in them, and by letting the dualizers $\dualizerof[\la]$ in $\snof{\dD}$ being $\dualizerof[\la]\fsubst yx$ in $\snof{\dD}\fsubst yx$.
	}
	\label{fig:sliceDeseq}
\end{figure}

\begin{corollary}
	The erasing criterion is polynomial in the size of the slice net.
\end{corollary}
\begin{proof}
	The length of a successful erasing path is bound by the size of $\Gamma$, and checking and applying each erasing step can be performed in polinomial time  $\gof\linkings$
\end{proof}

\begin{remark}\label{rem:confEras}
	By case analysis, it is easy to see that erasing steps are locally confluent.
\end{remark}

\begin{remark}\label{rem:confToSlice}
	The proof translation from derivations to slice nets could be equivalentely defined by letting
	$\linkings(\dD)$ to be the set of axiomatic linkings on $\Gamma$ defined from the conflict net $\cnof\dD$ as follows:
	$$
	\snof{\linktree}
	=
	\begin{cases}
		\set{\la }
		&
		\mbox{if $\la$ is an axiom link}
		\\
		X_1\cup \cdots \cup X_n
		&
		\mbox{if } \dD=\conf\left(X_1,\ldots, X_n\right)
		\\
		\Set{\bigcup\limits_{x_1 \in \snof{X_1}} \cdots \bigcup\limits_{x_n\in\snof{X_n}}
		\left(x_1\cup \cdots \cup x_n\right)}
		&
		\mbox{if } \dD=\conc\left(X_1,\ldots, X_n\right)
	\end{cases}$$
	In this translation, it is clear that each conflict node in $\cnof\dD$ multiplies the elements in the set of linkings, leading to a non-polinomial time transation.
\end{remark}

\section{Canonicity results}\label{sec:canon}

\begin{figure}[t]
	\adjustbox{max width=\textwidth}{$\begin{array}{c}
		\mbox{Local rule permutations}
	\\\\
		\vlderivation{
			\vliin{\rclr{\brrule[1]}}{}{ \sdash[\sS_1, \sS_2,\sS_3] \Gamma_1,\Gamma_2,\Gamma_3,\rclr{\Theta_1}, {\Theta_2}}{
				\vlhy{\sdash[\sS_1]\Gamma_1, \rclr{\Delta_1}}
			}{
				\vliin{\bclr{\brrule[2]}}{}{
					\sdash[\sS_2,\sS_3] \Gamma_2,\Gamma_3,\rclr{\Delta_2},\bclr{\Theta_2}
				}{
					\vlhy{ \sdash[\sS_2] \Gamma_2,{\Delta_2},\bclr{\Delta_3}}}{\vlhy{ \sdash[\sS_3] \Gamma_3, \bclr{\Delta_4}}
				}
			}
		}
		\peq
		\vlderivation{
			\vliin{\bclr{\brrule[2]}}{}{
				\sdash[\sS_1, \sS_2,\sS_3]  \Gamma_1, \Gamma_2, \Gamma_3,{\Theta_1},\bclr{\Theta_2}
			}{
				\vliin{\rclr{\brrule[1]}}{}{
					\sdash[\sS_1,\sS_2] \Gamma_1,\Gamma_2, \rclr{\Theta_1}, \bclr{\Delta_2}
				}{
					\vlhy{\sdash[\sS_1] \Gamma_1,  \rclr{\Delta_1}}
				}{
					\vlhy{\sdash[\sS_3] \Gamma_2, \rclr{\Delta_2},{\Delta_3}}
				}
			}{\vlhy{\sdash[\sS_3] \Gamma_3, \bclr{\Delta_4}}}
		}
	\\\\
		\vlderivation{
			\vlin{\bclr{\urrule[2]}}{}{
				\sdash[\sS_1,\sS_2]\Gamma, {\Theta_1},\bclr{\Theta_2}}{
				\vlin{\rclr{\urrule[1]}}{}{
					\sdash[\sS_1,\sS_2']\Gamma, \rclr{\Theta_1}, \bclr{\Delta_2}
				}{\vlhy{\sdash[\sS_1',\sS_2']\Gamma, \rclr{\Delta_1},{\Delta_2}}}}
		}
		\peq
		\vlderivation{
			\vlin{\rclr{\urrule[1]}}{}{
				\sdash[\sS_1,\sS_2]\Gamma, \rclr{\Theta_1}, {\Theta_2}
			}{
				\vlin{\bclr{\urrule[2]}}{}{
					\sdash[\sS'_1,\sS_2]\Gamma,  \rclr{\Delta_1},\bclr{\Theta_2}
				}{
					\vlhy{\sdash[\sS'_1,\sS'_2]\Gamma, {\Delta_1},\bclr{\Delta_2}}
				}
			}
		}
	\qquad
		\vlderivation{
			\vlin{\rclr{\urrule}}{}{\sdash[\sS_1, \sS_2]\Gamma_1, \Gamma_2,\rclr{\Theta_1},\bclr{\Theta_2}}{
				\vliin{\bclr{\brrule}}{}{\sdash[\sS_1', \sS_2] \Gamma_1, \Gamma_2,\rclr{\Delta_1},\bclr{\Theta_2}
				}{
					\vlhy{ \sdash[\sS_1']\Gamma_1, {\Delta_1},\bclr{\Delta_2}}
				}{
					\vlhy{\sdash[\sS_2] \Gamma_2, \bclr{\Delta_3}}
				}
			}
		}
		\peq
		\vlderivation{
			\vliin{\bclr{\brrule}}{}{
				\sdash[\sS_1, \sS_2] \Gamma_1, \Gamma_2,{\Theta_1},\bclr{\Theta_2}
			}{
				\vlin{\rclr{\urrule}}{}{
					\sdash[\sS_1]\Gamma,\rclr{\Theta_1}, \bclr{\Delta_2}
				}{
					\vlhy{\sdash[\sS_1']\Gamma_1, \rclr{\Delta_1}, {\Delta_2}}
				}
			}{
				\vlhy{\sdash[\sS_2] \Gamma_2, \bclr{\Delta_3}}
			}
		}
		\\\\
		\vlderivation{
			\vliin{\rclr{\lwith}}{}{
				\sdash \Gamma,\rclr{A\lwith B}, C\lwith D
			}{
				\vliin{\bclr{\lwith}}{}{\sdash \Gamma,\rclr{A},\bclr{C\lwith D}}{
					\vlhy{\sdash \Gamma,A,\bclr{C}}
				}{
					\vlhy{\sdash \Gamma,A,\bclr{D}}
				}
			}{
				\vliin{\bclr{\lwith}}{}{\sdash \Gamma,\rclr{B},\bclr{C\lwith D}}{
					\vlhy{\sdash \Gamma,B,\bclr{C}}
				}{
					\vlhy{\sdash \Gamma,B,\bclr{D}}
				}
			}
		}
		\peq
		\vlderivation{
			\vliin{\bclr{\with}}{}{
				\sdash \Gamma, A\lwith B, \bclr{C\lwith D}
			}{
				\vliin{\rclr{\lwith}}{}{
					\sdash \Gamma, \rclr{A\lwith B}, \bclr{C},
				}{
					\vlhy{\sdash \Gamma, \rclr{B},C}
				}{
					\vlhy{\sdash \Gamma, \rclr{A},C}
				}
			}{
				\vliin{\rclr{\lwith}}{}{
					\sdash \Gamma, \rclr{A\lwith B}, \bclr{D},
				}{
					\vlhy{\sdash \Gamma, \rclr{B},D}
				}{
					\vlhy{\sdash \Gamma, \rclr{A},D}
				}
			}
		}
	\\\\
		\vlderivation{
			\vlin{\bclr{\alpha}}{}{
				\sdash[\sS] \Gamma,A\lwith B,\bclr{\Theta}
			}{
				\vliin{\rclr{\lwith}}{}{
					\sdash[\sS'] \Gamma, \rclr{A\lwith B}, \bclr{\Delta}
				}{
					\vlhy{\sdash[\sS'] \Gamma, \rclr{B},\Delta}
				}{
					\vlhy{\sdash[\sS'] \Gamma, \rclr{A},\Delta}
				}
			}
		}
		\peq
		\vlderivation{
			\vliin{\rclr{\lwith}}{}{
				\sdash \Gamma,\rclr{A\lwith B}, \Theta
			}{
				\vlin{\bclr{\alpha}}{}{
					\sdash \Gamma,\rclr{A},\bclr{\Theta}
				}{
					\vlhy{\sdash[\sS'] \Gamma,A,\bclr{\Delta}}
				}
			}{
				\vlin{\bclr{\alpha}}{}{
					\sdash \Gamma,\rclr{B},\bclr{\Theta}
				}{
					\vlhy{\sdash[\sS'] \Gamma,B,\bclr{\Delta}}
				}
			}
		}
	\qquad
		\boxed{
			\vlderivation{
				\vlin{\nuloadr}{}{
					\sdash \Gamma, \rclr{\lNu xA}, \lYa xB
				}{
					\vlin{\nupopr}{}{
						\sdash[\sS,\isnu x] \Gamma, \rclr A, \bclr{\lYa xB}
					}{\vlhy{\sdash \Gamma, A, \bclr B}}
				}
			}
			\peq
			\vlderivation{
				\vlin{\yaloadr}{}{
					\sdash \Gamma, \lNu xA, \bclr{\lYa xB}
				}{
					\vlin{\yapopr}{}{
						\sdash[\sS,\isya x] \Gamma, \rclr{\lNu xA}, \bclr B
					}{\vlhy{\sdash \Gamma, \rclr A, B}}
				}
			}
		}
	\\\\\hline\\
		\mbox{Non-local rule permutations}
	\\\\
		\vlderivation{
			\vliin{\rclr{\beta}}{}{
				\sdash[\sS_1,\sS_2] \rclr{\Gamma}, C\lwith D
			}{
				\vlpr{\dD}{}{\sdash[\sS_1] \rclr{\Gamma_1}}
			}{
				\vliin{\bclr{\lwith}}{}{
					\sdash[\sS_2] \rclr{\Gamma_2},\bclr{C\lwith D}
				}{
					\vlhy{\sdash[\sS_2] \Gamma_2,\bclr{C}}
				}{
					\vlhy{\sdash[\sS_2] \Gamma_2,\bclr{D}}
				}
			}
		}
		\quad\speq\quad
		\vlderivation{
			\vliin{\bclr{\with}}{}{
				\sdash[\sS_1,\sS_2] \Gamma,\bclr{C\lwith D}
			}{
				\vliin{\rclr{\beta}}{}{
					\sdash[\sS_1,\sS_2] \rclr{\Gamma}, \bclr{C},
				}{
					\vlpr{\dD}{}{\sdash[\sS_1] \rclr{\Gamma_1}}
				}{
					\vlhy{\sdash[\sS_2] \rclr{\Gamma_2}, C}
				}
			}{
				\vliin{\rclr{\beta}}{}{
					\sdash[\sS_1,\sS_2] \rclr{\Gamma}, \bclr{D},
				}{
					\vlpr{\dD}{}{\sdash[\sS_1] \rclr{\Gamma_1}}
				}{
					\vlhy{\sdash[\sS_2] \rclr{\Gamma_2}, D}
				}
			}
		}
	\\\\\hline\hline\\
		\urrule,\urrule[1],\urrule[2]\in\set{\lpar,\lplus,\exists,\forall,\nuurule,\yurule,\nupopr,\yapopr,\nuloadr,\yaloadr}
	\quad
		\brrule,\brrule[1],\brrule[2]\in\set{\ltens,\lprec,\precur,\mixr}
	\end{array}$}
	\caption{Rule permutations in $\NL$.}
	\label{fig:permutations1}
\end{figure}

In this section we prove that conflict nets and slice nets for $\NL$ are proof systems with two distinct canonicity properties called \emph{local} and \emph{strong} canonicity, and they both identify derivations modulo what we refer to as \emph{witness renaming}.

We first introduce three notions of equivalence for derivations in $\NL$.
\begin{definition}
	We call the variable introduced during the proof search by a quantifier rule  the \defn{active variable} of that (occurrence of) rule.
	The active variable of an existential \resp{universal} quantifier may also be called its \defn{witness}
	\footnote{
		In the general setting of first-order logic, the witness of an existential quantifier could be any term of the language. 
		However, in $\PIL$ witnesses can only be variables because of the simple structure of terms.
	} 
	\resp{\defn{eigenvariable}}, while the active formula of a nominal quantifier may also be called a \defn{fresh name}.
	Two derivations $\dD_1$ and $\dD_2$ in $\NL$ are:
	\begin{itemize}
		\item \defn{equivalent modulo active variables} (denoted $\dD_1\weq \dD_2$) if it is possible to transform $\dD_1$ into $\dD_2$ by changing the active variables of the quantifier rules (and propagating the changes upwards in the derivation);
		\item \defn{equivalent modulo rule permutations} (denoted $\dD_1\speq \dD_2$) if it is possible to transform $\dD_1$ into $\dD_2$ using all transformations in \Cref{fig:permutations1};
		\item \defn{equivalent modulo local rule permutations} (denoted $\dD_1\peq \dD_2$) if it is possible to transform $\dD_1$ into $\dD_2$ using the local rule permutations in \Cref{fig:permutations1}.
	\end{itemize}
	Moreover, we write $\dD_1\spweq \dD_2$ \resp{$\dD_1\pweq \dD_2$} if they are equivalent modulo rule permutations \resp{local rule permutations} and active variables, that is, if there are derivations $\dD_1'$ and $\dD_2''$ such that $\dD_1\weq \dD_1'\speq \dD_2'\weq \dD_2$ \resp{$\dD_1\weq \dD_1'\peq \dD_2'\weq \dD_2$}.
\end{definition}

\begin{remark}\label{rem:proofEQ}
	In \Cref{eq:EXweq1}, we show three derivations which are equivalent modulo active variables.
	\begin{equation}\label{eq:EXweq1}
		\hfill
		\vlderivation{
			\vlin{\nuloadr}{}{
				\vdash \lNu x \lsend xa, \lYa y \lrecv ya
			}{
				\vlin{\nupopr}{}{
					\sdash[\isnu x] \lsend xa, \lYa y \lrecv ya
				}{
					\vlin{\axrule}{}{
						\vdash \lsend xa, \lrecv xa
					}{\vlhy{}}
				}
			}
		}
		\;\weq\;
		\vlderivation{
			\vlin{\nuloadr}{}{
				\vdash \lNu x \lsend xa, \lYa y \lrecv ya
			}{
				\vlin{\nupopr}{}{
					\sdash[\isnu z] \lsend za, \lYa y \lrecv ya
				}{
					\vlin{\axrule}{}{
						\vdash \lsend za, \lrecv za
					}{\vlhy{}}
				}
			}
		}
		\;\weq\;
		\vlderivation{
			\vlin{\nuloadr}{}{
				\vdash \lNu x \lsend xa, \lYa y \lrecv ya
			}{
				\vlin{\nupopr}{}{
					\sdash[\isnu y] \lsend ya, \lYa y \lrecv ya
				}{
					\vlin{\axrule}{}{
						\vdash \lsend ya, \lrecv ya
					}{\vlhy{}}
				}
			}
		}
		\hfill
	\end{equation}
	In \Cref{eq:EXweq2} we show two existential quantifier rules select two distinct witnesses $x$ and $z$, but the pair of atoms linked by an axiom rule is the same.
	\begin{equation}\label{eq:EXweq2}
		\hfill
		\vlderivation{
			\vliq{\exists}{}{
				\vdash \lEx x \lsend xa, \lEx y \lrecv ya
			}{
				\vlin{\axrule}{}{\vdash \lsend xa, \lrecv xa}{\vlhy{}}
			}
		}
		\quad\weq\quad
		\vlderivation{
			\vliq{\exists}{}{
				\vdash \lEx x \lsend xa, \lEx y \lrecv ya
			}{
				\vlin{\axrule}{}{\vdash \lsend ya, \lrecv ya}{\vlhy{}}
			}
		}
		\hfill
	\end{equation}
	We could argue that these two derivations should be not identified because the choice of the witness is part of the information of the proof.
	In a boarder sense, it may be useful to not identify a proof using a very elementary witness with a proof using a quite complex one.
	However, because of the quite limited syntax of atoms (terms) in $\PIL$, witness for quanfitiers can only be varibles.
	Therefore, as soon as the choice of witness do not change the pairs of atoms which are mated by the $\axrule$-rules in the derivation, such a choice can be considered irrelevant, especially if the choice of the witness of a $\napopr$ depends on varible previously stored by a $\naloadr$, which is arbitrary because of $\alpha$-equivalence.
	Note that the two sub-derivations of the $\weq$-equivalent derivations in \Cref{eq:EXweq2} made only of the $\axrule$-rules are not $\weq$-equivalent.
	That is
	$
	\vlderivation{
		\vlin{\axrule}{}{\vdash \lsend xa, \lrecv xa}{\vlhy{}}
	}
	\;\not\weq\;
	\vlderivation{
		\vlin{\axrule}{}{\vdash \lsend ya, \lrecv ya}{\vlhy{}}
	}
	$.

	However, the choice of active variables may change the pair of atoms linked by the $\axrule$-rules.
	For an example, see \Cref{eq:EXweq3} below, where we show two non $\weq$-equivalent derivations in which we trace the occurrences of atoms in the derivation to show how the choice of the active variables changes the pairs of atoms linked by the $\axrule$-rules.
	\begin{equation}\label{eq:EXweq3}
		\hfill
		\vlderivation{
			\vliq{\exists}{}{
				\vdash \vpz{21}{\lsend xa},\vpz{22}{\lsend xb} ,\lEx z \vpz{23}{\lrecv xz}, \lEx z \vpz{24}{\lrecv xz}
			}{
				\vliin{\mixr}{}{
					\vdash \vpz{11}{\lsend xa},\vpz{12}{\lsend xb} , \vpz{13}{\lrecv xa}, \vpz{14}{\lrecv xb}
				}{
					\vlin{\axrule}{}{\vdash \vpz1{\lsend xa}, \vpz2{\lrecv xa}}{\vlhy{}}
				}{
					\vlin{\axrule}{}{\vdash \vpz3{\lsend xb}, \vpz4{\lrecv xb}}{\vlhy{}}
				}
			}
		}
		\pzflows{pz1.center/pz2.center/8/magenta}
		\pzflows{pz3.center/pz4.center/8/magenta}
		\flowedges{%
			pz1.center/pz11.center,%
			pz11.center/pz21.center,%
			pz2.center/pz13.center,%
			pz13.center/pz23.center,%
			pz3.center/pz12.center,%
			pz12.center/pz22.center,%
			pz4.center/pz14.center,%
			pz14.center/pz24.center%
		}
		\quad\not\weq\quad
		\vlderivation{
			\vliq{\exists}{}{
				\vdash \vpz{21}{\lsend xa},\vpz{22}{\lsend xb} ,\lEx z \vpz{23}{\lrecv xz}, \lEx z \vpz{24}{\lrecv xz}
			}{
				\vliin{\mixr}{}{
					\vdash \vpz{11}{\lsend xa},\vpz{12}{\lsend xb} , \vpz{13}{\lrecv xb}, \vpz{14}{\lrecv xa}
				}{
					\vlin{\axrule}{}{\vdash \vpz1{\lsend xa}, \vpz2{\lrecv xa}}{\vlhy{}}
				}{
					\vlin{\axrule}{}{\vdash \vpz3{\lsend xb}, \vpz4{\lrecv xb}}{\vlhy{}}
				}
			}
		}
		\pzflows{pz1.center/pz2.center/8/magenta}
		\pzflows{pz3.center/pz4.center/8/magenta}
		\flowedges{%
			pz1.center/pz11.center,%
			pz11.center/pz21.center,%
			pz2.center/pz14.center,%
			pz14.center/pz24.center,%
			pz3.center/pz12.center,%
			pz12.center/pz22.center,%
			pz4.center/pz13.center,%
			pz13.center/pz23.center%
		}
		\hfill
	\end{equation}
\end{remark}

We conclude by proving our the canonicity results.

\begin{definition}
	Two conflict nets $\linktree_1$ and $\linktree_2$ are \defn{the same} (denoted $\linktree_1 = \linktree_2$) if they are isomorphic \cotree, that is, if there is a bijection between the nodes of $\linktree_1$ and $\linktree_2$ that preserves the structure of the tree.

	For slice nets $\linkings_1$ and $\linkings_2$ are \defn{the same} (denoted $\linkings_1 = \linkings_2$) if they are the same set of linkings.
\end{definition}

\begin{restatable}{lemma}{confCoal}\label{lem:confCoal}
	Let $\linktree$ be a conflict net on $\Gamma$ such that $\linktree$ sequetializes to $\dD_{\linktree}$.
	If $\linktree \coalto^* \linktree'$,
	then $\linktree'$ sequentializes to a $\dD_{\linktree'}$ such that $\dD_{\linktree'}\pweq \dD_{\linktree}$.
\end{restatable}
\begin{proof}
	It suffices to check that each critical pair of coalescence steps on deductive
	\cotrees converges.
	For this, we should consider the derivations labeling each link modulo the proof equivalence generated by local rule permutations \Cref{fig:permutations1}.

	The most convoluted case is the one for the pair $\conf / \conf$, for which we use the same argument in \cite{hei:hug:str:ALL1}, remarking that the base case of induction presents no problems in our setting (where we have the additional information of the dualizers) because of the associativity of the join operator on dualizers.
	Note that critical pairs with non-trivial confluence (i.e., non-local) are the ones where the rule $\lwith$ interacts with quantifiers, and the ones of the form $\mixr / \mixr$ and $\precur/\precur$.
	Details are provided in \Cref{sec:succonet}.
\end{proof}

\begin{theorem}\label{thm:PNareCanon}
	Let $\dD$ and $\dD'$ be two derivations.
	Then, the following hold:
	\begin{enumerate}
		\item $\dD\pweq\dD'$ iff $\cnof{\dD}=\cnof{\dD'}$.
		\item $\dD\spweq\dD'$ iff $\snof{\dD}=\snof{\dD'}$.
	\end{enumerate}
\end{theorem}
\begin{proof}
	It suffices to consider the case analysis given by the rule permutations in \Cref{fig:permutations1} as done in \cite{hei:hug:conflict} for conflict nets, and in \cite{hug:van:slice} for slice nets.
	\begin{enumerate}
		\item For $\pweq$ it follows from the definition of $\cnof\cdot$ and by \Cref{lem:confCoal}.
		The only new cases with respect to \cite{hei:hug:conflict} are given by the local rule permutation changing $\nupopr+\nuloadr$ into a $\yapopr+\yaloadr$ (the permutation boxed in \Cref{fig:permutations1}), and the definition of equivalence modulo fresh names.

		\item Similarly, for $\spweq$ it follows from the definition of $\snof\cdot$ and by \Cref{rem:confEras}.
	\end{enumerate}
\end{proof}

\section{Processes as Formulas}\label{sec:PaF}
\def\tree{\mathcal T}

In this section we show how $\NL$ can be used as a logical framework in which we can intepret proofs as execution trees in the \picalc.

First we provide a translation from \procs to formulas in $\NL$ and show that the logical implication in $\NL$ captures the structural congruence in the \picalc.
According to \cite{miller:pi}, the absence of this property suggests that the logical framework may lack a robust design.

\begin{nota}
	Because of the monoidal laws, we consider generalized $n$-ary versions (with $n>0$) of the additive connectives $\lplus$ and $\lwith$, which are more convenient for the translation of processes.
	Their inference rules are defined as expected:
	the $n$-ary version of the $\lplus$-rule keeps a unique component of the $n$-ary disjunction,
	while the $n$-ary version of the $\lwith$-rule has $n$ premises containing only one of the component of the $n$-ary conjunction, and a copy of the context.
\end{nota}

\begin{definition}[Processes-as-Formulas]
	The formula $\fof P$ associated to a \proc $P$ is inductively defined as follows:
	\begin{equation}\label{eq:translation}\adjustbox{max width=.92\textwidth}{$\begin{array}{r@{\;=\;}l@{\qquad}r@{\;=\;}l}
			\fof{\pnil}		= \lunit
			\qquad
			\fof{P\ppar Q}		& 	\fof P \lpar \fof Q
			&
			\fof{\pnu x(P)}		& 	\lNu x {\fof P}
			\\
			\fof{\psend xyP}	&
			\lsend xy\lprec \fof P 
			&
			\fof{\precv xyP}	&
			\lExp{y}{\lrecv xy\lprec \fof P}
			\\
			\fof{\psels x {\lab:P_\lab}{\lab\in L}}
			&
			\biglsel[\ell\in L]{\left(\lrecv{x}{\ell} \lprec \fof{P_{\ell}}\right)}
			&
			\fof{\pbras x {\lab:P_\lab}{\lab\in L}}
			&
			\biglbra[\ell\in L]{\left(\lsend{x}\ell\lprec\fof{P_{\ell}} \right)}
		\end{array}$}
	\end{equation}
	We denote by $\sof P$ the sequent obtained by recursively removing all top-level $\lpar$-connectives and nominal quantifiers from the sequent, as well as units and unary additive connectives. If the obtained sequent is empty, then we let $\sof P=\lunit$.
\end{definition}

\begin{corollary}
	Let $P$ and $Q$ be \procs.
	If $P\steq Q$, then $\fof Q \feq \fof P$.
\end{corollary}
\begin{nota}
	When translating the processes $\psenda xy = \psend xy \pnil$  (and similarly for $\precva xy$), we will simply write $\fof{\psend xy \pnil}=\lsend xy$ since $\fof{\psend xy \pnil}= \lsend xy\lprec \lunit $ is logically equivalent to $\lsend xy$.
\end{nota}

\begin{remark}
	For a counter-example of \procs which are not $\steq$-equivalent, but whose corresponding formulas are $\feq$-equivalent, consider the \proc where fresh new name is chosen before a choice, and the one in which a fresh name is chosen after a choice.
	\begin{equation}\label{eq:exMilner}
		\hfill
		P=\pnup x {\pbras x{\lab : P_\lab}{\lab\in L}}
		\qquand
		Q=\pbras x{\lab : \pnu x P_\lab}{\lab\in L}
		\hfill
	\end{equation}
	These processes are not equivalent modulo the structural equivalence $\steq$ provided in literature \cite{vasco:pi,gay:hole}, which is the same as the one we provide in \Cref{fig:terms}.

	It is worth noticing that Milner's original \picalc includes the structural equivalence in the top of \Cref{eq:MilnerNC}, while, at the best of our knowledge, the literature on the \picalc as presented in this paper (that is, as in \cite{vasco:pi,gay:hole}) does not include the corresponding structural equivalence in the bottom of \Cref{eq:MilnerNC}, required to capture similar interactions between choices and restriction.
	\begin{equation}\label{eq:MilnerNC}
		\hfill
		\begin{array}{c@{\quad\steq\quad}c}
			\pnup x{A+B}
			&
			\pnu xA + \pnu xB
			\\\hline\hline
			\pbras x{\lab : \pnu y P_\lab}{\lab\in L}
			&
			\pnu y \left(\pbras x{\lab : P_\lab}{\lab\in L}\right)
			\\
			\psels x{\lab : \pnu y P_\lab}{\lab\in L}
			&
			\pnu y \left(\psels x{\lab :  P_\lab}{\lab\in L}\right)
		\end{array}
		\hfill
	\end{equation}
	This rises an interesting question on why those structural equivalences have not being used in the literature, even if the two processes in \Cref{eq:exMilner} have the same behavior with respect to the results in session types.
\end{remark}

As shown in detail in \cite{acc:man:mon:FaP}, it is possible to associate to each execution tree of a \proc $P$ to an open derivation in $\NL$ of a formula $\fof P$, therefore to characterize deadlock-freedom in terms of derivability in $\NL$.
We report here only a sketch of the proof of this result.
\begin{theorem}[\cite{acc:man:mon:FaP}]\label{thm:deadlock}
	Let $P$ be a process.
	\begin{enumerate}
		\item\label{deadlock1} If $P$ is a deadlock-free, then $\proves[\NL]\fof P$.
		\item\label{deadlock2} If $P$ is race-free, then $P$ is deadlock-fre iff $\proves[\NL]\fof P$.
	\end{enumerate}
\end{theorem}
\begin{proof}[Sketch of proof]
	If $P$ is deadlock-free, then each (maximal) execution tree $\tree$ of $P$ has leaves $\pnil$.
	Since terms are considered up-to structural equivalence, we can assume without loss of generality that no child of a process $P$ contains more occurrences of $\pnil$ than $P$.
	This can be obtained by orienting the structural equivalence $P\ppar\pnil \steq P$ in the natural way.

	For each such tree, we define a derivation $\fof \tree$ by induction on the structure of $\tree$ as shown in \Cref{fig:ctTOder}.
	\Cref{deadlock1} follows by definition.
	To prove \Cref{deadlock2}, we show that we can transform a derivation of a formula $\fof P$ into a derivation made of blocks of rules as in \Cref{fig:ctTOder} using rule permutations from \Cref{fig:permutations1}.
	We conclude by remarking that it suffices to check a unique derivation because when $P$ is deadlock-free, then all derivations of $\fof P$  are equivalent with respect to the interleaving relation defined in \Cref{fig:interleaving}.
\end{proof}

\begin{figure}[t]
	\adjustbox{max width=\textwidth}{$\begin{array}{c}
		\fof{\vmod1{\pnil}}
	\quad=\quad
		\vlderivation{
			\vlid{=}{}{
				\vdash \sof{\pnil}
			}{
				\vlin{\lunit}{}{\vdash \lunit}{\vlhy{}}
			}
		}
	\qquad\qquad
		\fof{\begin{array}{c}
			\vmod2{\pnu x \pnus{y}\left( {P} \ppar {Q\fsubst ab}  \ppar R\right)}
			\\[30pt]
			\vmod1{\pnu x \pnus{y}\left( \psend xa {P} \ppar \precv xb {Q}  \ppar R\right)}
		\end{array}}
		\Dledges{mod1/mod2/{\rscomr}}
	\quad=\quad
		\vlderivation{
			\vlid{=}{}{
				\vdash \sof{\pnu x \pnus{y}\left( \psend xa {P} \ppar \precv xb {Q}  \ppar R\right)}
			}{
				\vlin{\exists}{}{
					\vdash \lsend xa\lprec \fof P,\lEx b\left(\lrecv xb \lprec \fof Q \right), \Gamma
				}{
					\vliin{\lprec}{}{
						\vdash \lsend xa\lprec \fof P ,\lrecv xa \lprec \fof {Q \fsubst ab}, \sof R
					}{
						\vlin{\axrule}{}{\vdash \lsend xa,\lrecv xa}{\vlhy{}}
					}{
						\vlde{}{\set{\lpar,\nuur,\mixr,\lunit}}{
							\vdash \fof P , \fof {Q \fsubst ab}, \sof R
						}{
							\vlid{=}{}{
								\vdash \sof P , \sof {Q \fsubst ab}, \sof R
							}{
								\vlhy{
									\vdash \sof{\pnu x \pnus{y}\left( {P} \ppar {Q\fsubst ab}  \ppar R\right)}
								}
							}
						}
					}
				}
			}
		}
	\\\\\\
		\fof{\begin{array}{c}
			\vmod2{\pnu x \pnus{y}\left(
				P_{\lab_k}
				\ppar
				Q_{\lab_k}
				\ppar
				R
			\right)}
			\\[30pt]
			\vmod1{\pnu x \pnus{y}\left(
				\pbra x{\lab}{P_{\lab_k}}
				\ppar
				\psels x{\lab:Q_{\lab}}{\lab\in L}
				\ppar
				R
			\right)}
		\end{array}}
		\Dledges{mod1/mod2/{\rsselr}}
	\quad=\quad
		\vlderivation{
			\vlid{=}{}{
				\vdash \sof{\pnu x \pnus{y}\left(
					\pbra x{\lab_k}{P_{\lab_k}}
					\ppar
					\psels x{\lab:Q_{\lab}}{\lab\in L}
					\ppar
					R
				\right)}
			}{
				\vlin{\oplus}{}{
					\vdash
					\lsend{x}{\lab_k}\lprec\fof{P_{\lab_k}},
					\biglsel[\lab\in L]{\left(\lrecv{x}{\lab} \lprec \fof{Q_{\lab}}\right)},
					\sof R
				}{
					\vliin{\lprec}{}{
						\vdash
						\lsend{x}{\lab_k}\lprec\fof{P_{\lab_k}},
						\lrecv{x}{\lab_k} \lprec \fof{Q_{\lab_k}},
						\sof R
					}{
						\vlin{\axrule}{}{\vdash \lsend x{\lab_k},\lrecv x{\lab_k}}{\vlhy{}}
					}{
						\vlde{}{\set{\lpar,\nuur,\mixr,\lunit}}{
							\vdash
							\fof{P_{\lab_k}},
							\fof{Q_{\lab_k}},
							\sof R
						}{
							\vlid{=}{}{
								\vdash
								\sof{P_{\lab_k}},
								\sof{Q_{\lab_k}},
								\sof R
							}{
								\vlhy{
									\vdash
									\sof{\pnu x \pnus{y}\left(
										P_{\lab_k}
										\ppar
										Q_{\lab_k}
										\ppar
										R
									\right)}
								}
							}
						}
					}
				}
			}
		}
	\\\\\\
		\fof{\begin{array}{c}
			\vmod2{\pnus{y}\left( \pbra x{\lab_1}{P_{\lab_1}}  \ppar R\right)}
			\qquad
			\cdots
			\qquad
			\vmod3{\pnus{y}\left( \pbra x{\lab_m}{P_{\lab_m}} \ppar R\right)}
			\\[30pt]
			\vmod1{\pnus{y}\left( \pbras x{\lab: P_\lab}{\lab\in \set{\lab_1,\ldots, \lab_m}}  \ppar R\right)}
		\end{array}}
		\Dledges{mod1/mod2/{\rsbrar},mod1/mod3/{\rsbrar}}
	\quad=\quad
		\vlderivation{
			\vlid{}{}{
				\vdash \sof{\pnus{y}\left( \pbras x{\lab: P_\lab}{\lab\in \set{\lab_1,\ldots, \lab_m}}  \ppar R\right)}
			}{
				\vlin{\lwith}{}{
					\vdash
					\bigwith_{\lab\in \set{\lab_1,\ldots, \lab_m}} \lsend x\lab\lprec \fof {P_\lab}
					, \sof R
				}{
					\vlhy{\left\{\vlderivation{
						\vlid{=}{}{
							\vdash
							\lsend x{\lab}\lprec \fof {P_{\lab}}
							, \sof R
						}{
							\vlhy{\sof{\pnus{y}\left( \pbra x{\lab}{P_{\lab}}  \ppar R\right)}}
						}
					}\right\}_{\lab\in \set{\lab_1,\ldots, \lab_m}}}
				}
			}
		}
	\end{array}
	$}
	\caption{Translation of execution trees to derivations. If $m=1$ in the case of $\rsbrar$, then the open derivation is just the sequent $\sof{\pnus{y}\left( \pbras x{\lab: P_\lab}{\lab\in \set{\lab_1,\ldots, \lab_m}}  \ppar R\right)}$, i.e., no rule.}
	\label{fig:ctTOder}
\end{figure}

\begin{remark}
	Certain works on \picalc (e.g., \cite{sangiorgi:internal}) restrict communication and selection on restricted channels (i.e., communication or selection on a channel $y$ can be performed only if $y$ is bound by a $\nu$).
	To capture such a restriction it would be sufficient to require that an $\axrule$-rule with conclusion $\lsend xy$ and $\lrecv xy$ can be applied only if the $x$ is bounded by a $\lnewsymb$ in a sequent occurring in the derivation below the rule.
	In the proof net defined in \Cref{sec:PN}, this restriction corresponds to require that the two formulas in an axiomatic link $\la=\set{\vpz1{\lsend xy},\vpz2{ \lrecv vw}}$ contains two vertices above a same $\lnewsymb z$ node.
\end{remark}

The canonicity result on proof nets with respect to the sequent calculus allows us to provide canonical representatives of execution trees modulo inteleaving.
To prove this result not only for deadlock-free processes, we extend the syntax of proof nets to include \emph{non-logical axioms} to model open premises of a derivation.
\begin{definition}\label{def:PNopen}
	A \defn{open conflict net} \resp{open slice net} is a coalescent  canonical \cotree $\linktree$ \resp{an erasable set of axiomatic linkings} on a sequent $\Gamma$.
	The (top-down) translation in \Cref{fig:deseq} is extended from derivations to open derivations by translating each open premise $\sdash A_1,\ldots, A_n$ in a the link $\la=\set{\vA1_1,\vpz1{\ldots}, \vA n_n}\pzlinks{A1/An/-10/\la/red/{pz1}}$ with $\dualizerof[\la]=\emptyset$.
\end{definition}

\begin{theorem}
	Two execution trees of a \procs $P$ are equivalent modulo interleaving \resp{local interleaving}
	iff
	they can be represented by the same open slice net \resp{open conflict net}.
\end{theorem}
\begin{proof}
	We associate each execution tree $\tree$ the slice net $\snof \tree$ by combining the translations in \Cref{fig:ctTOder} and \Cref{fig:sliceDeseq} (plus the special case for open premises defined in \Cref{def:PNopen}).
	If $\tree\teq \tree'$, then $\fof{\tree} \speq \fof{\tree'}$; therefore $\snof\tree=\snof{\tree'}$ by \Cref{thm:PNareCanon}.
	The result for local interleaving is similar.
\end{proof}

\begin{corollary}
	A \proc $P$	 is race-free iff there is a unique possible slice net for $\fof P$.
\end{corollary}

\section{Conclusion and Future Works}\label{sec:conc}

In this paper we presented $\PIL$, an extension of first-order multiplicative-additive linear logic in which non-commutativity is not obtained by considering sequents as lists of formulas (as in \cite{abr:rue:noncomI,ruet:noncomII,gal:not:noncom}), but rather including a non-commutative binary connective, as in Retoré's $\pomset$ logic \cite{retore:phd,ret:newPomset} or Guglielmi's $\BV$ \cite{gug:SIS}.
We have shown that by requiring such non-commutative connective to also be non-associative, we can design sequent calculus in which the $\cutr$-rule is admissible.
We also provided proof nets for this logic with a polynomial-time correctness criterion (that is, proof nets form a proof system in the sense of \cite{cook:reckhow:79}), polynomial-time sequentialization and proof translation, and we showed that proof nets provide canonical representatives of derivations modulo local rule permutations.
Note that our result is stronger than what needed to define proof nets (both conflict and slice nets) for $\MALL^1$, addressing a question left open in the literature. Moreover, it could be possible to define both conflict and slice `witness nets' for $\PIL$ (and then for $\MALL^1$) by including in the initial data of the proof net the information of a specific witness map instead of the initial witness map.

\myparagraph{Characerization of deadlock-freedom}\label{sec:rec}
%
In  \cite{acc:man:mon:FaP} we have shown that each derivation in $\PIL$ of the formula $\fof P$ can be interpreted as a execution tree of the \proc $P$ of the \picalc (defined as in \cite{vasco:pi,gay:hole}).
In the same paper, we used such a correspondence to show that (finite) deadlock-freedom and race-free processes can be characterized in terms of derivability in $\PIL$.
To remove the race-free condition, we should be able to reason over all possible proof-search attempts, which has an exponential blowup because of interleaving.
We are currently studying the possibility of reducing this blowup by using slice nets, allowing us to quantify over sets of linkings which may be valid slice nets instead of all possible ``maximal open derivations''.

\myparagraph{Extensions of $\NL$ with fixpoints}\label{sec:rec}
%
In this work, we have studied the recursion-free fragment of the \picalc, but we foresee the possibility of modeling the following three main approaches for the definition of infinite behaviors (see \cite{bus:gab:zav:expressive} for a comparison of their expressive power).
\emph{Replication} \resp{\emph{iteration}} could be modeled using the modality \emph{why-not} \resp{\emph{flag}} defined as fixpoint of the equation $\wn A=A\lpar (\wn A)$ \resp{$\lflag A=A\lprec (\lflag A)$} as in parsimonious linear logic \cite{mazza:15,acc:cur:gue:CSL24,acc:cur:gue:CSL24ext} \resp{as a ``parsimonious'' version of the modalities from \cite{reddy:llstate} and from \cite{ret:flag}},
and \emph{recursion} using the greatest-fixpoint operator (which we here denote $\lnu$ to avoid confusion with the restriction) from $\muMALL$ \cite{BaeldeM07,bae:dou:sau:16}.
\begin{equation}
	\hfill
	\begin{array}{c@{\qquad}|@{\qquad}c@{\qquad}|@{\qquad}c}
		\mbox{Replication}
		&
		\mbox{Iteration}
		&
		\mbox{Recursion}
	\\\hline
		\vlinf{\wn\mathsf{b}}{}{\sdash \Gamma, \wn A}{\sdash \Gamma, A, \wn A}
		&
		\vlinf{\lflag\mathsf{b}}{}{\sdash \Gamma, \lflag A}{\sdash \Gamma, A\lprec \lflag A}
		&
		\vlinf{\lnu}{}{\sdash \Gamma, \lnu A.P(A)}{\sdash \Gamma, P(\lnu A)}
	\end{array}
	\hfill
\end{equation}
Proof systems capturing these operators should include rules allowing the definition of correct non-wellfounded derivations as in, e.g., \cite{bae:dou:sau:16,acc:cur:gue:CSL24,acc:mon:per:OPDL}.
An interesting challenge will be to find a suitable syntax for infinitary proof nets for these systems aiming at proof canonicity rather than to well-behavior with respect to cut-elimination as in \cite{de:sau:infininets,de:phd}.

\myparagraph{Coherent spaces for $\NML$}
%
As a consequence of \Cref{thm:embedding} and cut-elimination, we have that $\NML$ embeds in $\BV$, therefore $\NML$ embeds in $\pomset$.
Since $\pomset$ admits a semantics in terms of coherent spaces \cite{retore:phd,ret:newPomset}, it should be possible to characterize the class of $\NML$ theorems inside $\pomset$, and use this logic to study sequential algorithms \cite{reddy:llstate,ehr:seq1,ehr:seq2}.

\myparagraph{Applications to concurrent programming languages}
%
This paper represents the first step in the definition of a novel paradigm allowing to intepret proofs as execution trees, and using proof equivalence to identify execution trees differing from interleaving concurrency, allowing for the developments of new methods complementary to the ones based session types.

In \cite{acc:man:mon:FaP} we have shown that deadlock-freedom can be characterized in terms of derivability in $\NL$, providing the first completeness result for choreographic programming with respect to recursion-free deadlock-freedom.
This was an open question in the literature, and the main difficulty of proving this result was due to syntactic limitations of session types in presence of key features such as name mobility or cyclic dependencies.
By developing a theory of $\NL$ with fixpoints, we expect to be able to extend this result to the full \picalc.

We plan to study applications of our framework such as the definition of a notion of orthogonality for formulas based on the existence of (open) proof nets which could be completed by set of axioms connecting them (see the notion of module in \cite{and:maz:concPN,acc:mai:DCM}).
This could be used to characterize the \emph{testing preorders} \cite{den:hen:testing,hennessy1988algebraic,ber:hen:testing}, thus designing verification tools whose efficiency relies on the low complexity of the proof net correctness criterion, combined with the fact that a single proof net can encode an exponential number of equivalent executions.

Moreover, our methodology is complementary to algebraic methods which, to the best of our knowledge, currently lack of methods to model name mobility and restrictions.
This because these features are strongly tight to quantification, and modeling quantification has proven to be a complex task in algebraic settings:
while Boolean and Heyting algebras provide straightforward algebraic models for classical and intuitionistic propositional logic respectively, the algebraic structures modeling first-order classical logic have only recently been studied (see \cite{bonchi:FO}).

\myparagraph{Towards a semantics for sequentiality}
%
In this paper we used the non-associative connective $\lprec$ to model a special form of sequentiality, the prefix operator, but we consider extending this work to sequent calculi using graphical connectives (in the sense of \cite{acc:IJCAR24}) to model more complex pattern of interactions as done in \cite{acc:FSCD22}, as well as study deep inference systems to recover associativity of the sequential operator -- thus more similar to the systems studied in \cite{bru:02,hor:nom,hor:tiu:19}.

\bibliography{biblio}

\begin{thebibliography}{10}

\bibitem{abr:rue:noncomI}
V.Michele Abrusci and Paul Ruet.
\newblock Non-commutative logic {I}: the multiplicative fragment.
\newblock {\em Annals of Pure and Applied Logic}, 101(1):29--64, 1999.
\newblock URL: \url{https://www.sciencedirect.com/science/article/pii/S0168007299000147}, \href {https://doi.org/10.1016/S0168-0072(99)00014-7} {\path{doi:10.1016/S0168-0072(99)00014-7}}.

\bibitem{acc:EHPN}
Matteo Acclavio.
\newblock Exponentially handsome proof nets and their normalization.
\newblock {\em Electronic Proceedings in Theoretical Computer Science}, 353:1–25, Dec 2021.
\newblock URL: \url{http://dx.doi.org/10.4204/EPTCS.353.1}, \href {https://doi.org/10.4204/eptcs.353.1} {\path{doi:10.4204/eptcs.353.1}}.

\bibitem{acc:IJCAR24}
Matteo Acclavio.
\newblock Sequent systems on undirected graphs.
\newblock In Christoph Benzm{\"u}ller, Marijn~J.H. Heule, and Renate~A. Schmidt, editors, {\em Automated Reasoning}, pages 216--236, Cham, 2024. Springer Nature Switzerland.

\bibitem{acc:cur:gue:CSL24}
Matteo Acclavio, Gianluca Curzi, and Giulio Guerrieri.
\newblock Infinitary cut-elimination via finite approximations.
\newblock {\em CoRR}, abs/2308.07789, 2023.
\newblock URL: \url{https://doi.org/10.48550/arXiv.2308.07789}, \href {https://arxiv.org/abs/2308.07789} {\path{arXiv:2308.07789}}, \href {https://doi.org/10.48550/ARXIV.2308.07789} {\path{doi:10.48550/ARXIV.2308.07789}}.

\bibitem{acc:cur:gue:CSL24ext}
Matteo Acclavio, Gianluca Curzi, and Giulio Guerrieri.
\newblock Infinitary cut-elimination via finite approximations (extended version), 2024.
\newblock URL: \url{https://arxiv.org/abs/2308.07789}, \href {https://arxiv.org/abs/2308.07789} {\path{arXiv:2308.07789}}.

\bibitem{acc:FSCD22}
Matteo Acclavio, Ross Horne, Sjouke Mauw, and Lutz Stra{\ss}burger.
\newblock {A Graphical Proof Theory of Logical Time}.
\newblock In Amy~P. Felty, editor, {\em 7th International Conference on Formal Structures for Computation and Deduction (FSCD 2022)}, volume 228 of {\em Leibniz International Proceedings in Informatics (LIPIcs)}, pages 22:1--22:25, Dagstuhl, Germany, 2022. Schloss Dagstuhl -- Leibniz-Zentrum f{\"u}r Informatik.
\newblock URL: \url{https://drops.dagstuhl.de/entities/document/10.4230/LIPIcs.FSCD.2022.22}, \href {https://doi.org/10.4230/LIPIcs.FSCD.2022.22} {\path{doi:10.4230/LIPIcs.FSCD.2022.22}}.

\bibitem{acc:mai:DCM}
Matteo Acclavio and Roberto Maieli.
\newblock Logic programming with multiplicative structures, 2024.
\newblock URL: \url{https://arxiv.org/abs/2403.03032}, \href {https://arxiv.org/abs/2403.03032} {\path{arXiv:2403.03032}}.

\bibitem{acc:man:mon:FaP}
Matteo Acclavio, Giulia Manara, and Fabrizio Montesi.
\newblock Formulas as processes, deadlock-freedom as choreographies (extended version), 2025.
\newblock URL: \url{https://arxiv.org/abs/2501.08928}, \href {https://arxiv.org/abs/2501.08928} {\path{arXiv:2501.08928}}.

\bibitem{acc:mon:per:OPDL}
Matteo Acclavio, Fabrizio Montesi, and Marco Peressotti.
\newblock On propositional dynamic logic and concurrency, 2024.
\newblock \href {https://arxiv.org/abs/2403.18508} {\path{arXiv:2403.18508}}.

\bibitem{acc:str:relevant}
Matteo Acclavio and Lutz Stra{\ss}burger.
\newblock On combinatorial proofs for logics of relevance and entailment.
\newblock In Rosalie Iemhoff, Michael Moortgat, and Ruy de~Queiroz, editors, {\em Logic, Language, Information, and Computation}, pages 1--16, Berlin, Heidelberg, 2019. Springer Berlin Heidelberg.

\bibitem{acc:str:modal}
Matteo Acclavio and Lutz Stra{\ss}burger.
\newblock On combinatorial proofs for modal logic.
\newblock In Serenella Cerrito and Andrei Popescu, editors, {\em Automated Reasoning with Analytic Tableaux and Related Methods}, pages 223--240, Cham, 2019. Springer International Publishing.

\bibitem{relevant}
Alan~Ross Anderson and Nuel D.~jun. Belnap.
\newblock Entailment. {The} logic of relevance and necessity. {Vol}. {I}.
\newblock Princeton, {N}. {J}.: {Princeton} {University} {Press}. {XXXII}, 543 p. \$ 23.50 (1975)., 1975.

\bibitem{andreoli:focPN}
Jean~Marc Andreoli.
\newblock Focussing proof-net construction as a middleware paradigm.
\newblock In Andrei Voronkov, editor, {\em Automated Deduction---CADE-18}, pages 501--516, Berlin, Heidelberg, 2002. Springer Berlin Heidelberg.

\bibitem{and:maz:concPN}
Jean-Marc Andreoli and Laurent Mazar{\'e}.
\newblock Concurrent construction of proof-nets.
\newblock In Matthias Baaz and Johann~A. Makowsky, editors, {\em Computer Science Logic}, pages 29--42, Berlin, Heidelberg, 2003. Springer Berlin Heidelberg.

\bibitem{avr:canonical:01}
Arnon Avron and Iddo Lev.
\newblock Canonical propositional {Gentzen}-type systems.
\newblock In Rajeev Gor{\'e}, Alexander Leitsch, and Tobias Nipkow, editors, {\em Automated Reasoning}, pages 529--544, Berlin, Heidelberg, 2001. Springer Berlin Heidelberg.

\bibitem{bae:dou:sau:16}
David Baelde, Amina Doumane, and Alexis Saurin.
\newblock Infinitary proof theory: the multiplicative additive case.
\newblock In Jean{-}Marc Talbot and Laurent Regnier, editors, {\em 25th {EACSL} Annual Conference on Computer Science Logic, {CSL} 2016, August 29 - September 1, 2016, Marseille, France}, volume~62 of {\em LIPIcs}, pages 42:1--42:17. Schloss Dagstuhl - Leibniz-Zentrum f{\"{u}}r Informatik, 2016.
\newblock \href {https://doi.org/10.4230/LIPIcs.CSL.2016.42} {\path{doi:10.4230/LIPIcs.CSL.2016.42}}.

\bibitem{BaeldeM07}
David Baelde and Dale Miller.
\newblock Least and greatest fixed points in linear logic.
\newblock In Nachum Dershowitz and Andrei Voronkov, editors, {\em Logic for Programming, Artificial Intelligence, and Reasoning, 14th International Conference, {LPAR} 2007, Yerevan, Armenia, October 15-19, 2007, Proceedings}, volume 4790 of {\em Lecture Notes in Computer Science}, pages 92--106. Springer, 2007.
\newblock \href {https://doi.org/10.1007/978-3-540-75560-9\_9} {\path{doi:10.1007/978-3-540-75560-9\_9}}.

\bibitem{bellin:subnets}
Gianluigi Bellin.
\newblock Subnets of proof-nets in multiplicative linear logic with mix.
\newblock {\em Mathematical Structures in Computer Science}, 7(6):663--669, 1997.

\bibitem{ber:hen:testing}
Giovanni Bernardi and Matthew Hennessy.
\newblock {Mutually Testing Processes}.
\newblock {\em {Logical Methods in Computer Science}}, {Volume 11, Issue 2}, April 2015.
\newblock URL: \url{https://lmcs.episciences.org/776}, \href {https://doi.org/10.2168/LMCS-11(2:1)2015} {\path{doi:10.2168/LMCS-11(2:1)2015}}.

\bibitem{bonchi:FO}
Filippo Bonchi, Alessandro Di~Giorgio, Nathan Haydon, and Pawel Sobocinski.
\newblock Diagrammatic algebra of first order logic.
\newblock In {\em Proceedings of the 39th Annual ACM/IEEE Symposium on Logic in Computer Science}, LICS '24, New York, NY, USA, 2024. Association for Computing Machinery.
\newblock \href {https://doi.org/10.1145/3661814.3662078} {\path{doi:10.1145/3661814.3662078}}.

\bibitem{bru:02}
Paola Bruscoli.
\newblock A purely logical account of sequentiality in proof search.
\newblock In {\em International Conference on Logic Programming}, pages 302--316. Springer, 2002.

\bibitem{bus:gab:zav:expressive}
Nadia Busi, Maurizio Gabbrielli, and Gianluigi Zavattaro.
\newblock On the expressive power of recursion, replication and iteration in process calculi.
\newblock {\em Mathematical Structures in Computer Science}, 19(6):1191–1222, 2009.
\newblock \href {https://doi.org/10.1017/S096012950999017X} {\path{doi:10.1017/S096012950999017X}}.

\bibitem{che:nominal}
James Cheney.
\newblock A simpler proof theory for nominal logic.
\newblock In Vladimiro Sassone, editor, {\em Foundations of Software Science and Computational Structures}, pages 379--394, Berlin, Heidelberg, 2005. Springer Berlin Heidelberg.

\bibitem{cook:reckhow:79}
Stephen~A. Cook and Robert~A. Reckhow.
\newblock The relative efficiency of propositional proof systems.
\newblock {\em Journal of Symbolic Logic}, 44(1):36–50, 1979.
\newblock \href {https://doi.org/10.2307/2273702} {\path{doi:10.2307/2273702}}.

\bibitem{de:phd}
Abhishek De.
\newblock {\em Linear logic with the least and greatest fixed points : truth semantics, complexity and a parallel syntax. (La logique lin{\'{e}}aire avec les plus petits et les plus grands points fixes : la s{\'{e}}mantique de v{\'{e}}rit{\'{e}}, la complexit{\'{e}}, et une syntaxe parall{\`{e}}le)}.
\newblock PhD thesis, Paris Cit{\'{e}} University, France, 2022.
\newblock URL: \url{https://tel.archives-ouvertes.fr/tel-04274800}.

\bibitem{de:sau:infininets}
Abhishek De and Alexis Saurin.
\newblock Infinets: The parallel syntax for non-wellfounded proof-theory.
\newblock In Serenella Cerrito and Andrei Popescu, editors, {\em Automated Reasoning with Analytic Tableaux and Related Methods}, pages 297--316, Cham, 2019. Springer International Publishing.

\bibitem{den:hen:testing}
Rocco De~Nicola and Matthew~CB Hennessy.
\newblock Testing equivalences for processes.
\newblock {\em Theoretical computer science}, 34(1-2):83--133, 1984.

\bibitem{ehr:seq1}
Thomas Ehrhard.
\newblock Projecting sequential algorithms on strongly stable functions.
\newblock {\em Annals of Pure and Applied Logic}, 77(3):201--244, 1996.
\newblock URL: \url{https://www.sciencedirect.com/science/article/pii/0168007295000267}, \href {https://doi.org/10.1016/0168-0072(95)00026-7} {\path{doi:10.1016/0168-0072(95)00026-7}}.

\bibitem{ehr:seq2}
Thomas Ehrhard.
\newblock Parallel and serial hypercoherences.
\newblock {\em Theoretical Computer Science}, 247(1):39--81, 2000.
\newblock URL: \url{https://www.sciencedirect.com/science/article/pii/S0304397500001730}, \href {https://doi.org/10.1016/S0304-3975(00)00173-0} {\path{doi:10.1016/S0304-3975(00)00173-0}}.

\bibitem{gabbay:pitts:nominal}
Murdoch~J. Gabbay and Andrew~M. Pitts.
\newblock A new approach to abstract syntax with variable binding.
\newblock {\em Form. Asp. Comput.}, 13(3–5):341–363, jul 2002.
\newblock \href {https://doi.org/10.1007/s001650200016} {\path{doi:10.1007/s001650200016}}.

\bibitem{gal:not:noncom}
D.~Galmiche and J.~M. Notin.
\newblock Connection-based proof construction in non-commutative logic.
\newblock In Moshe~Y. Vardi and Andrei Voronkov, editors, {\em Logic for Programming, Artificial Intelligence, and Reasoning}, pages 422--436, Berlin, Heidelberg, 2003. Springer Berlin Heidelberg.

\bibitem{gay:hole}
Simon Gay and Malcolm Hole.
\newblock Subtyping for session types in the pi calculus.
\newblock {\em Acta Informatica}, 42:191--225, 2005.

\bibitem{gir:ll}
Jean-Yves Girard.
\newblock Linear logic.
\newblock {\em Theoretical Computer Science}, 50(1):1--101, 1987.
\newblock \href {https://doi.org/10.1016/0304-3975(87)90045-4} {\path{doi:10.1016/0304-3975(87)90045-4}}.

\bibitem{gir:quant1}
Jean-Yves Girard.
\newblock Quantifiers in linear logic.
\newblock In {\em Proceedings of the SILFS conference held in Cesena}, volume~41, 1987.

\bibitem{gir:quant2}
Jean-Yves Girard.
\newblock Quantifiers in linear logic ii.
\newblock {\em Nuovi problemi della logica e della filosofia della scienza}, 2(1), 1991.

\bibitem{girard:96:PN}
Jean-Yves Girard.
\newblock Proof-nets : the parallel syntax for proof-theory.
\newblock In Aldo Ursini and Paolo Agliano, editors, {\em Logic and Algebra}. Marcel Dekker, New York, 1996.

\bibitem{gir:blind}
Jean-Yves Girard.
\newblock The blind spot.
\newblock {\em HAL}, 2011(0), 2011.
\newblock URL: \url{http://dml.mathdoc.fr/item/ISBN: 978-3-03719-088-3}.

\bibitem{guglielmi:concurrecy}
Alessio Guglielmi.
\newblock Concurrency and plan generation in a logic programming language with a sequential operator.
\newblock In {\em ICLP}, pages 240--254. Citeseer, 1994.

\bibitem{guglielmi:95:sequentiality}
Alessio Guglielmi.
\newblock Sequentiality by linear implication and universal quantification.
\newblock In J{\"o}rg Desel, editor, {\em Structures in Concurrency Theory}, pages 160--174, London, 1995. Springer London.

\bibitem{gug:SIS}
Alessio Guglielmi.
\newblock A system of interaction and structure.
\newblock 8(1):1--64, 2007.
\newblock \href {https://doi.org/10.1145/1182613.1182614} {\path{doi:10.1145/1182613.1182614}}.

\bibitem{heijltjes:houston:14}
Willem Heijltjes and Robin Houston.
\newblock No proof nets for {MLL} with units: proof equivalence in {MLL} is {PSPACE}-complete.
\newblock In Thomas~A. Henzinger and Dale Miller, editors, {\em Joint Meeting of the Twenty-Third {EACSL} Annual Conference on Computer Science Logic {(CSL)} and the Twenty-Ninth Annual {ACM/IEEE} Symposium on Logic in Computer Science (LICS), {CSL-LICS} '14, Vienna, Austria, July 14 - 18, 2014}, pages 50:1--50:10. {ACM}, 2014.

\bibitem{hei:hug:str:ALL1}
Willem~B. Heijltjes, Dominic J.~D. Hughes, and Lutz Stra{\ss}burger.
\newblock {Proof Nets for First-Order Additive Linear Logic}.
\newblock In Herman Geuvers, editor, {\em 4th International Conference on Formal Structures for Computation and Deduction (FSCD 2019)}, volume 131 of {\em Leibniz International Proceedings in Informatics (LIPIcs)}, pages 22:1--22:22, Dagstuhl, Germany, 2019. Schloss Dagstuhl -- Leibniz-Zentrum f{\"u}r Informatik.
\newblock URL: \url{https://drops.dagstuhl.de/entities/document/10.4230/LIPIcs.FSCD.2019.22}, \href {https://doi.org/10.4230/LIPIcs.FSCD.2019.22} {\path{doi:10.4230/LIPIcs.FSCD.2019.22}}.

\bibitem{hennessy1988algebraic}
Matthew Hennessy.
\newblock {\em Algebraic theory of processes}.
\newblock MIT press, 1988.

\bibitem{hon:yos:car:multiparty}
Kohei Honda, Nobuko Yoshida, and Marco Carbone.
\newblock Multiparty asynchronous session types.
\newblock {\em SIGPLAN Not.}, 43(1):273–284, jan 2008.
\newblock \href {https://doi.org/10.1145/1328897.1328472} {\path{doi:10.1145/1328897.1328472}}.

\bibitem{hor:tiu:19}
Ross Horne and Alwen Tiu.
\newblock Constructing weak simulations from linear implications for processes with private names.
\newblock {\em Mathematical Structures in Computer Science}, 29(8):1275–1308, 2019.
\newblock \href {https://doi.org/10.1017/S0960129518000452} {\path{doi:10.1017/S0960129518000452}}.

\bibitem{hor:tiu:tow}
Ross Horne, Alwen Tiu, Bogdan Aman, and Gabriel Ciobanu.
\newblock {Private Names in Non-Commutative Logic}.
\newblock In Jos\'{e}e Desharnais and Radha Jagadeesan, editors, {\em 27th International Conference on Concurrency Theory (CONCUR 2016)}, volume~59 of {\em Leibniz International Proceedings in Informatics (LIPIcs)}, pages 31:1--31:16, Dagstuhl, Germany, 2016. Schloss Dagstuhl -- Leibniz-Zentrum f{\"u}r Informatik.
\newblock URL: \url{https://drops.dagstuhl.de/entities/document/10.4230/LIPIcs.CONCUR.2016.31}, \href {https://doi.org/10.4230/LIPIcs.CONCUR.2016.31} {\path{doi:10.4230/LIPIcs.CONCUR.2016.31}}.

\bibitem{hor:tiu:ama:cio:private}
Ross Horne, Alwen Tiu, Bogdan Aman, and Gabriel Ciobanu.
\newblock {Private Names in Non-Commutative Logic}.
\newblock In Jos\'{e}e Desharnais and Radha Jagadeesan, editors, {\em 27th International Conference on Concurrency Theory (CONCUR 2016)}, volume~59 of {\em Leibniz International Proceedings in Informatics (LIPIcs)}, pages 31:1--31:16, Dagstuhl, Germany, 2016. Schloss Dagstuhl -- Leibniz-Zentrum f{\"u}r Informatik.
\newblock URL: \url{https://drops.dagstuhl.de/entities/document/10.4230/LIPIcs.CONCUR.2016.31}, \href {https://doi.org/10.4230/LIPIcs.CONCUR.2016.31} {\path{doi:10.4230/LIPIcs.CONCUR.2016.31}}.

\bibitem{hor:nom}
Ross Horne, Alwen Tiu, Bogdan Aman, and Gabriel Ciobanu.
\newblock De morgan dual nominal quantifiers modelling private names in non-commutative logic.
\newblock {\em ACM Trans. Comput. Logic}, 20(4), jul 2019.
\newblock \href {https://doi.org/10.1145/3325821} {\path{doi:10.1145/3325821}}.

\bibitem{hughes:pws}
Dominic Hughes.
\newblock Proofs {W}ithout {S}yntax.
\newblock {\em Annals of Mathematics}, 164(3):1065--1076, 2006.
\newblock \href {https://doi.org/10.4007/annals.2006.164.1065} {\path{doi:10.4007/annals.2006.164.1065}}.

\bibitem{hughes:invar}
Dominic Hughes.
\newblock Towards {H}ilbert's 24\({}^{\mbox{th}}\) problem: Combinatorial proof invariants: (preliminary version).
\newblock {\em Electr. Notes Theor. Comput. Sci.}, 165:37--63, 2006.

\bibitem{hughes:firstorder}
Dominic Hughes.
\newblock First-order proofs without syntax.
\newblock Berkeley Logic Colloquium, 2014.

\bibitem{hei:hug:conflict}
Dominic Hughes and Willem Heijltjes.
\newblock Conflict nets: Efficient locally canonical mall proof nets.
\newblock In {\em Proceedings of the 31st Annual ACM/IEEE Symposium on Logic in Computer Science}, LICS '16, page 437–446, New York, NY, USA, 2016. Association for Computing Machinery.
\newblock \href {https://doi.org/10.1145/2933575.2934559} {\path{doi:10.1145/2933575.2934559}}.

\bibitem{hughes:conflict}
Dominic J.~D. Hughes.
\newblock Abstract p-time proof nets for mall: Conflict nets, 2008.
\newblock URL: \url{https://arxiv.org/abs/0801.2421}, \href {https://arxiv.org/abs/0801.2421} {\path{arXiv:0801.2421}}.

\bibitem{hug:unification}
Dominic J.~D. Hughes.
\newblock Unification nets: canonical proof net quantifiers.
\newblock In {\em Proceedings of the 33rd Annual ACM/IEEE Symposium on Logic in Computer Science}, LICS '18, page 540–549, New York, NY, USA, 2018. Association for Computing Machinery.
\newblock \href {https://doi.org/10.1145/3209108.3209159} {\path{doi:10.1145/3209108.3209159}}.

\bibitem{hug:van:sliceLICS}
Dominic J.~D. Hughes and Rob~J. Van~Glabbeek.
\newblock Proof nets for unit-free multiplicative-additive linear logic (extended abstract).
\newblock In {\em Proceedings of the 18th Annual IEEE Symposium on Logic in Computer Science}, LICS '03, page~1, USA, 2003. IEEE Computer Society.

\bibitem{hug:van:slice}
Dominic J.~D. Hughes and Rob~J. Van~Glabbeek.
\newblock Proof nets for unit-free multiplicative-additive linear logic.
\newblock {\em ACM Trans. Comput. Logic}, 6(4):784–842, oct 2005.
\newblock \href {https://doi.org/10.1145/1094622.1094629} {\path{doi:10.1145/1094622.1094629}}.

\bibitem{ind:monotonicModal}
Andrzej Indrzejczak.
\newblock Sequent calculi for monotonic modal logics.
\newblock {\em Bulletin of the Section of Logic}, 34(3):151--164, 2005.

\bibitem{lambek:deductive}
Joachim Lambek.
\newblock Deductive systems and categories {I}, {II}, {III}, 1968-1972.

\bibitem{lel:pim:modal}
Bj\"{o}rn Lellmann and Elaine Pimentel.
\newblock Modularisation of sequent calculi for normal and non-normal modalities.
\newblock {\em ACM Trans. Comput. Logic}, 20(2), feb 2019.
\newblock \href {https://doi.org/10.1145/3288757} {\path{doi:10.1145/3288757}}.

\bibitem{manara:phd}
Giulia Manara.
\newblock {\em Linear Logic: Parallel cut elimination and Computation-as-Deduction for the $\pi$-calculus}.
\newblock Ph.{D}. thesis, Paris Cité University, France \& University Roma Tre, Italy, 2025.

\bibitem{mar:mon:unification}
Alberto Martelli and Ugo Montanari.
\newblock {\em Unification in linear time and space: A structured presentation}.
\newblock Istituto di Elaborazione della Informazione, Consiglio Nazionale delle Ricerche, 1976.

\bibitem{mazza:15}
Damiano Mazza.
\newblock Simple parsimonious types and logarithmic space.
\newblock In {\em 24th {EACSL} Annual Conference on Computer Science Logic, {CSL} 2015}, volume~41 of {\em LIPIcs}, pages 24--40. Schloss Dagstuhl - Leibniz-Zentrum f{\"{u}}r Informatik, 2015.
\newblock \href {https://doi.org/10.4230/LIPIcs.CSL.2015.24} {\path{doi:10.4230/LIPIcs.CSL.2015.24}}.

\bibitem{menni:nominal}
Matias Menni.
\newblock About $\lnewsymb$-quantifiers.
\newblock {\em Applied categorical structures}, 11:421--445, 2003.

\bibitem{miller:pi}
Dale Miller.
\newblock The $\pi$-calculus as a theory in linear logic: Preliminary results.
\newblock In E.~Lamma and P.~Mello, editors, {\em Extensions of Logic Programming}, pages 242--264, Berlin, Heidelberg, 1993. Springer Berlin Heidelberg.

\bibitem{mil:uniform}
Dale Miller, Gopalan Nadathur, Frank Pfenning, and Andre Scedrov.
\newblock Uniform proofs as a foundation for logic programming.
\newblock {\em Annals of Pure and Applied Logic}, 51(1):125--157, 1991.
\newblock URL: \url{https://www.sciencedirect.com/science/article/pii/016800729190068W}, \href {https://doi.org/10.1016/0168-0072(91)90068-W} {\path{doi:10.1016/0168-0072(91)90068-W}}.

\bibitem{mil:pim:13}
Dale Miller and Elaine Pimentel.
\newblock A formal framework for specifying sequent calculus proof systems.
\newblock {\em Theoretical Computer Science}, 474:98--116, 2013.

\bibitem{mil:tiu:nabla}
Dale Miller and Alwen Tiu.
\newblock A proof theory for generic judgments.
\newblock {\em ACM Trans. Comput. Logic}, 6(4):749–783, oct 2005.
\newblock \href {https://doi.org/10.1145/1094622.1094628} {\path{doi:10.1145/1094622.1094628}}.

\bibitem{M80}
Robin Milner.
\newblock {\em A Calculus of Communicating Systems}, volume~92 of {\em Lecture Notes in Computer Science}.
\newblock Springer, 1980.
\newblock \href {https://doi.org/10.1007/3-540-10235-3} {\path{doi:10.1007/3-540-10235-3}}.

\bibitem{mil:par:wal:pi}
Robin Milner, Joachim Parrow, and David Walker.
\newblock A calculus of mobile processes, i.
\newblock {\em Information and Computation}, 100(1):1--40, 1992.
\newblock URL: \url{https://www.sciencedirect.com/science/article/pii/0890540192900084}, \href {https://doi.org/10.1016/0890-5401(92)90008-4} {\path{doi:10.1016/0890-5401(92)90008-4}}.

\bibitem{montesi:book}
Fabrizio Montesi.
\newblock {\em Introduction to Choreographies}.
\newblock Cambridge University Press, 2023.
\newblock \href {https://doi.org/10.1017/9781108981491} {\path{doi:10.1017/9781108981491}}.

\bibitem{tito:str:SIS-III}
Lê Thành~Dũng Nguyên and Lutz Straßburger.
\newblock {A System of Interaction and Structure III: The Complexity of BV and Pomset Logic}.
\newblock working paper or preprint, 2022.
\newblock URL: \url{https://hal.inria.fr/hal-03909547}.

\bibitem{tito:lutz:csl22}
Lê Thành~Dũng Nguyên and Lutz Straßburger.
\newblock {BV and Pomset Logic are not the same}.
\newblock In Florin Manea and Alex Simpson, editors, {\em 30th EACSL Annual Conference on Computer Science Logic (CSL 2022)}, volume 216 of {\em Leibniz International Proceedings in Informatics (LIPIcs)}, pages 3:1--3:17, Dagstuhl, Germany, 2022. Schloss Dagstuhl -- Leibniz-Zentrum f{\"u}r Informatik.
\newblock URL: \url{https://drops.dagstuhl.de/opus/volltexte/2022/15723}, \href {https://doi.org/10.4230/LIPIcs.CSL.2022.3} {\path{doi:10.4230/LIPIcs.CSL.2022.3}}.

\bibitem{pitts:nominal}
Andrew~M. Pitts.
\newblock Nominal logic, a first order theory of names and binding.
\newblock {\em Information and Computation}, 186(2):165--193, 2003.
\newblock Theoretical Aspects of Computer Software (TACS 2001).
\newblock URL: \url{https://www.sciencedirect.com/science/article/pii/S089054010300138X}, \href {https://doi.org/10.1016/S0890-5401(03)00138-X} {\path{doi:10.1016/S0890-5401(03)00138-X}}.

\bibitem{reddy:llstate}
Uday~S. Reddy.
\newblock A linear logic model of state.
\newblock 1993.
\newblock URL: \url{https://api.semanticscholar.org/CorpusID:10342053}.

\bibitem{retore:phd}
Christian Retor{\'e}.
\newblock {\em R{\'e}seaux et S{\'e}quents Ordonn{\'e}s}.
\newblock PhD thesis, Universit{\'e} Paris VII, 1993.

\bibitem{ret:newPomset}
Christian Retor{\'e}.
\newblock Pomset logic: The other approach to noncommutativity in logic.
\newblock {\em Joachim Lambek: The Interplay of Mathematics, Logic, and Linguistics}, pages 299--345, 2021.

\bibitem{ret:flag}
Christian Retoré.
\newblock Flag: a self-dual modality for non-commutative contraction and duplication in the category of coherence spaces.
\newblock {\em Electronic Proceedings in Theoretical Computer Science}, 353:157–174, December 2021.
\newblock URL: \url{http://dx.doi.org/10.4204/EPTCS.353.8}, \href {https://doi.org/10.4204/eptcs.353.8} {\path{doi:10.4204/eptcs.353.8}}.

\bibitem{rov:bind}
Luca Roversi.
\newblock A deep inference system with a self-dual binder which is complete for linear lambda calculus.
\newblock {\em Journal of Logic and Computation}, 26(2):677--698, 2016.
\newblock \href {https://doi.org/10.1093/logcom/exu033} {\path{doi:10.1093/logcom/exu033}}.

\bibitem{ruet:noncomII}
Paul Ruet.
\newblock Non-commutative logic {II}: sequent calculus and phase semantics.
\newblock {\em Mathematical Structures in Computer Science}, 10(2):277–312, 2000.
\newblock \href {https://doi.org/10.1017/S0960129599003084} {\path{doi:10.1017/S0960129599003084}}.

\bibitem{sangiorgi:internal}
Davide Sangiorgi.
\newblock $\pi$-calculus, internal mobility, and agent-passing calculi.
\newblock {\em Theoretical Computer Science}, 167(1):235--274, 1996.
\newblock URL: \url{https://www.sciencedirect.com/science/article/pii/0304397596000758}, \href {https://doi.org/10.1016/0304-3975(96)00075-8} {\path{doi:10.1016/0304-3975(96)00075-8}}.

\bibitem{slav:pomset}
Sergey Slavnov.
\newblock {On noncommutative extensions of linear logic}.
\newblock {\em {Logical Methods in Computer Science}}, {Volume 15, Issue 3}, September 2019.
\newblock URL: \url{https://lmcs.episciences.org/5774}, \href {https://doi.org/10.23638/LMCS-15(3:30)2019} {\path{doi:10.23638/LMCS-15(3:30)2019}}.

\bibitem{strassburger:problem}
Lutz Stra{\ss}burger.
\newblock The problem of proof identity, and why computer scientists should care about hilbert's 24th problem.
\newblock {\em Philosophical Transactions of the Royal Society A}, 377(2140):20180038, 2019.
\newblock \href {https://doi.org/10.1098/rsta.2018.0038} {\path{doi:10.1098/rsta.2018.0038}}.

\bibitem{tiu:SIS-II}
Alwen~Fernanto Tiu.
\newblock A system of interaction and structure {II}: {T}he need for deep inference.
\newblock 2(2):1--24, 2006.
\newblock \href {https://doi.org/10.2168/LMCS-2(2:4)2006} {\path{doi:10.2168/LMCS-2(2:4)2006}}.

\bibitem{troelstra_schwichtenberg_2000}
A.~S. Troelstra and H.~Schwichtenberg.
\newblock {\em Basic Proof Theory}.
\newblock Cambridge Tracts in Theoretical Computer Science. Cambridge University Press, 2 edition, 2000.
\newblock \href {https://doi.org/10.1017/CBO9781139168717} {\path{doi:10.1017/CBO9781139168717}}.

\bibitem{vasco:pi}
Vasco~T. Vasconcelos.
\newblock Fundamentals of session types.
\newblock {\em Information and Computation}, 217:52--70, 2012.
\newblock URL: \url{https://www.sciencedirect.com/science/article/pii/S0890540112001022}, \href {https://doi.org/10.1016/j.ic.2012.05.002} {\path{doi:10.1016/j.ic.2012.05.002}}.

\bibitem{winskel:event}
Glynn Winskel.
\newblock Event structures.
\newblock In W.~Brauer, W.~Reisig, and G.~Rozenberg, editors, {\em Petri Nets: Applications and Relationships to Other Models of Concurrency}, pages 325--392, Berlin, Heidelberg, 1987. Springer Berlin Heidelberg.

\end{thebibliography}

\clearpage

\appendix

\begin{figure}
	$$
	\begin{array}{r@{\;}c@{\;}ll}
		\pnil & \alphaeq & \pnil &
	\\
		\precv xy P &			\alphaeq 		& \precv xz P \fsubst zy
		&
		\mbox{$z$ fresh for $P$}
	\\
		\psend xy P &			\alphaeq		& \psend xy Q
		& \mbox{if $P \alphaeq Q$}
	\\
		P \ppar Q 		&		\alphaeq									&
		R \ppar S &
		\mbox{if $P \alphaeq R$ and $Q \alphaeq S$}
	\\
		\pnu x P	 &			\alphaeq 		& \pnu u P \fsubst ux
		&
		\mbox{$u$ fresh for $P$}
	\\
		\pbras x{\lab:P_\lab}{\lab\in L} 		&	\alphaeq		&
		\pbras x{\lab:Q_\lab}{\lab\in L}						&
		\mbox{if $P_{\lab} \alphaeq Q_{\lab}$ for all $\lab \in L$}
	\\
		\psels x{\lab:P_\lab}{\lab\in L}			& 	\alphaeq		&
		\psels x{\lab:Q_\lab}{\lab\in L} &
		\mbox{if $P_{\lab} \alphaeq Q_{\lab}$ for all $\lab \in L$}
	\end{array}
	$$
	\caption{Definition of $\alphaeq$ for \procs.}
	\label{fig:alphaEq}
\end{figure}


\section{Embedding $\NL$ into $\MAVq$}\label{sec:embedding}
\def\mavform{$\MAVq$-formula\xspace}
\def\mavforms{$\MAVq$-formulas\xspace}

\begin{figure}[t]
	\adjustbox{max width=\textwidth}{$\begin{array}{c|c}
			\begin{array}{c}
				\mbox{$\MAVq$-Formulas}
				\\
				\begin{array}{r@{\;}c@{}c@{\mid}c@{\mid}c@{\mid}c@{\mid}c@{\mid}cc}
					A,B &\coloneqq 	&
					\lunit & A \lpar B & A \lseq B & A \ltens B & A \lplus B & A \lwith B
					\\	&\mid	&
					\swn[x]A &\soc[x]A&\lFa xA & \lNu xA& \lWe xA& \lEx xA
				\end{array}
				\\\hline
				\mbox{Formula equivalences}
				\\
				A=A\lpar \lunit = A\lseq \lunit =  \lunit\lseq  A = A \ltens \lunit
				\\
				\lunit = \lunit \lplus \lunit = \lunit \lwith \lunit = \lFa x\lunit = \lNu x\lunit = \lWe x\lunit = \lEx x\lunit
				\\
				A\circleddot B = B\circleddot A \mbox{ with } \circleddot\in\set{\lpar,\ltens,\lplus,\lwith}
				\\
				A \circleddot (B\circleddot C)=(A\circleddot B)\circleddot C
				\mbox{ with } \circleddot\in\set{\lpar,\ltens,\lplus,\lwith}
				\\
				A\lseq (B\lseq C) = (A\lseq B)\lseq C
				\\\hline
				\mbox{Derivations}
				\\
				\dD
				\coloneqq
				\;
				\begin{array}{c|c|c|cc}
					A
					&
					\vlderivation{\vlde{\dD_2}{}{B}{\vlin{\rrule}{}{C}{\vlde{\dD_1}{}{C'}{\vlhy{A}}}}}
					&
					\od{\odd{\odh{A_1}}{\dD_1}{B_1}{}}
					\circleddot
					\od{\odd{\odh{A_1}}{\dD_1}{B_1}{}}
					&
					\lqusymb x.\od{\odd{\odh{A_1}}{\dD_1}{B_1}{}}
					&\mbox{where }
					\od{\odp{\dD_1}{A}{}}\coloneqq\od{\odd{\odh{\lunit}}{\dD_1}{A}{}}
				\end{array}
				\\
				\mbox{with }A,B,C,C'\mbox{ formulas}
				\mbox{and } \rrule \mbox{ rule }
				\mbox{and }
				\\
				\circleddot\in\Set{\lpar,\ltens,\lseq,\lplus,\lwith}
				\quad
				\mbox{and } \lqusymb\in\Set{\exists,\forall,\lnewsymb,\lwensymb}
			\end{array}
			&
			\begin{array}{c}
				\mbox{Rules}
				\\
				\vlinf{\aidr}{}{a\lpar \cneg a}{\lunit}
				\qquad
				\vlinf{\swir}{}{(A\ltens B)\lpar C}{A \ltens (B\lpar C)}
				\qquad
				\vlinf{\qdr}{}{(A\lseq C)\lpar(B\lseq D)}{(A\lpar B) \lseq (C\lpar D)}
				\\\hline
				\vlinf{\lwith}{}{A\lpar(B\lwith C)}{(A\lpar B) \lwith (A\lpar C)}
				\qquad
				\vlinf{\lplus}{}{A\lplus B}{A}
				\qquad
				\vlinf{\lseq\mhyphen\lwith}{}{(A\lseq C)\lwith(B\lseq D)}{(A\lwith B) \lseq (C\lwith D)}
				\\\hline
				\vlinf{\scope[\forall]}{}{(\lFa xA)\lpar F}{\lFap x{A\lpar F}}
				\qquad
				\vlinf{\lseq\mhyphen\forall}{}{\lFap x{A\lseq B}}{(\lFa x A) \lseq (\lFa xB)}
				\qquad
				\vlinf{\exists}{}{\lEx xA}{A\fsubsts{c/x}}
				\\\hline
				\vlinf{\scope[\lqusymb]}{}{(\lQu xA)\lpar F}{\lQup x{A\lpar F}}
				\qquad
				\vlinf{\scope[\lwensymb]^\lseq}{}{\lWe xA \lseq F}{\lWep x {A\lseq F}}
				\qquad
				\vlinf{\scope[\lwensymb]^\lseq}{}{ F\lseq \lWe xA }{\lWep x {F\lseq A}}
				\\
				\vlinf{\shiftr}{}{\lNu x\lWe xA}{\lWe x\lNu yA}
				\qquad
				\vlinf{\lnewsymb\mhyphen\lwensymb}{}{(\lNu xA) \lpar (\lWe xB)}{\lNup x{A\lpar B}}
				\qquad
				\vlinf{\lwensymb}{}{\lWe xA}{\lNu xA}
				\\
				\vlinf{\shiftr_\forall}{}{\lFa x\lQu xA}{\lQu x\lFa yA}
				\qquad
				\vlinf{\scope[\lqusymb]^\lwith}{}{(\lQu xA) \lwith F}{\lQup x {A\lwith F}}
				\qquad
				\vlinf{\mathsf{nom-choice}}{}{\lQu xA \lwith \lQu xB}{\lQup x {A\lwith B}}
				\\\hline
				\vlinf{\lseq\mhyphen\lnewsymb}{}{\lNup x{A\lseq B}}{(\lNu x A) \lseq (\lNu xB)}
				\qquad
				\vlinf{\lwensymb\mhyphen\lseq}{}{\lWe xA \lseq \lWe xB}{\lWep x{A\lseq B}}
				\\
				\mbox{where $\lqusymb\in\set{\lnewsymb,\lwensymb}$ and $\circleddot\in\set{\lpar,\lseq}$ and $x\notin\freeof F$}
				\\\hline\hline
				\mbox{Systems}
				\\
				\begin{array}{ll}
					\BV=\Set{\aidr,\swir,\qdr}
					&
					\MAV=\BV\cup\Set{\lplus,\lwith,\lseq\mhyphen\lwith}
					\\
					\BVq=\MAVq\setminus\Set{\lseq\mhyphen\lnewsymb,\lwensymb\mhyphen\lseq}
					&
					\MAVq=\mbox{all rules above}
				\end{array}
			\end{array}
		\end{array}$}
	\caption{
		Inductive definition of deep inference derivation and the rules in the system $\MAVq$.
	}
	\label{fig:deepMAV}
\end{figure}

We recall in \Cref{fig:deepMAV} the definition of \defn{\mavforms}, formula equivalence, and deep inference derivations for $\MAVq$.
Rules have been reorganized, also relying on a stronger formula equivalence capturing derivable equivalences involving additive connectives, to improve readability over the intuitive reading of the formula-as-process interpretation.

\thmMAVembed*
\begin{proof}
	For each derivation $\dD$ in $\NL$ conclusion $A_1,\ldots, A_n$,
	we define a deep-inference derivation $\trbv{\dD}$ in $\MAVq$ with premise $\lunit$ and conclusion $\bigparr_{i=1}^n\trbv{A_i}$ as shown in \Cref{fig:embedding}.
\end{proof}

\begin{figure}[t]
	\adjustbox{max width=\textwidth}{$\begin{array}{c}
		\trbv{\vlinf{\lunit}{}{\sdash \lunit}{}}
		=
		\lunit
	\qquad
		\trbv{\vlinf{\axrule}{}{\sdash \lsend xy, \lrecv xy}{}}
		=
		\odn{\lunit}{\aidr}{\lsend xy \lpar \lrecv xy}{}
	\qquad
		\trbv{
				\vlderivation{
						\vlin{\rrule[1]}{}{\sdash \Gamma , A}{
								\vlpr{\dD'}{}{\sdash \Gamma , A'}
							}
					}
			}
		=
		\od{
				\odp{\trbv{\dD_1}}{
						\trbv{\Gamma}\lpar
						\odn{A'}{\rrule[1]}{A}{}
					}{}
			}
	\\
		\trbv{\vlderivation{
						\vliin{\mixr}{}{\sdash \Gamma, \Delta}{\vlpr{}{}{\sdash \Gamma} }{\vlpr{}{}{\sdash \Delta}}
					}
			}
		=
		\od{\odp{\trbv{\dD_1}}{\trbv\Gamma}{}}
		\lpar
		\od{\odp{\trbv{\dD_2}}{\trbv\Delta}{}}
	\qquad
		\trbv{\vlderivation{
						\vliin{\ltens}{}{\sdash \Gamma, A\ltens B, \Delta}{\vlpr{}{}{\sdash \Gamma,A} }{\vlpr{}{}{\sdash B,\Delta}}
					}
			}
		=
		\odN{
				\od{\odp{\trbv{\dD_1}}{\trbv\Gamma \lpar \trbv A}{}}
				\ltens
				\od{\odp{\trbv{\dD_2}}{\trbv B \lpar \trbv\Delta}{}}
			}{\swir}{
				\trbv\Gamma
				\lpar
				(\trbv A \ltens \trbv B )
				\lpar
				\trbv\Delta
			}{}
	\\
		\trbv{\vlderivation{
						\vliin{\lwith}{}{\sdash \Gamma, A\lwith B}{\vlpr{}{}{\sdash \Gamma,A} }{\vlpr{}{}{\sdash B,\Delta}}
					}
			}
		=
		\odn{
				\od{\odp{\trbv{\dD_1}}{\trbv\Gamma \lpar \trbv A}{}}
				\lwith
				\od{\odp{\trbv{\dD_2}}{\trbv B \lpar \trbv\Gamma}{}}
			}{\lwith}{
				\trbv\Gamma
				\lpar
				(\trbv A \lwith\trbv B )
			}{}
	\qquad
		\trbv{\vlderivation{
						\vliin{\precur}{}{\sdash \Gamma, A\lprec B}{\vlpr{}{}{\sdash \Delta,A} }{\vlpr{}{}{\sdash B,\Delta}}
					}
			}
		=
		\odn{
				\od{\odp{\trbv{\dD_1}}{\trbv\Gamma \lpar \trbv A}{}}
				\lseq
				\od{\odp{\trbv{\dD_2}}{\trbv B \lpar \trbv\Delta}{}}
			}{\qdr}{
				\odn{\trbv\Gamma \lseq \trbv\Delta}{\qdr}{\trbv\Gamma \lpar \trbv\Delta}{}
				\lpar
				(\trbv A \lseq \trbv B)
			}{}
	\\
		\trbv{\vlderivation{
						\vliin{\lprec}{}{\sdash \Gamma, A\lprec B, C\lprec D}{\vlpr{}{}{\sdash \Delta,A,C} }{\vlpr{}{}{\sdash B,D,\Delta}}
					}
			}
		=
		\odn{
				\od{\odp{\trbv{\dD_1}}{\trbv\Gamma \lpar (\trbv A \lpar \trbv D)}{}}
				\lseq
				\od{\odp{\trbv{\dD_2}}{(\trbv B \lpar \trbv C)\lpar \trbv\Delta}{}}
			}{\qdr}{
				\odn{\trbv\Gamma \lseq \trbv\Delta}{\qdr}{\trbv\Gamma \lpar \trbv\Delta}{}
				\lpar
				\odn{\left((\trbv A \lpar \trbv C )\lseq(\trbv B \lpar \trbv D )\right)}{\qdr}{
						(\trbv A \lseq \trbv B )\lpar (\trbv C \lseq \trbv D )
					}{}
			}{}
	\\
		\trbv{\vlderivation{
						\vlin{\lqusymb}{}{\sdash \Gamma, \lQu xA}{\vlpr{\dD_1}{}{\sdash \Gamma,A}}
					}
			}
		=
		\odn{
				\lQup x{\od{\odp{\trbv{\dD_1}}{\trbv{\Gamma} \lpar \trbv{A}}{}}}
			}{\scope[\lqusymb']}{
				\trbv{\Gamma} \lpar \lQu x\trbv{A}
			}{}
	\qquad
		\trbv{\vlderivation{
						\vlin{\nqsrule}{}{\sdash \Gamma, \lNu xA,\lYa xB}{\vlpr{\dD_1}{}{\sdash \Gamma,A\fsubst yx,B\fsubst yx}}
					}
			}
		=
		\odn{
				\lNup x{\od{\odp{\trbv{\dD_1}}{\trbv{\Gamma} \lpar \trbv{A} \lpar \trbv{B}}{}}}
			}{\scope[\lnewsymb]}{
				\trbv{\Gamma} \lpar
				\odn{
						\lNup x{\trbv{A} \lpar \trbv{B}}
					}{\lnewsymb\mhyphen\lwensymb}{
						(\lNu x{\trbv{A}} )\lpar (\lWe x{\trbv{B}})
					}{}
		}{}
	\end{array}$}
	\caption{How to define a derivation $\trbv \dD$ in $\MAVq$ from a derivation $\dD$ in $\NL$, with $\rrule[1]\in\set{\lpar,\lplus,\exists}$ and $\lqusymb\in\set{\forall,\lnewsymb,\lyasymb}$ and $\lqusymb'=\lqusymb$ except if $\lqusymb=\lyasymb$, in which case $\lqusymb'=\lwensymb$.}
	\label{fig:embedding}
\end{figure}

\begin{remark}\label{rem:embedding}
	In $\NML$ we have the same connectives $\lpar$ and $\ltens$ of $\BV$, as well as a non-commutative self-dual connective $\lprec$, all sharing the same unit $\lunit$.
	Moreover, as in $\BV$, the implications
	$(A\ltens B) \limp (A\lprec B)$
	and
	$(A\lprec B)\limp (A\lpar B)$
	hold, as well as the ones proving that $\lunit$ is the unit for the three connectives $\lpar$, $\lprec$, and $\ltens$.

	However, we know from \cite{tiu:SIS-II} that $\BV$ cannot have $\cutr$-free a sequent calculus, and the same holds for Retore's $\pomset$%
	\footnote{
		Note that a $\cutr$-free sequent calculus for $\pomset$ has been proposed in \cite{slav:pomset}, but the side conditions of its sequent rules cannot be checked in polynomial time. Therefore such a sequent system cannot be considered a proper proof system, as intended in \cite{cook:reckhow:79}.
	}
	which is a proper conservative extension of $\BV$ \cite{tito:lutz:csl22,tito:str:SIS-III}.
	Thus we conjecture that the cause of the impossibility of having a $\cutr$-free sequent calculus for Guglielmi's $\BV$ \cite{gug:SIS} in the associativity of the connective $\lseq$.
\end{remark}

\section{Confluence of coalescence}\label{sec:succonet}

\confCoal*
\begin{proof}
\def\nowDer{
	\newline The two distinct derivations labeling the link in the bottom-right corner of the diagram according to the anticlockwise and clockwise sequence of coalescence steps are:
}

	We only discuss the critical pairs for coalescence rules not already discussed in \cite{hei:hug:conflict}.
	Together with the confluence diagram of each critical pair, we show the two derivations corresponding to the two sequences of coalescence steps.
	\begin{itemize}
		\item Case $\lpar/\lprec$:

		$$
		\begin{array}{ccc}
			\viA1 \vpz1{\lprec} \viB1,\viA2 \vpz2{\lprec} \viB2, \viA3 \vpz3{\lpar} \viA4, \vpz4{\Gamma},\vpz5{\Delta}
			\pzlinks{A1/A2/12/\la/red/{A3,A4,pz4}}
			\pzlinks{B1/B2/-12/\lb/blue/{pz5}}
			&\to&
			\viA1 \vpz1{\lprec} \viB1,\viA2 \vpz2{\lprec} \viB2,\viA3 \vpz3{\lpar} \viA4, \vpz4{\Gamma},\vpz5{\Delta}
			\pzlinks{pz1/pz2/12/\labb/pzgreen/{A3,A4,pz4,pz5}}
			\\\\
			\downarrow && \downarrow
			\\\\
			\viA1 \vpz1{\lprec} \viB1,\viA2 \vpz2{\lprec} \viB2, \viA3 \vpz3{\lpar} \viA4, \vpz4{\Gamma},\vpz5{\Delta}
			\pzlinks{A1/A2/12/\laa/red/{pz3,pz4}}
			\pzlinks{B1/B2/-12/\lb/blue/{pz5}}
			&\to&
			\viA1 \vpz1{\lprec} \viB1,\viA2 \vpz2{\lprec} \viB2, \viA3 \vpz3{\lpar} \viA4, \vpz4{\Gamma},\vpz5{\Delta}
			\pzlinks{pz1/pz2/12/\laabb/magenta/{pz3,pz4,pz5}}
		\end{array}
		$$

		With $\dualizerof[\laa]= \dualizerof[\la]$ and $\dualizerof[\labb] = \dualizerof[\laabb] = \dualizerof[\la] \duasum \dualizerof[\lb]$.
		\nowDer

		$$
		\vlderivation{
			\vlin{\lpar}{}{
				\sdash A_1\lprec B_1,A_2\lprec B_2,A_3\lpar A_4,\Gamma, \Delta
			}{
				\vliin{\lprec}{}{
					\sdash A_1\lprec B_1,A_2\lprec B_2,A_3, A_4,\Gamma, \Delta
				}{
					\vlhy{\sdash A_1,A_2,A_3\lpar A_4,\Gamma}
				}{
					\vlhy{\sdash B_1,B_2,\Delta}
				}
			}
		}
		\quad\peq\quad
		\vlderivation{
			\vliin{\lprec}{}{\sdash A_1\lprec B_1,A_2\lprec B_2,A_3\lpar A_4,\Gamma, \Delta}{
				\vlin{\lpar}{}{
					\sdash A_1,A_2,A_3\lpar A_4,\Gamma
				}{
					\vlhy{\sdash A_1,A_2,A_3, A_4,\Gamma}
				}
			}{
				\vlhy{\sdash B_1,B_2,\Delta}
			}
		}
		$$

		\item Case $\lplus/\lprec$:
		$$
		\begin{array}{ccc}
			\viA1 \vpz1{\lprec} \viB1,\viA2 \vpz2{\lprec} \viB2, \viA3 \vpz3{\oplus} \viA4, \vpz4{\Gamma},\vpz5{\Delta}
			\pzlinks{A1/A2/12/\la/red/{A3,pz4}}
			\pzlinks{B1/B2/-12/\lb/blue/{pz5}}
			&\to&
			\viA1 \vpz1{\lprec} \viB1,\viA2 \vpz2{\lprec} \viB2,\viA3 \vpz3{\oplus} \viA4, \vpz4{\Gamma},\vpz5{\Delta}
			\pzlinks{pz1/pz2/12/\labb/pzgreen/{A3,pz4,pz5}}
			\\\\
			\downarrow && \downarrow
			\\\\
			\viA1 \vpz1{\lprec} \viB1,\viA2 \vpz2{\lprec} \viB2, \viA3 \vpz3{\oplus} \viA4, \vpz4{\Gamma},\vpz5{\Delta}
			\pzlinks{A1/A2/12/\laa/red/{pz3,pz4}}
			\pzlinks{B1/B2/-12/\lb/blue/{pz5}}
			&\to&
			\viA1 \vpz1{\lprec} \viB1,\viA2 \vpz2{\lprec} \viB2, \viA3 \vpz3{\oplus} \viA4, \vpz4{\Gamma},\vpz5{\Delta}
			\pzlinks{pz1/pz2/12/\laabb/magenta/{pz3,pz4,pz5}}
		\end{array}
		$$

		With $\dualizerof[\laa]= \dualizerof[\la]$ and $\dualizerof[\labb] = \dualizerof[\laabb] = \dualizerof[\la] \duasum \dualizerof[\lb]$.
		\nowDer
		$$
		\vlderivation{
			\vlin{\lplus}{}{
				\sdash A_1\lprec B_1,A_2\lprec B_2,A_3\lplus A_4,\Gamma, \Delta
			}{
				\vliin{\lprec}{}{
					\sdash A_1\lprec B_1,A_2\lprec B_2,A_3,\Gamma, \Delta
				}{
					\vlhy{\sdash A_1,A_2,A_3,\Gamma}
				}{
					\vlhy{\sdash B_1,B_2,\Delta}
				}
			}
		}
		\quad\peq\quad
		\vlderivation{
			\vliin{\lprec}{}{\sdash A_1\lprec B_1,A_2\lprec B_2,A_3\lplus A_4,\Gamma, \Delta}{
				\vlin{\lplus}{}{
					\sdash A_1,A_2,A_3,\Gamma
				}{
					\vlhy{\sdash A_1,A_2,A_3,\Gamma}
				}
			}{
				\vlhy{\sdash B_1,B_2,\Delta}
			}
		}
		$$

		\item Case $\lprec/\lprec$:

		$$
		\adjustbox{max width=\textwidth}{$
		\begin{array}{ccc}
			\viA1 \vpz1{\lprec} \viB1,\viA2 \vpz2{\lprec} \viB2, \viA3 \vpz3{\lprec} \viC1, \viA4 \vpz4{\lprec} \viC2, \vpz5{\Gamma},\vpz6{\Delta}, \vpz7{\Sigma}
			\pzlinks{A1/A2/12/\la/red/{A3,pz5}}
			\pzlinks{B1/B2/-12/\lb/blue/{pz6}}
			\pzlinks{C1/C2/-18/\lc/magenta/{pz7}}
			&\to&
			\viA1 \vpz1{\lprec} \viB1,\viA2 \vpz2{\lprec} \viB2, \viA3 \vpz3{\lprec} \viC1, \viA4 \vpz4{\lprec} \viC2, \vpz5{\Gamma},\vpz6{\Delta}, \vpz7{\Sigma}
			\pzlinks{pz1/pz2/12/\labb/pzgreen/{A3,A4,pz5,pz6}}
			\pzlinks{C1/C2/-12/\lc/magenta/{pz7}}
			\\\\
			\downarrow && \downarrow
			\\\\
			\viA1 \vpz1{\lprec} \viB1,\viA2 \vpz2{\lprec} \viB2, \viA3 \vpz3{\lprec} \viC1, \viA4 \vpz4{\lprec} \viC2, \vpz5{\Gamma},\vpz6{\Delta}, \vpz7{\Sigma}
			\pzlinks{A1/A2/12/\lac/violet/{pz3,pz4,pz5,pz7}}
			\pzlinks{B1/B2/-12/\lb/blue/{pz6}}
			&\to&
			\viA1 \vpz1{\lprec} \viB1,\viA2 \vpz2{\lprec} \viB2, \viA3 \vpz3{\lprec} \viC1, \viA4 \vpz4{\lprec} \viC2, \vpz5{\Gamma},\vpz6{\Delta}, \vpz7{\Sigma}
			\pzlinks{pz1/pz2/12/\labc/brown/{pz3,pz4,pz5,pz6,pz7}}
		\end{array}
		$}$$

		With $\dualizerof[\labb] = \dualizerof[\la] \duasum \dualizerof[\lb]$, $\dualizerof[\lac] = \dualizerof[\la] \duasum \dualizerof[\lc]$
		and $\dualizerof[\labc] = \dualizerof[\la] \duasum \dualizerof[\lb] \duasum \dualizerof[\lc]$.
		\nowDer
		$$\adjustbox{max width=\textwidth}{$
		\vlderivation{
			\vliin{\lprec}{}{
				\sdash A_1\lprec B_1,A_2\lprec B_2,A_3\lprec C_1,A_4\lprec C_2,\Gamma, \Delta,\Sigma
			}{
				\vliin{\lprec}{}{
					\sdash A_1,A_2,A_3\lprec C_1,A_4\lprec C_2,\Gamma,\Sigma
				}{
					\vlhy{\sdash A_1,A_2,A_3,A_4,\Gamma}
				}{
					\vlhy{\sdash C_1,C_2,\Sigma}
				}
			}{
				\vlhy{\sdash B_1,B_2,\Delta}
			}
		}
		\quad\peq\quad
		\vlderivation{
			\vliin{\lprec}{}{
				\sdash A_1\lprec B_1,A_2\lprec B_2,A_3\lprec C_1,A_4\lprec C_2,\Gamma, \Delta,\Sigma
			}{
				\vliin{\lprec}{}{
					\sdash A_1\lprec B_1,A_2\lprec B_2,A_3,A_4,\Gamma, \Delta
				}{
					\vlhy{\sdash A_1,A_2,A_3,A_4,\Gamma}
				}{
					\vlhy{\sdash B_1,B_2,\Delta}
				}
			}{
				\vlhy{\sdash C_1,C_2,\Sigma}
			}
		}
		$}$$

		\item Case $\ltens/\lprec$:
		$$\adjustbox{max width=\textwidth}{$
			\begin{array}{ccc}
				\\
				\viA1 \vpz1{\lprec} \viB1,\viA2 \vpz2{\lprec} \viB2, \viA3 \vpz3{\ltens} \vC1, \vpz4{\Gamma}, \vpz5{\Delta}, \vpz6{\Sigma}
				\pzlinks{A1/A2/12/\la/red/{A3,pz4}}
				\pzlinks{B1/B2/-12/\lb/blue/{pz5}}
				\pzlinks{C1/pz6/-18/\lc/violet/}
				&\to&
				\viA1 \vpz1{\lprec} \viB1,\viA2 \vpz2{\lprec} \viB2, \viA3 \vpz3{\ltens} \vC1, \vpz4{\Gamma}, \vpz5{\Delta}, \vpz6{\Sigma}
				\pzlinks{pz1/pz2/12/\labb/pzgreen/{A3,pz4,pz5}}
				\pzlinks{C1/pz6/-12/\lc/violet/}
				\\\\
				\downarrow && \downarrow
				\\\\
				\viA1 \vpz1{\lprec} \viB1,\viA2 \vpz2{\lprec} \viB2, \viA3 \vpz3{\ltens} \vC1, \vpz4{\Gamma}, \vpz5{\Delta}, \vpz6{\Sigma}
				\pzlinks{A1/A2/12/\lac/orange/{pz3,pz4,pz6}}
				\pzlinks{B1/B2/-12/\lb/blue/{pz5}}
				&\to&
				\viA1 \vpz1{\lprec} \viB1,\viA2 \vpz2{\lprec} \viB2, \viA3 \vpz3{\ltens} \vC1, \vpz4{\Gamma}, \vpz5{\Delta}, \vpz6{\Sigma}
				\pzlinks{pz1/pz2/12/\labc/brown/{pz3,pz4,pz5,pz6}}
				\\\\
			\end{array}
			$}$$

		With $\dualizerof[\labb] = \dualizerof[\la] \duasum \dualizerof[\lb]$, $\dualizerof[\lac] = \dualizerof[\la] \duasum \dualizerof[\lc]$
		and $\dualizerof[\labc] = \dualizerof[\la] \duasum \dualizerof[\lb] \duasum \dualizerof[\lc]$.
		\nowDer
		$$\adjustbox{max width=\textwidth}{$
			\vlderivation{
				\vliin{\ltens}{}{
					\sdash A_1\lprec B_1,A_2\lprec B_2,A_3\ltens C,\Gamma, \Delta,\Sigma
				}{
					\vliin{\lprec}{}{
						\sdash A_1,A_2,A_3\ltens C,\Gamma,\Sigma
					}{
						\vlhy{\sdash A_1,A_2,A_3,A_4,\Gamma}
					}{
						\vlhy{\sdash B_1,B_2,\Delta}
					}
				}{
					\vlhy{\sdash C,\Sigma}
				}
			}
			\quad\peq\quad
			\vlderivation{
				\vliin{\lprec}{}{
					\sdash A_1\lprec B_1,A_2\lprec B_2,A_3\ltens C,\Gamma, \Delta,\Sigma
				}{
					\vliin{\ltens}{}{
						\sdash A_1\lprec B_1,A_2\lprec B_2,A_3,A_4,\Gamma, \Delta
					}{
						\vlhy{\sdash A_1,A_2,A_3,A_4,\Gamma}
					}{
						\vlhy{\sdash C,\Sigma}
					}
				}{
					\vlhy{\sdash B_1,B_2,\Delta}
				}
			}
		$}$$

		\item Case $\exists/\lprec$ :
		$$
		\begin{array}{ccc}
			\viA1 \vpz1{\lprec} \viB1,\viA2 \vpz2{\lprec} \viB2, \vpz3{\exists}x.\vC1, \vpz4{\Gamma}, \vpz5{\Delta}
			\pzlinks{A1/A2/12/\la/red/{C1,pz4}}
			\pzlinks{B1/B2/-12/\lb/blue/{pz5}}
			&\to&
			\viA1 \vpz1{\lprec} \viB1,\viA2 \vpz2{\lprec} \viB2, \vpz3{\exists}x.\vC1, \vpz4{\Gamma}, \vpz5{\Delta}
			\pzlinks{pz1/pz2/12/\labb/pzgreen/{C1,pz4,pz5}}
			\\\\
			\downarrow && \downarrow
			\\\\
			\viA1 \vpz1{\lprec} \viB1,\viA2 \vpz2{\lprec} \viB2, \vpz3{\exists}x.\vC1, \vpz4{\Gamma}, \vpz5{\Delta}
			\pzlinks{A1/A2/12/\laa/violet/{pz3,pz4}}
			\pzlinks{B1/B2/-12/\lb/blue/{pz5}}
			&\to&
			\viA1 \vpz1{\lprec} \viB1,\viA2 \vpz2{\lprec} \viB2, \vpz3{\exists}x.\vC1, \vpz4{\Gamma}, \vpz5{\Delta}
			\pzlinks{pz1/pz2/12/\laabb/brown/{pz3,pz4,pz5}}
		\end{array}
		$$

		With $\dualizerof[ab] = \dualizerof[\la] \duasum \dualizerof[\lb]$, $\dualizerof[\laa] = \dualizerof[\la]\fsubminus{x}$
		and $\dualizerof[a'b] = \dualizerof[\labb] \fsubminus{x}$
		\nowDer
		$$
		\vlderivation{
			\vlin{\exists}{}{
				\sdash A_1\lprec B_1,A_2\lprec B_2,\exists x C,\Gamma, \Delta
			}{
				\vliin{\lprec}{}{
					\sdash A_1\lprec B_1,A_2\lprec B_2,C \fsubst{c}{x},\Gamma, \Delta
				}{
					\vlhy{\sdash A_1,A_2,C \fsubst{c}{x},\Gamma}
				}{
					\vlhy{\sdash B_1,B_2,\Delta}
				}
			}
		}
		\quad\peq\quad
		\vlderivation{
			\vliin{\lprec}{}{\sdash A_1\lprec B_1,A_2\lprec B_2, \exists xC , \Gamma, \Delta}{
				\vlin{\exists}{}{
					\sdash A_1,A_2,\exists xC ,\Gamma
				}{
					\vlhy{\sdash A_1,A_2,C \fsubst{c}{x} ,\Gamma}
				}
			}{
				\vlhy{\sdash B_1,B_2,\Delta}
			}
		}
		$$

		\item Case $\forall/\lprec$:
		$$
		\begin{array}{ccc}
			\viA1 \vpz1{\lprec} \viB1,\viA2 \vpz2{\lprec} \viB2, \vpz3{\forall}x.\vC1, \vpz4{\Gamma}, \vpz5{\Delta}
			\pzlinks{A1/A2/12/\la/red/{C1,pz4}}
			\pzlinks{B1/B2/-12/\lb/blue/{pz5}}
			&\to&
			\viA1 \vpz1{\lprec} \viB1,\viA2 \vpz2{\lprec} \viB2, \vpz3{\forall}x.\vC1, \vpz4{\Gamma}, \vpz5{\Delta}
			\pzlinks{pz1/pz2/12/\labb/pzgreen/{C1,pz4,pz5}}
			\\\\
			\downarrow && \downarrow
			\\\\
			\viA1 \vpz1{\lprec} \viB1,\viA2 \vpz2{\lprec} \viB2, \vpz3{\forall}x.\vC1, \vpz4{\Gamma}, \vpz5{\Delta}
			\pzlinks{A1/A2/12/\la/violet/{pz3,pz4}}
			\pzlinks{B1/B2/-12/\lb/blue/{pz5}}
			&\to&
			\viA1 \vpz1{\lprec} \viB1,\viA2 \vpz2{\lprec} \viB2, \vpz3{\forall}x.\vC1, \vpz4{\Gamma}, \vpz5{\Delta}
			\pzlinks{pz1/pz2/12/\labb/brown/{pz3,pz4,pz5}}
		\end{array}
		$$

		With $\dualizerof[\labb] = \dualizerof[\la] \duasum \dualizerof[\lb]$.
		\nowDer
		$$
		\vlderivation{
			\vlin{\forall}{}{
				\sdash A_1\lprec B_1,A_2\lprec B_2,\lFa x C,\Gamma, \Delta
			}{
				\vliin{\lprec}{}{
					\sdash A_1\lprec B_1,A_2\lprec B_2,C  ,\Gamma, \Delta
				}{
					\vlhy{\sdash A_1,A_2,C ,\Gamma}
				}{
					\vlhy{\sdash B_1,B_2,\Delta}
				}
			}
		}
		\quad\peq\quad
		\vlderivation{
			\vliin{\lprec}{}{\sdash A_1\lprec B_1,A_2\lprec B_2, \lFa xC , \Gamma, \Delta}{
				\vlin{\forall}{}{
					\sdash A_1,A_2,\lFa xC ,\Gamma
				}{
					\vlhy{\sdash A_1,A_2,C  ,\Gamma}
				}
			}{
				\vlhy{\sdash B_1,B_2,\Delta}
			}
		}
		$$

		\item Case $\naloadr/\lprec$ with $\nabla\in\set{\lnewsymb,\lyasymb}$:
		similarly to the case $\forall/\lprec$, but considering the rule $\naloadr$ and the nominal quantifier $\nabla$ instead of $\forall$.

		$$
		\begin{array}{ccc}
			\viA1 \vpz1{\lprec} \viB1,\viA2 \vpz2{\lprec} \viB2, \vpz3{\nabla}x.\vC1, \vpz4{\Gamma}, \vpz5{\Delta}
			\pzlinks{A1/A2/12/\la/red/{C1,pz4}}
			\pzlinks{B1/B2/-12/\lb/blue/{pz5}}
			&\to&
			\viA1 \vpz1{\lprec} \viB1,\viA2 \vpz2{\lprec} \viB2, \vpz3{\nabla}x.\vC1, \vpz4{\Gamma}, \vpz5{\Delta}
			\pzlinks{pz1/pz2/12/\labb/pzgreen/{C1,pz4,pz5}}
			\\\\
			\downarrow && \downarrow
			\\\\
			\viA1 \vpz1{\lprec} \viB1,\viA2 \vpz2{\lprec} \viB2, \vpz3{\nabla}x.\vC1, \vpz4{\Gamma}, \vpz5{\Delta}
			\pzlinks{A1/A2/12/\la/violet/{pz3,pz4}}
			\pzlinks{B1/B2/-12/\lb/blue/{pz5}}
			&\to&
			\viA1 \vpz1{\lprec} \viB1,\viA2 \vpz2{\lprec} \viB2, \vpz3{\nabla}x.\vC1, \vpz4{\Gamma}, \vpz5{\Delta}
			\pzlinks{pz1/pz2/12/\labb/brown/{pz3,pz4,pz5}}
		\end{array}
		$$
		With $\dualizerof[\labb] = \dualizerof[\la] \duasum \dualizerof[\lb]$.
		\nowDer
		$$
		\vlderivation{
			\vlin{\naloadr}{}{
				\sdash A_1\lprec B_1,A_2\lprec B_2,\lNa x C,\Gamma, \Delta
			}{
				\vliin{\lprec}{}{
					\sdash[x] A_1\lprec B_1,A_2\lprec B_2,C ,\Gamma, \Delta
				}{
					\vlhy{\sdash[x] A_1,A_2,C ,\Gamma}
				}{
					\vlhy{\sdash[x] B_1,B_2,\Delta}
				}
			}
		}
		\quad\peq\quad
		\vlderivation{
			\vliin{\lprec}{}{\sdash A_1\lprec B_1,A_2\lprec B_2, \lNa xC , \Gamma, \Delta}{
				\vlin{\naloadr}{}{
					\sdash A_1,A_2,\lNa xC ,\Gamma
				}{
					\vlhy{\sdash[x] A_1,A_2,C ,\Gamma}
				}
			}{
				\vlhy{\sdash B_1,B_2,\Delta}
			}
		}
		$$

		\item Case $\naur/\lprec$ with $\nabla\in\set{\lnewsymb,\lyasymb}$:
		similar to the previous case, but considering the rule $\naur$ instead of $\naloadr$.

		\item Case $\napopr/\lprec$ with $\nabla\in\set{\lnewsymb,\lyasymb}$:

		$$
		\begin{array}{ccc}
			\viA1 \vpz1{\lprec} \viB1,\viA2 \vpz2{\lprec} \viB2, \vpz3{\nabla} \vx1.\vC1,\vpz4{\cneg \nabla} \vy1.\vD1, \vpz5{\Gamma}, \vpz6{\Delta}
			\pzlinks{A1/A2/12/\la/red/{C1,D1,pz5}}
			\pzlinks{B1/B2/-12/\lb/blue/{pz6}}
			\pzlinks{x1/y1/15/\lc/brown/}
			&\to&
			\viA1 \vpz1{\lprec} \viB1,\viA2 \vpz2{\lprec} \viB2, \vpz3{\nabla}\vx1.\vC1,\vpz4{\cneg \nabla}\vy1.\vD1, \vpz5{\Gamma}, \vpz6{\Delta}
			\pzlinks{pz1/pz2/12/\labb/pzgreen/{C1,D1,pz5,pz6}}
			\pzlinks{x1/y1/15/\lc/brown/}
			\\\\
			\downarrow && \downarrow
			\\\\
			\viA1 \vpz1{\lprec} \viB1,\viA2 \vpz2{\lprec} \viB2, \vpz3{\nabla}x.\vC1,\vpz4{\cneg \nabla}y.\vD1, \vpz5{\Gamma}, \vpz6{\Delta}
			\pzlinks{A1/A2/12/\laa/violet/{C1,pz4,pz5}}
			\pzlinks{B1/B2/-12/\lb/blue/{pz6}}
			&\to&
			\viA1 \vpz1{\lprec} \viB1,\viA2 \vpz2{\lprec} \viB2, \vpz3{\nabla}x.\vC1,\vpz4{\cneg \nabla}y.\vD1, \vpz5{\Gamma}, \vpz6{\Delta}
			\pzlinks{pz1/pz2/12/\laabb/purple/{C1,pz4,pz5,pz6}}
		\end{array}
		$$

		With $\dualizerof[\laa] = \dualizerof[a] \fsubminus{y}$, $\dualizerof[\labb] = \dualizerof[\la] \duasum \dualizerof[\lb]$
		and $\dualizerof[\laabb] = \dualizerof[\labb]\fsubminus{y}$.
		\nowDer
		$$\adjustbox{max width=\textwidth}{$
		\vlderivation{
			\vlin{\napopr}{}{
				\sdash[\sS_1,\sS_2,\isna x] A_1\lprec B_1,A_2\lprec B_2, C, \lnNa yD,\Gamma, \Delta
			}{
				\vliin{\lprec}{}{
					\sdash[\sS_1,\sS_2] A_1\lprec B_1,A_2\lprec B_2, C , D\fsubst xy,\Gamma, \Delta
				}{
					\vlhy{\sdash[\sS_1] A_1,A_2, C, D\fsubst xy ,\Gamma}
				}{
					\vlhy{\sdash[\sS_1] B_1,B_2,\Delta}
				}
			}
		}
		\quad\peq\quad
		\vlderivation{
			\vliin{\lprec}{}{
				\sdash[\sS_1,\isna x] A_1\lprec B_1,A_2\lprec B_2, C,\lnNa yD,\Gamma, \Delta
			}{
				\vlin{\napopr}{}{
					\sdash[\sS_1, \isna x] A_1,A_2, C,\lnNa yD,\Gamma
				}{
					\vlhy{\sdash[\sS_1] A_1,A_2,C ,D\fsubst{x}{y} ,\Gamma}
				}
			}{
				\vlhy{\sdash[\sS_2] B_1,B_2,\Delta}
			}
		}
		$}$$

		\item
		Case $\exists / \naloadr$ with  $\nabla\in\set{\lnewsymb,\lyasymb}$:

		$$
		\begin{array}{ccc}
			\vpz1{\exists}x. \vA1, \vpz3{\nabla}y.\vB1, \vpz4{\Gamma}
			\pzlinks{A1/B1/12/\la/red/{pz4}}
			&\to&
			\vpz1{\exists}x. \vA1, \vpz3{\nabla}y.\vB1, \vpz4{\Gamma}
			\pzlinks{pz1/B1/12/\link[black]{a_1}/pzgreen/{pz4}}
			\\
			\downarrow && \downarrow
			\\\\
			\vpz1{\exists}x. \vA1,\quad \vpz3{\nabla}y.\vB1, \vpz4{\Gamma}
			\pzlinks{A1/pz3/12/\link[black]{a_2}/violet/{pz4}}
			&\to&
			\vpz1{\exists}x. \vA1,\quad \vpz3{\nabla}y.\vB1, \vpz4{\Gamma}
			\pzlinks{pz1/pz3/12/\link[black]{a_3}/purple/{pz4}}
		\end{array}
		$$

		With $\dualizerof[a_1] = \dualizerof[\la] \fsubminus{x}$, $\dualizerof[a_2] = \dualizerof[\la] $
		and $\dualizerof[a_3] = \dualizerof[a_1]$.
		\nowDer

		$$
		\vlderivation{
			\vlin{\naloadr}{}{
				\sdash \lEx xA, \lNa yB,\Gamma
			}{
				\vlin{\exists}{}{
					\sdash[\sS,\isna y] \lEx xA, B ,\Gamma
				}{
					\vlhy{\sdash[\sS,\isna y] A \fsubst{c}{x}, B ,\Gamma}
				}
			}
		}
		\peq
		\vlderivation{
			\vlin{\exists}{}{
				\sdash \lEx xA, \lNa yB,\Gamma
			}{
				\vlin{\naloadr}{}{
					\sdash A\fsubst{c}{x}, \lNa yB,\Gamma
				}{
					\vlhy{\sdash[\sS,\isna y] A\fsubst{c}{x}, B ,\Gamma}
				}
			}
		}
		$$

		\item
		Case $\exists / \naur$ with $\nabla\in\set{\lnewsymb,\lyasymb}$:
		similar to the previous case, but considering the rule $\naur$ instead of $\naloadr$.

		\item
		Case $\exists / \napopr$ with  $\nabla\in\set{\lnewsymb,\lyasymb}$:

		$$
		\begin{array}{ccc}
			\vpz1{\exists}z. \vA1, \vpz3{\nabla}\vx1.\vB1,\vpz4{\cneg\nabla}\vy1.\vC1, \vpz5{\Gamma}
			\pzlinks{A1/B1/12/\la/red/{C1,pz5}}
			\pzlinks{x1/y1/15/\lc/brown/}
			&\to&
			\vpz1{\exists}z. \vA1, \vpz3{\nabla}\vx1.\vB1,\vpz4{\cneg\nabla}\vy1.\vC1, \vpz5{\Gamma}
			\pzlinks{pz1/B1/12/\labb/pzgreen/{C1,pz5}}
			\pzlinks{x1/y1/15/\lc/brown/}
			\\
			\downarrow && \downarrow
			\\\\
			\vpz1{\exists}z. \vA1, \vpz3{\nabla}x.\vB1,\vpz4{\cneg\nabla}y.\vC1, \vpz5{\Gamma}
			\pzlinks{B1/pz4/12/\laaa/violet/{A1,pz5}}
			&\to&
			\vpz1{\exists}z. \vA1, \vpz3{\nabla}x.\vB1,\vpz4{\cneg\nabla}y.\vC1, \vpz5{\Gamma}
			\pzlinks{pz1/B1/12/\laaabb/purple/{pz4,pz5}}
		\end{array}
		$$

		With $\dualizerof[\labb] = \dualizerof[\la] \fsubminus{z}$,
		$\dualizerof[\laaa] = \dualizerof[\la] \fsubminus{y}$ and  $\dualizerof[\laaabb] = (\dualizerof[\la]\fsubminus{z} )\fsubminus{y}$
		\nowDer
		$$
		\vlderivation{
			\vlin{\napopr}{}{
				\sdash[\sS,\isna x] \lEx zA, B, \lnNa yC,\Gamma
			}{
				\vlin{\exists}{}{
					\sdash \lEx zA,B , C\fsubst{x}{y}, \Gamma
				}{
					\vlhy{\sdash A\fsubst{c}{z}, B ,C\fsubst{x}{y}, \Gamma}
				}
			}
		}
		\peq
		\vlderivation{
			\vlin{\exists}{}{
				\sdash[\sS,\isna x] \lEx zA, B, \lnNa yC,\Gamma
			}{
				\vlin{\napopr}{}{
					\sdash[\sS, \isna x] A \fsubst{c}{z}, B, \lnNa yC,\Gamma
				}{
					\vlhy{\sdash A \fsubst{c}{z}, B ,C \fsubst{x}{y}, \Gamma}
				}
			}
		}
		$$

		\item
		Case $\rrule_1/\rrule_2$ with $\rrule_1,\rrule_2\in\set{\forall, \naloadr,\naur}$ with $\nabla\in\set{\lnewsymb,\lyasymb}$:

		$$
		\begin{array}{ccc}
			\vpz1{\lqusymb_1}x. \vA1, \vpz3{\lqusymb_2}y.\vB1, \vpz4{\Gamma}
			\pzlinks{A1/B1/12/\la/red/{pz4}}
			&\to&
			\vpz1{\lqusymb_1}x. \vA1, \vpz3{\lqusymb_2}y.\vB1, \vpz4{\Gamma}
			\pzlinks{pz1/B1/12/\link[black]{a_1}/pzgreen/{pz4}}
			\\
			\downarrow && \downarrow
			\\\\
			\vpz1{\lqusymb_1}x. \vA1,\vpz3{\lqusymb_2}y.\vB1, \vpz4{\Gamma}
			\pzlinks{pz4/pz3/12/\link[black]{a_2}/violet/{A1}}
			&\to&
			\vpz1{\lqusymb_1}x. \vA1, \vpz3{\lqusymb_2}y.\vB1, \vpz4{\Gamma}
			\pzlinks{pz1/pz3/12/\link[black]{a_3}/brown/{pz4}}
		\end{array}
		$$

		With $\dualizerof[\la] = \dualizerof[a_i]$ for $i\in\set{1,2,3}$.
		\nowDer
		$$
		\vlderivation{
			\vlin{\rrule_1}{}{
				\sdash \lQu[1] xA, \lQu[2] yB, \Gamma
			}{
				\vlin{\rrule_2}{}{
					\sdash A , \lQu[2] yB, \Gamma
				}{
					\vlhy{\sdash A, B, \Gamma}
				}
			}
		}
		\quad\peq\quad
		\vlderivation{
			\vlin{\rrule_2}{}{
				\sdash \lQu[1] xA, \lQu[2] yB, \Gamma
			}{
				\vlin{\rrule_1}{}{
					\sdash \lQu[1] xA, B, \Gamma
				}{
					\vlhy{\sdash A, B, \Gamma}
				}
			}
		}
		\quad\mbox{if }
		\rrule_1,\rrule_2\in\set{\forall,\nuur,\yaur}
		$$

		$$
		\vlderivation{
			\vlin{\rrule_1}{}{
				\sdash \lNa xA, \lQu yB, \Gamma
			}{
				\vlin{\rrule_2}{}{
					\sdash[\sS,\isna x] A , \lQu yB, \Gamma
				}{
					\vlhy{\sdash[\sS,\isna x] A,B, \Gamma}
				}
			}
		}
		\quad\peq\quad
		\vlderivation{
			\vlin{\rrule_2}{}{
				\sdash \lNa xA, \lQu yB, \Gamma
			}{
				\vlin{\rrule_1}{}{
					\sdash[\sS,\isna x] \lNa xA, B, \Gamma
				}{
					\vlhy{\sdash[\sS,\isna x] A, B, \Gamma}
				}
			}
		}
		\quad
		\mbox{if }
		\rrule_1\in\set{\nuloadr,\yaloadr}
		\mbox{ and }
		\rrule_2\in\set{\forall,\nuur,\yaur}
		$$

		$$
		\vlderivation{
			\vlin{\rrule_1}{}{
				\sdash \lNai1 xA, \lNai2 yB, \Gamma
			}{
				\vlin{\rrule_2}{}{
					\sdash[\sS, {\isna[1]x}] A , \lNai2 yB, \Gamma
				}{
					\vlhy{\sdash[\sS,{\isna[1]x},{\isna[2]y}] A, B, \Gamma}
				}
			}
		}
		\peq
		\vlderivation{
			\vlin{\rrule_2}{}{
				\sdash \lNai1 xA, \lNai2 yB, \Gamma
			}{
				\vlin{\rrule_1}{}{
					\sdash[\sS,{\isna[2]y}] \lNai1 xA, B, \Gamma
				}{
					\vlhy{\sdash[\sS,{\isna[1]x},{\isna[2]y}] A, B, \Gamma}
				}
			}
		}
		\quad\mbox{if }
		\rrule_1,\rrule_2\in\set{\nuloadr,\yaloadr}
		$$

		\

		\item
		Case $\napopr/\rrule$ with $\rrule\in\set{\forall,\naloadr,\naur}$ and  $\nabla\in\set{\lnewsymb,\lyasymb}$:

		\

		$$
		\begin{array}{ccc}
			\vpz1{\lqusymb}z. \vA1, \vpz3{\nabla}\vx1.\vB1,\vpz4{\cneg\nabla}\vy1.\vC1, \vpz5{\Gamma}
			\pzlinks{A1/B1/16/\la/red/{C1,pz5}}
			\pzlinks{x1/y1/12/\lc/brown/}
			&\to&
			\vpz1{\lqusymb}z. \vA1, \vpz3{\nabla}\vx1.\vB1,\vpz4{\cneg\nabla}\vy1.\vC1, \vpz5{\Gamma}
			\pzlinks{pz1/B1/16/\labb/pzgreen/{C1,pz5}}
			\pzlinks{x1/y1/12/\lc/brown/}
			\\
			\downarrow && \downarrow
			\\\\
			\vpz1{\lqusymb}z. \vA1,\quad \vpz3{\nabla}x.\vB1,\vpz4{\cneg\nabla}y.\vC1, \vpz5{\Gamma}
			\pzlinks{A1/B1/16/\laa/violet/{pz4,pz5}}
			&\to&
			\vpz1{\lqusymb}z. \vA1,\quad \vpz3{\nabla}x.\vB1,\vpz4{\cneg\nabla}y.\vC1, \vpz5{\Gamma}
			\pzlinks{pz1/B1/16/\laabb/purple/{pz4,pz5}}
		\end{array}
		$$

		With $\dualizerof[\labb] = \dualizerof[\la]$,
		$\dualizerof[\laaa] = \dualizerof[\la] \fsubminus{y}$ and  $\dualizerof[\laaabb] = \dualizerof[\la]\fsubminus{y}$
		\nowDer
		$$
		\vlderivation{
			\vlin{\napopr}{}{
				\sdash[\sS,\isna x] \lQu zA, B,\lnNa yC, \Gamma
			}{
				\vlin{\rrule}{}{
					\sdash \lQu zA, B, C \fsubst{x}{y}, \Gamma
				}{
					\vlhy{\sdash A,B, C\fsubst{x}{y},\Gamma}
				}
			}
		}
		\quad\peq\quad
		\vlderivation{
			\vlin{\rrule}{}{
				\sdash[\sS,\isna x] \lQu zA, B,\lnNa yC, \Gamma
			}{
				\vlin{\napopr}{}{
					\sdash[\sS,\isna x] A, B,\lnNa yC, \Gamma
				}{
					\vlhy{\sdash A ,B , C\fsubst{x}{y},\Gamma}
				}
			}
		}
		\quad
		\mbox{if }
		\rrule\in\set{\forall,\nuur,\yaur}
		$$

		$$
		\vlderivation{
			\vlin{\nnapopr}{}{
				\sdash[\sS,\isnna x] \lQu zA, B,\lNa yC, \Gamma
			}{
				\vlin{\rrule}{}{
					\sdash \lQu zA, B, C \fsubst xy, \Gamma
				}{
					\vlhy{\sdash[\sS,\isqu z] A,B, C\fsubst xy,\Gamma}
				}
			}
		}
		\quad\peq\quad
		\vlderivation{
			\vlin{\rrule}{}{
				\sdash[\sS, \isnna x] \lQu zA, B,\lNa yC, \Gamma
			}{
				\vlin{\nnapopr}{}{
					\sdash[\sS,\isqu z, \isnna x] A, B,\lNa yC, \Gamma
				}{
					\vlhy{\sdash[\sS,\isqu z ] A ,B , C\fsubst xy,\Gamma}
				}
			}
		}
		\mbox{if }
		\rrule\in\set{\nuloadr,\yaloadr}
		$$

		\

		\item Case $\lwith/\exists$:
		$$
		\begin{array}{ccccc}
			&&
			\vA1 \vlwith1 \vB1,\quad \vpz1{\exists}x.\vC1 , \vpz2{\Gamma}
			\pzlinks{A1/pz1/12/\laaa/magenta/{pz2}}
			\pzlinks{B1/C1/-12/\lb/blue/{pz2}}
			&\to&
			\vA1 \vlwith1 \vB1, \quad \vpz1{\exists}x.\vC1 , \quad \vpz2{\Gamma}
			\pzlinks{A1/pz1/12/\laaa/magenta/{pz2}}
			\pzlinks{pz1/pz2/-12/\lbbb/violet/{B1}}
			\\
			&\nearrow &&&
			\\
			\vA1 \vlwith1 \vB1, \vpz1{\exists}x.\vC1 , \vpz2{\Gamma}
			\pzlinks{A1/C1/12/\la/red/{pz2}}
			\pzlinks{B1/C1/-12/\lb/blue/{pz2}}
			&&&&\downarrow
			\\
			&\searrow &&&
			\\
			&&
			\vA1 \vlwith1 \vB1, \vpz1{\exists}x.\vC1 , \vpz2{\Gamma}
			\pzlinks{lwith1/C1/12/\lcd/pzgreen/{pz2}}
			&\to&
			\vA1 \vlwith1 \vB1,\quad \; \vpz1{\exists}x.\vC1 , \vpz2{\Gamma}
			\pzlinks{lwith1/pz1/12/\lcccd/brown/{pz2}}
		\end{array}
		$$

		With $\dualizerof[\laa] = \dualizerof[\la] \fsubminus{y}$,
		$\dualizerof[\lbb] = \dualizerof[\lb] \fsubminus{y}$,
		 $\dualizerof[\lc] = \dualizerof[\la] \join \dualizerof[\lb]$ and
		 $\dualizerof[\lccd] = \dualizerof[\lc]\fsubminus{y}$.
		\nowDer
		$$
		\vlderivation{
			\vlin{\exists}{}{
				\sdash A\lwith B, \lEx xC, \Gamma
			}{
				\vliin{\lwith}{}{
					\sdash A\lwith B, C\fsubst{c}{x}, \Gamma
				}{
					\vlhy{\sdash A,C\fsubst{c}{x},\Gamma}
				}{
					\vlhy{\sdash B,C\fsubst{c}{x},\Gamma}
				}
			}
		}
		\quad\peq\quad
		\vlderivation{
			\vliin{\lwith}{}{
				\sdash A\lwith B, \lEx xC, \Gamma
			}{
				\vlin{\exists}{}{
					\sdash A, \lEx xC, \Gamma
				}{
					\vlhy{\sdash A,C\fsubst{c}{x},\Gamma}
				}
			}{
				\vlin{\exists}{}{
					\sdash  B, \lEx xC, \Gamma
				}{
					\vlhy{\sdash B,C\fsubst{c}{x},\Gamma}
				}
			}
		}
		$$

		\item
		Case $\lwith/\rrule$ with $\rrule\in\set{\forall,\naloadr,\naur}$ with $\nabla\in\set{\lnewsymb,\lyasymb}$:

		$$
		\begin{array}{ccccc}
			&&
			\vA1 \vlwith1 \vB1, \vpz1{\lqusymb}x.\vC1 , \vpz2{\Gamma}
			\pzlinks{A1/pz1/12/\laa/magenta/{pz2}}
			\pzlinks{B1/C1/-12/\lb/blue/{pz2}}
			&\to&
			\vA1 \vlwith1 \vB1, \quad \vpz1{\lqusymb}x.\vC1 , \vpz2{\Gamma}
			\pzlinks{A1/pz1/12/\laa/magenta/{pz2}}
			\pzlinks{B1/pz1/-12/\lbb/violet/{pz2}}
			\\
			&\nearrow &&&
			\\
			\vA1 \vlwith1 \vB1, \vpz1{\lqusymb}x.\vC1 , \vpz2{\Gamma}
			\pzlinks{A1/C1/12/\la/red/{pz2}}
			\pzlinks{B1/C1/-12/\lb/blue/{pz2}}
			&&&&\downarrow
			\\
			&\searrow &&&
			\\
			&&
			\vA1 \vlwith1 \vB1, \vpz1{\lqusymb}x.\vC1 , \vpz2{\Gamma}
			\pzlinks{lwith1/C1/12/\lcd/pzgreen/{pz2}}
			&\to&
			\vA1 \vlwith1 \vB1,\quad \vpz1{\lqusymb}x.\vC1 , \vpz2{\Gamma}
			\pzlinks{lwith1/pz1/12/\lccd/brown/{pz2}}
		\end{array}
		$$

		With $\dualizerof[\laa] = \dualizerof[\la] $,
		$\dualizerof[\lbb] = \dualizerof[\lb] $ and
		 $\dualizerof[\lc] = \dualizerof[\la] \join \dualizerof[\lb] =\dualizerof[\lccd]$ .
		\nowDer

		$$\adjustbox{max width=\textwidth}{$
		\vlderivation{
			\vlin{\rrule}{}{
				\sdash A\lwith B, \lQu xC, \Gamma
			}{
				\vliin{\lwith}{}{
					\sdash A\lwith B, C, \Gamma
				}{
					\vlhy{\sdash A, C,\Gamma}
				}{
					\vlhy{\sdash B, C,\Gamma}
				}
			}
		}
		\quad\peq\quad
		\vlderivation{
			\vliin{\lwith}{}{
				\sdash A\lwith B, \lQu xC, \Gamma
			}{
				\vlin{\rrule}{}{
					\sdash A, \lQu xC, \Gamma
				}{
					\vlhy{\sdash A,C,\Gamma}
				}
			}{
				\vlin{\rrule}{}{
					\sdash  B, \lQu xC, \Gamma
				}{
					\vlhy{\sdash B,C ,\Gamma}
				}
			}
		}
		\quad
		\mbox{if }
		\lqusymb \in\set{\forall,\nuur,\yaur}
		$}$$

		$$\adjustbox{max width=\textwidth}{$
		\vlderivation{
			\vlin{\rrule}{}{
				\sdash A\lwith B, \lQu xC, \Gamma
			}{
				\vliin{\lwith}{}{
					\sdash[\sS,\isqu x] A\lwith B, C, \Gamma
				}{
					\vlhy{\sdash[\sS,\isqu x] A, C,\Gamma}
				}{
					\vlhy{\sdash[\sS,\isqu x] B, C,\Gamma}
				}
			}
		}
		\quad\peq\quad
		\vlderivation{
			\vliin{\lwith}{}{
				\sdash A\lwith B, \lQu xC, \Gamma
			}{
				\vlin{\rrule}{}{
					\sdash A, \lQu xC, \Gamma
				}{
					\vlhy{\sdash[\sS,\isqu x] A,C ,\Gamma}
				}
			}{
				\vlin{\rrule}{}{
					\sdash  B, \lQu xC, \Gamma
				}{
					\vlhy{\sdash[\sS,\isqu x] B, C,\Gamma}
				}
			}
		}
		\quad
		\mbox{if }
		\rrule \in\set{\nuloadr,\yaloadr}
		$}$$

		\item
		Case $\rrule / \conf$ with $\rrule \in \set{\exists,\forall, \nuloadr,\yaloadr, \nuur,\yaur}$:

		$$
		\begin{array}{ccc}
			\conf (a_1 , \dots, a_l , b_1, \dots, b_k , c) & \to & \conf (\conf(a_1 , \dots, a_l ),\conf( b_1, \dots, b_k , c))
			\\
			\downarrow && \downarrow
			\\
			\conf (a_1 , \dots, a_l , b_1, \dots, b_k , c') & \to & \conf (\conf(a_1 , \dots, a_l ),\conf( b_1, \dots, b_k , c'))
		\end{array}
		$$

		Where either
		\begin{itemize}
			\item $\rrule\in\set{\forall,\nuur,\yaur,\nuloadr,\yaloadr}$, thus
			$c = \Gamma , A$, $c' = \Gamma, \lQu xA$, and $\dualizerof[c'] = \dualizerof[c]$; or
			\item $\rrule\in\set{\exists,\nupopr,\yapopr}$, thus
			$c = \Gamma , A$ ,  $c' = \Gamma, \lEx x A$ and $\dualizerof[c'] = \dualizerof[c] \fsubminus{x}$.
		\end{itemize}
		The derivation labelling the link in the bottom-right corner of the diagram is the same, independently of the sequence of coalescence steps.

		\item
		Case $\mixr / \mixr$ :

		if this case applies, then the \cotree has a $\conc$-node $x$ with leaves on which only $(\mixr)$ can be applied to merge some of them.
		In particular, $x$ must have at least three leaves $\la$, $\lb$ and $\lc$ such that, without loss of generaltity, a $(\mixr)$ can be appliet to merge $\la$ and $\lb$, or to merge $\lb$ and $\lc$.
		Let $\labb$ be the link obtained by applying a $(\mixr)$ step to merge $\la$ and $\lb$.
		To conclude it suffices to remark that we can always find a continuation of the coalescence path containing such step $(\mixr)$ which also contain  another step $(\mixr)$ merging $\lc$ and the link obtained by applying coalescence steps involving $\labb$.
		This follows from the fact that the side condition of the step $(\mixr)$ implies that, if exists,
		the least common ancestor (in the forest of $\Gamma$) of two formulas from two links which could be have been merged using a step $(\mixr)$  must be a $\lpar$.
		Thus, under the conditio that $\linktree$ is a proof net, once we cannot apply any more coalescence step to a link obtained from $\labb$, we are still able to apply a $(\mixr)$ step to merge it with $\lc$.

		\item
		Case $\precur/\precur$:
		similar to the previous case.

	\end{itemize}

	Finally, we have the following special additional cases, in which the dualizer is not the same, but the proof nets are.

	\begin{itemize}
		\item
		Case
		$\nupopr/\yapopr$:

		$$
		\begin{array}{ccc}
			\vpz1{\lnewsymb}\vx1. \vA1, \vpz3{\lyasymb}\vy2.\vB1, \vpz4{\Gamma}
			\pzlinks{A1/B1/12/\la/red/{pz4}}
			\pzlinks{x1/y2/16/\lb/blue/}
			&\to&
			\vpz1{\lnewsymb}\vx1. \vA1, \vpz3{\lyasymb}\vy2.\vB1, \vpz4{\Gamma}
			\pzlinks{pz1/B1/12/c/pzgreen/{pz4}}
			\\
			\downarrow && \downarrow
			\\\\
			\vpz1{\lnewsymb}\vx1. \vA1, \vpz3{\lyasymb}\vy2.\vB1, \vpz4{\Gamma}
			\pzlinks{A1/pz3/12/d/violet/{pz4}}
			&\to&
			\vpz1{\lnewsymb}\vx1. \vA1, \vpz3{\lyasymb}\vy2.\vB1, \vpz4{\Gamma}
			\pzlinks{pz1/pz3/12/e/brown/{pz4}}
		\end{array}
		$$
		where
		$\dualizerof[c]=\dualizerof[d]=\dualizerof[e]=\dualizerof[a]\setminus\set{y}$.
		\nowDer
		$$
			\dD_1
		\quad=\quad
		\vlderivation{
				\vlin{\nuloadr}{}{
					\sdash \Gamma, \lNu xA, \lYa yB
				}{
					\vlin{\nupopr}{}{
						\sdash[\sS,\isnu x] \Gamma, A, \lYa yB
					}{\vlhy{\sdash \Gamma, A, B\fsubst xy}}
				}
			}
		\quad\pweq\quad
			\vlderivation{
				\vlin{\yaloadr}{}{
					\sdash \Gamma, \lNu xA, \lYa yB
				}{
					\vlin{\yapopr}{}{
						\sdash[\sS,\isya y] \Gamma, \lNu xA, B
					}{\vlhy{\sdash \Gamma, A\fsubst yx, B}}
				}
			}
		\quad=\quad
			\dD_2
		$$

		Where $\dD_1\pweq \dD_2$ because
		$$
			\dD_1
		\;\weq\;
			\vlderivation{
				\vlin{\nuloadr}{}{
					\sdash \Gamma, \lNu xA, \lYa yB
				}{
					\vlin{\nupopr}{}{
						\sdash[\sS,\isnu z] \Gamma, A\fsubst zx, \lYa yB
					}{\vlhy{\sdash \Gamma, A\fsubst zx, B\fsubst zy}}
				}
			}
		\quad\peq\quad
			\vlderivation{
				\vlin{\yaloadr}{}{
					\sdash \Gamma, \lNu xA, \lYa xB
				}{
					\vlin{\yapopr}{}{
						\sdash[\sS,\isya z] \Gamma, \lNu xA, B\fsubst zx
					}{\vlhy{\sdash \Gamma, A\fsubst zx , B\fsubst zx}}
				}
			}
		\;\weq\;
			\dD_2
		$$

	\end{itemize}
\end{proof}

\begin{remark}\label{remark:notcritical}
	As already observed in \cite{hei:hug:conflict}, coalescence is not confluent on non-coalescent \cotree, as shown in the following examples.
	$$\adjustbox{max width=\textwidth}{$\begin{array}{c}
		\begin{array}{ccc}
			&&
			\viA1 \vpz1{\lprec} \viB1,\viA2 \vpz2{\lprec} \viB2, \viA3 \vpz3{\lprec} \viC1, \viA4 \vpz4{\lprec} \viC2, \vpz5{\Gamma},\vpz6{\Delta}
			\pzlinks{pz1/pz2/12/\labb/pzgreen/{A3,A4,pz5,pz6}}
			\pzlinks{C1/C2/-12/\lc/violet/{pz6}}
			\\
			&\nearrow &
			\\
			\viA1 \vpz1{\lprec} \viB1,\viA2 \vpz2{\lprec} \viB2, \viA3 \vpz3{\lprec} \viC1, \viA4 \vpz4{\lprec} \viC2, \vpz5{\Gamma},\vpz6{\Delta}
			\pzlinks{A1/A2/12/\la/red/{A3,pz5}}
			\pzlinks{B1/pz6/-12/\lb/blue/{B2}}
			\pzlinks{C1/C2/18/\lc/violet/{pz6}}
			&&
			\\
			&\searrow &
			\\
			&&
			\viA1 \vpz1{\lprec} \viB1,\viA2 \vpz2{\lprec} \viB2, \viA3 \vpz3{\lprec} \viC1, \viA4 \vpz4{\lprec} \viC2, \vpz5{\Gamma},\vpz6{\Delta}
			\pzlinks{A1/A2/12/\lac/violet/{pz3,pz4,pz5,pz6}}
			\pzlinks{B1/B2/-12/\lb/blue/{pz6}}
		\end{array}
	\\\\\\
		\begin{array}{ccc}
			&&
			\viA1 \vpz1{\lprec} \viB1,\viA2 \vpz2{\lprec} \viB2, \viA3 \vpz3{\ltens} \vC1, \vpz4{\Gamma}, \vpz5{\Delta}
			\pzlinks{pz1/pz2/12/\labb/pzgreen/{A3,pz4,pz5}}
			\pzlinks{C1/pz5/-12/\lc/violet/}
			\\
			&\nearrow &
			\\
			\viA1 \vpz1{\lprec} \viB1,\viA2 \vpz2{\lprec} \viB2, \viA3 \vpz3{\ltens} \vC1, \vpz4{\Gamma}, \vpz5{\Delta}
			\pzlinks{A1/A2/12/\la/red/{A3,pz4}}
			\pzlinks{B1/B2/-12/\lb/blue/{pz5}}
			\pzlinks{C1/pz5/18/\lc/violet/}
			&&
			\\
			&\searrow &
			\\
			&&
			\viA1 \vpz1{\lprec} \viB1,\viA2 \vpz2{\lprec} \viB2, \viA3 \vpz3{\ltens} \vC1, \vpz4{\Gamma}, \vpz5{\Delta}
			\pzlinks{A1/A2/12/\lac/violet/{pz3,pz4,pz5}}
			\pzlinks{B1/B2/-12/\lb/blue/{pz5}}
		\end{array}
	\end{array}$}$$
\end{remark}
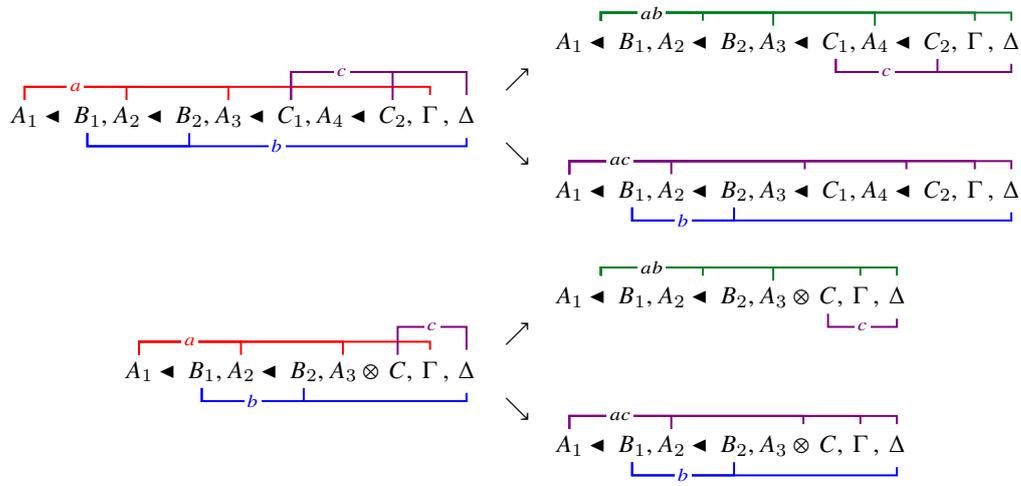
\end{document}